\documentclass[journal, final]{IEEEtran}

\newif \iftwocol
\twocolfalse

\usepackage{cite}      
\usepackage{graphicx}  
\usepackage{psfrag}    
\usepackage{overpic}
\usepackage{subfigure} 
\usepackage{url}       
\usepackage{stfloats}  
\usepackage{amsmath}   
\usepackage{amssymb}   
\usepackage{tabularx}
\usepackage{amsmath}
\usepackage{amsthm}
\usepackage{amssymb}
\usepackage{cite}
\usepackage{graphicx}
\usepackage{booktabs}
\usepackage{subfigure}

\newcommand{\R}{\mathcal{R}}

\newcommand{\INR}{\mathsf{INR}}
\newcommand{\SNR}{\mathsf{SNR}}

\newcommand{\hf}{\frac{1}{2}}

\newcommand{\q}{\mathsf{q}}
\newcommand{\C}{\mathsf{C}}

\newtheorem{theo}{Theorem}
\newtheorem{corollary}{Corollary}

\newtheorem{lemma}{Lemma}

\begin{document}

\title{Incremental Relaying for the Gaussian Interference Channel with
a Degraded Broadcasting Relay
\thanks{
Manuscript received November 7, 2011; revised August 9, 2012 and Nov 24, 2012; accepted Dec 7, 2012. Date of current version December 12, 2012. This work was supported by the Natural Science and Engineering Research Council (NSERC) of Canada. The material in this paper has been
presented in part at Allerton Conf. Commun., Control and Computing,
Sept. 2011.}  \thanks{The authors are with  The Edward S. Rogers Sr. Department of Electrical and Computer Engineering, University of Toronto, Toronto, ON M5S 3G4 Canada (email:
zhoulei@comm.utoronto.ca; weiyu@comm.utoronto.ca).
}}

\author{ Lei Zhou, {\it Student Member, IEEE} and
	Wei Yu, {\it Senior Member, IEEE}}

\markboth{To Appear in IEEE Transactions on Information Theory}
{Incremental Relaying for the Gaussian Interference Channel with
a Degraded Broadcasting Relay}

\maketitle

\begin{abstract}

This paper studies incremental relay strategies for a two-user Gaussian relay-interference channel with an in-band-reception and
out-of-band-transmission relay, where the link between the relay and
the two receivers is modelled as a degraded broadcast channel. It is
shown that generalized hash-and-forward (GHF) can achieve the capacity
region of this channel to within a constant number of bits in a
certain weak-relay regime, where the transmitter-to-relay link gains
are not unboundedly stronger than the interference links between the
transmitters and the receivers.  The GHF relaying strategy is ideally
suited for the broadcasting relay because it can be implemented in an
incremental fashion, i.e., the relay message to one receiver is a
degraded version of the message to the other receiver.  A
generalized-degree-of-freedom (GDoF) analysis in the high
signal-to-noise ratio (SNR) regime reveals that in the symmetric
channel setting, each common relay bit can improve the sum rate
roughly by either one bit or two bits asymptotically depending on the
operating regime, and the rate gain can be interpreted as coming
solely from the improvement of the common messages rate, or
alternatively in the very weak interference regime as solely coming
from the rate improvement of the private messages. Further,
this paper studies an asymmetric case in which the relay has only a
single single link to one of the destinations. It is shown that with
only one relay-destination link, the approximate capacity region can
be established for a larger regime of channel parameters. Further,
from a GDoF point of view, the sum-capacity gain due to the relay
can now be thought as coming from either signal relaying only, or
interference forwarding only.

\end{abstract}

\section{Introduction}

Interference is a key limiting factor in modern communication systems.
In a wireless cellular network, the performance of cell-edge users is
severely limited by intercell interference. This paper considers the
use of relays in cellular networks. The uses of relays to combat channel
shadowing and to extend coverage for wireless systems have been widely
studied in the literature. The main goal of this paper is to
demonstrate the benefit of relaying for interference mitigation in
the interference-limited regime.

Consider a two-cell wireless network with two base-stations each
serving their respective receivers while interfering with each other,
as shown in Fig.~\ref{two_cell_network}.  The deployment of a {\em
cell-edge} relay, which observes a linear combination of the two
transmit signals from the base-stations and is capable of
independently communicating with the receivers over a pair of relay
links, can significantly help the receivers mitigate intercell
interference. This model is often referred to as an in-band-reception
and out-of-band-transmission relay-interference channel, as the
relay-to-receiver transmission can be thought of as taking place on a
different frequency band.

A particular feature of the channel model considered in this paper is
that the relay-to-receivers link is modeled as a Gaussian broadcast
channel. This is motivated by the fact that the relay's transmission
to the remote receivers often 
takes place in a wireless medium. Consequently, the same relay message
can be heard by both receivers and can potentially help both receivers
at the same time. Further, it is convenient (and without loss of generality as shown
later for the achievability scheme and the converse proved in this
paper) to model the relay-to-receiver
links as digital links with capacities $\C_1$ and $\C_2$ respectively,
but where one relay message is required to be a degraded version of
the other relay message, as in a Gaussian broadcast channel.
The goal of this paper is to devise an incremental relaying strategy
and to quantify its benefit for this particular relay-interference
channel.

\subsection{Related Work}

The classic two-user interference channel consists of two
transmitter-receiver pairs communicating in the presence of interference
from each other. Although the capacity region of the two-user Gaussian
interference channel is still not known exactly, it can be
approximated to within one bit \cite{Tse2007} using a
Han-Kobayashi power splitting strategy \cite{HK1981}.

The use of cooperative communication for interference mitigation has
received much attention recently. For example, \cite{zhou_z_relay, zhou_isita, zhou_z_relay_typeII}
studied the Gaussian Z-interference channel with
a unidirectional receiver cooperation link, and
\cite{Wang_ReceiverCooperation, Wang_TransmitterCooperation,
Promond_DestinationCooperation, Promond_SourceCooperation} studied the
Gaussian interference channel with bi-directional transmitter/receiver
cooperation links.
In addition, the Gaussian interference channel with an additional
relay node has also been studied extensively in the literature.
Depending on the types of the links
between the relay and the transmitters/receivers, the relay-interference
channel can be categorized as having in-band transmission/reception
\cite{Sahin_inband_relay_ic,
Tian_inband_relay_ic, Dabora_Maric, Maric_Dabora, Chaaban_asilomar10, gdof_relayic, Lou_ICC10, Sahin_Erkip_1}, out-of-band transmission/reception
\cite{Sahin_outofband_relay_ic, Tian_outofband_relay_ic, Simeone_outofband_tit}, out-of-band
transmission and in-band reception \cite{Sahin_outinband_relay_ic, Peyman_GHF, peyman_ghf_jnl, zhou_incremental_relay}, or
in-band transmission and out-of-band reception \cite{Sahin_ITA09},
the last of which is directly related to the channel model in this paper.
In the following, we review different
transmission schemes and relaying strategies that have emerged for
each of these cases.

For interference channels equipped with an in-band transmission and
reception relay, the relay interacts with both transmitters and
receivers in the same frequency band.  Relaying strategies that have
been investigated in the literature include decode-and-forward,
compress-and-forward, and amplify-and-forward.  For example,
\cite{Dabora_Maric, Maric_Dabora} show that decoding-and-forwarding
either the intended signal or the
interfering signal to a receiver can both be beneficial. The former is termed as signal relaying, the latter interference forwarding. Decode-and-forward and half-duplex
amplify-and-forward strategies are also studied in
\cite{Chaaban_asilomar10, Lou_ICC10}.  When combining
decode-and-forward relaying strategy and the Han-Kobayashi rate
splitting input scheme, \cite{Sahin_Erkip_1} gives an achievable rate
region that has a shape similar to the Chong-Motani-Garg (CMG) region for the interference channel \cite{Chong2006}.
The exact capacity for this type of relay-interference channel is in general open, but there
is a special potent-relay case \cite{Tian_inband_relay_ic} for which
the sum capacity is known in some specific regimes.

The difficulty in establishing the capacity of the interference
channel with in-band transmission/reception relay is in part due to
the fact that the relay's received and transmit signals intertwine
with that of the underlying interference channel.
To simplify the matter, the interference channel with an out-of-band
transmission/reception relay has been studied in
\cite{Sahin_outofband_relay_ic, Simeone_outofband_tit,
Tian_outofband_relay_ic}. In this channel model, the relay essentially
operates on a separate set of parallel channels. Based on signal relaying and
interference forwarding strategies, \cite{Sahin_outofband_relay_ic}
identifies the condition under which the capacity
region can be achieved with separable or nonseparable coding between
the out-of-band relay and the underlying interference channel.
Further, \cite{Tian_outofband_relay_ic} studies this channel model in a
symmetric setting and characterizes the sum capacity to within $1.15$
bits. The transmission scheme of  \cite{Tian_outofband_relay_ic}
involves further splitting of
the common messages in the Han-Kobayashi scheme and a relay strategy
that combines nested lattice coding and Gaussian codes.  It is shown
that in the strong interference regime, the use of structured codes is
optimal.

\begin{figure} [t]
\vspace{1.8em}
\centering
\begin{overpic}[width=3.2in]{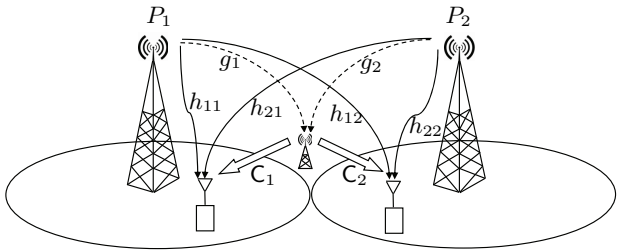}
\put(30, 23){\small $h_{11}$}
\put(40, 22){\small $h_{21}$}
\put(53, 21){\small $h_{12}$}
\put(66, 19){\small $h_{22}$}
\put(35, 30){\small $g_1$}
\put(58, 29.5){\small $g_2$}
\put(40, 11){\small $\mathsf{C}_{1}$}
\put(55, 11){\small $\mathsf{C}_{2}$}
\put(23, 37){\small $P_1$}
\put(72, 37){\small $P_2$}
\end{overpic}
\caption{A two-cell network with an in-band reception and
out-of-band-broadcasting relay for interference mitigation}
\label{two_cell_network}
\end{figure}

Another variation of the relay-interference channel involves an out-of-band reception and in-band transmission relay. This channel is studied in \cite{Sahin_outinband_relay_ic}, in which the transmitter
further splits the transmit signal according to the Han-Kobayashi scheme; the relay decodes
only part of the message depending the capacity of the
transmitter-relay links; the rest of the codewords are transmitted
directly from the sources to the destinations without the help of the
relay. With this partial decode-and-forward relaying scheme, the sum
capacity is found under a so-called strong relay-interference
condition.

The interference channel with an in-band reception/out-of-band
transmission relay has been briefly discussed in \cite{Sahin_ITA09},
and studied in \cite{Peyman_GHF, peyman_ghf_jnl} for a case where the
relay-destination links are shared between the two receivers.
Conventional decode-and-forward and compress-and-forward relay
strategies are not well matched for helping both
receivers simultaneously with a common relayed message. Thus,
\cite{Peyman_GHF, peyman_ghf_jnl} consider a generalized
hash-and-forward (GHF) strategy, which generalizes the conventional
compress-and-forward scheme, and is shown to achieve the capacity
region of this channel model to within a constant number of bits in the regime where the shared relay-destination link rate is
sufficiently small. The channel model under consideration in this
paper further extends the shared relay-destinations link to be a degraded
broadcast channel. We focus on a different weak-relay regime. The main objective is similar: to efficiently
use the relay bits to simultaneously benefit both users and to achieve capacity to within a constant gap.


Finally, the GHF relay strategy used in this paper is essentially the
same as the noisy network coding \cite{Kim_noisy_network_coding, Skoglund_1, Skoglund_2} and the
quantize-map-and-forward relay strategies \cite{Avestimehr_relaynetwork}. The result of this paper can be thought of as an effort in generalizing these relay strategies to a particular case of
the multiple unicast setting, for which constant-gap result continues to
hold for certain channel-parameter regimes. Related
works for the multiple unicast problem include \cite{Mohajer_tit_ic_relay, Gou_222, Avestimehr_twounicast}.

\subsection{Main Contributions}

This paper considers a relay-interference channel with in-band
reception and out-of-band degraded broadcasting links from the relay to the receivers.  The key features of the transmission strategy and the main
results of the paper are as follows.


\subsubsection{Incremental Relaying}

This paper uses a GHF relaying strategy
to take advantage of the in-band reception link and the out-of-band
broadcasting link from the relay to the receivers. In GHF, the relay
quantizes its observation, which is a linear combination of the
transmitted signals, using a {\em fixed} quantizer, then bins and
forwards the quantized observation to the receivers. This strategy
of fixing the quantization level is near optimal when a certain {\em weak-relay}
condition is satisfied, and is ideally matched to the degraded
broadcasting relay-to-receiver links with capacities $\C_1$ and
$\C_2$, because it allows an incremental binning strategy at the
relay.  Assuming that $\C_1 \le \C_2$, the relay may first bin its
quantized observation into $2^{n\C_1}$ bins and send the bin index to
both receivers, then further divide each bin into $2^{n(\C_2-\C_1)}$
sub-bins and sends the extra bin index to receiver $2$ only. Thus, the
relay message to the first receiver is a degraded version of the
message to the second receiver.

\subsubsection{Oblivious Power Splitting}

The transmission scheme used in this paper consists of a Han-Kobayashi
power splitting strategy \cite{HK1981} at the transmitter.
The common-private power splitting ratio in such a strategy is
crucial.  In a study of the interference channel with conferencing
links \cite{Wang_ReceiverCooperation}, Wang and Tse used the power
splitting strategy of Etkin, Tse and Wang \cite{Tse2007} where the
private power is set at the noise level at the receivers.  This is
sensible for the conferencing-receiver model considered in
\cite{Wang_ReceiverCooperation}, but not necessarily so for the
interference channel with an independent relay, unless again a certain
{\em weak-relay} condition is satisfied. This strategy of fixing the
power splitting at the transmitter to be independent of the relay is
termed {\em oblivious power splitting} in \cite{peyman_ghf_jnl}.
Oblivious power splitting is used in this paper as well.

\subsubsection{Constant Gap to Capacity in the weak-relay Regime}

The main result of this paper is that when the relay links are not
unboundedly stronger than the interfering links, i.e.,
\begin{equation}
\label{weak_relay_regime}
\max \left\{\frac{|g_1|^2}{|h_{12}|^2}, \frac{|g_2|^2}{|h_{21}|^2} \right\}
= \rho < \infty,
\end{equation}
for some fixed $\rho$, the capacity of the relay-interference channel
with a broadcast link can 
be achieved to within a constant gap, where the gap is a function of
$\rho$ but otherwise independent of channel parameters. This operating
regime is called the {\em weak-relay regime} in this paper.

The main result of this paper is motivated by
the results in \cite{Peyman_GHF} and \cite{peyman_ghf_jnl}, which
studies a
two-user interference channel augmented with a shared digital
relay link to the receivers of rate $R_0$, 
and obtains
a constant-gap-to-capacity result under a certain
small-$R_0$ condition using GHF and oblivious power splitting.
The relay strategy studied in this paper goes one step further in that
the relay-to-receivers link is modeled as a degraded broadcast channel.
Moreover, the weak-relay regime studied in this paper is a
counterpart of the small-$R_0$ regime studied in
\cite{peyman_ghf_jnl}, as can be visualized in the practical setup of
Fig.~\ref{two_cell_network}.  When the mobiles are close to their
respective cell centers, the relay link capacities $\C_1$ and $\C_2$
are small, thereby satisfying the small-$R_0$ condition of
\cite{peyman_ghf_jnl}. In the more practically important regime where
the mobile terminals are close to the cell edge, the channel falls
into the weak-relay regime of this paper. An interesting feature of
the result in this paper is that the gap to capacity is a function of
$\rho$, the relative channel strength between the interfering channel
and the channel to the relay;  the gap becomes smaller as $\rho
\rightarrow 1$.  In the limiting case with $\rho=1$, corresponding to
the situation where the mobiles are at the cell edge, the capacity
region can be achieved to within $\hf \log \frac{5+\sqrt{33}}{2}=1.2128$ bits.

A technical contribution of this paper is a particular set of capacity region
outer bounds which are established by giving different combinations of
side information (genies) to the receivers and by applying the known
outer-bound results of the Gaussian interference channel
\cite{Tse2007} and the single-input multiple-output (SIMO) Gaussian
interference channel \cite{Wang_Allerton}.  It is shown that there are
two constraints for the individual rates $R_1$ and $R_2$, twelve
constraints for the sum rate $R_1 +R_2$, six constraints for $2R_1
+R_2$, and six constraints for $R_1 +2R_2$. Furthermore, the outer bounds established in this paper hold for all channel
parameters. This set of outer bounds is tight to within a constant
gap in the weak-relay regime. 

To obtain insights from
the performance gain brought by the relay, this paper
further investigates the improvement in the generalized degrees of
freedom (GDoF) per user for the relay-interference channel due to
a broadcasting link. In the symmetric setting, it is shown that a
common broadcast link can improve the sum capacity by two bits per each
relay bit in the very weak, moderately weak, and very strong
interference regimes, but by one bit per each relay bit in other
regimes. This asymptotic behavior can be interpreted by noting that
the relay link essentially behaves like a deterministic channel in
the high signal-to-noise-ratio (SNR) regime.  Further, in the
symmetric setting, the sum-capacity gain due to the relay can be
thought of as solely coming from the rate improvement of the common
messages, or alternatively in a very weak interference regime as
solely coming from the rate improvement of the private messages.

In asymmetric settings, the improvement in the sum capacity by the
relay can be interpreted in different ways. To illustrate this point,
this paper investigates a special case of the channel model, where the
relay link is available
to only one but not both destinations. In this case, the relay may
forward information about both the intended signal and the
interference, and the capacity can benefit from both signal-relaying
and interference-forwarding.  This paper shows that a
constant-gap-to-capacity result can be derived for this setting under
a more relaxed weak-relay condition that requires only
$|g_2| \le \sqrt{\rho}|h_{21}|$ (and not
$|g_1| \le \sqrt{\rho}|h_{12}|$). Moreover,
this paper shows that in term of GDoF, when the relay link is
above a certain threshold, the sum-capacity gain is equivalent to that of
that of a single relay link from user $1$. When the relay link is below the threshold, the
sum-capacity gain is equivalent to that of a single relay link from user $2$.

Finally, the results of this paper show that GHF is sufficient for achieving the approximated capacity region of an in-band reception and out-of-band transmission Gaussian relay-interference channel in the weak-relay regime. Thus, more recently proposed relay
techniques based on compute-and-forward \cite{Nazer_ComputeForward} or lattice coding \cite{Zamir_LatticeCoding} is not
necessary in this regime as far as constant gap to capacity is concerned . Outside of the weak-relay regime, the optimal relay
strategies remain an open problem; lattice coding strategies may be helpful.

\subsection{Organization of This Paper}

The rest of the paper is organized as follows. Section II introduces
the Gaussian relay-interference channel model, derives capacity
region outer bounds that hold for all channel parameters and an
achievable rate region, and presents the main constant-gap theorem
and the GDoF analysis.  Section III deals with the
relay-interference channel with a single relay link, derives the
corresponding constant-gap result, and gives a quantitative
analysis on the relation between signal relaying and interference
forwarding. Section IV concludes the paper.


\section{Gaussian Relay-Interference Channel: General Case}

\begin{figure} [t]
\centering
\centering \psfrag{X1}{$X_1$} \psfrag{X2}{$X_2$} \psfrag{Y1}{$Y_1$}
\psfrag{Y2}{$Y_2$}
\psfrag{Z1}{$Z_1$} \psfrag{Z2}{$Z_2$} \psfrag{ZR}{$Z_R$}
 \psfrag{h11}{$h_{11}$} \psfrag{h22}{$h_{22}$}
 \psfrag{h21}{$h_{21}$}  \psfrag{h12}{$h_{12}$} \psfrag{G1}{$g_1$} \psfrag{G2}{$g_2$}  \psfrag{Relay}{$\mathrm{relay}$} \psfrag{C1}{$\mathsf{C}_1$} \psfrag{C2}{$\mathsf{C}_2$} \psfrag{YR}{$Y_R$}
\includegraphics[width=3.1in]{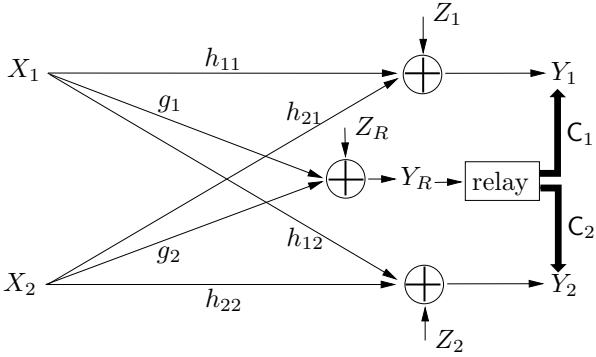}
\caption{Gaussian relay-interference channel with two independent digital relay links}
\label{channel_model}
\end{figure}
\subsection{Channel Model and Definitions}
A Gaussian relay-interference channel consists of two
transmitter-receiver pairs and an independent relay. Each transmitter
communicates with the intended receiver while causing interference to
the other transmitter-receiver pair. The relay receives a linear
combination of the two transmit signals and helps the
transmitter-receiver pairs by forwarding a message to receiver $1$ and
another message to receiver $2$ through rate-limited digital links
with capacities $\C_1$ and $\C_2$ respectively. We start by treating
a channel model with independent relay links, and later show that
requiring one relay message to be a degraded version of the other is
without loss of approximate optimality. As shown in
Fig.~\ref{channel_model}, $X_1, X_2$ and $Y_1, Y_2$ are real-valued input and
output signals, respectively, and $Y_R$ is the observation of the relay.
The receiver noises are assumed to be independent and identically
distributed (i.i.d.) Gaussian random variables with variance one,
i.e., $Z_i \sim \mathcal{N}(0, 1), i=1, 2$ and $R$. The input-output
relationship can be described by
\begin{align}
Y_1 &= h_{11}X_1 + h_{21}X_2  + Z_1, \\
Y_2 &= h_{22}X_2 + h_{12}X_1  + Z_2, \\
Y_R &= g_1 X_1 + g_2 X_2 + Z_R,
\end{align}
where $h_{ij}$ is the channel gain from transmitter $i$ to receiver $j$, and
$g_j$ is the channel gain from transmitter $j$ to the relay, all real valued. The powers of the
input signals are normalized to one, i.e., $\mathbb{E}[|X_i|^2] \le 1, i=1, 2$.

Define the signal-to-noise ratios and interference-to-noise ratios as follows:
\begin{eqnarray}
 \SNR_i = |h_{ii}|^2, \;&& \quad \SNR_{ri} = |g_i|^2, \;\; i=1, 2 \nonumber \\
\INR_1 = |h_{12}|^2, && \quad \INR_2 \;\;= |h_{21}|^2. \nonumber
\end{eqnarray}
Define functions $\alpha(\cdot)$ and $\beta(\cdot)$ as
\begin{equation} \label{alpha_beta}
\alpha(x) = \hf \log(2x+2 + \rho), \quad \beta(x) = \hf + \hf
\log\left(1 + \frac{1+\rho}{x}\right),
\end{equation}
where $\log(\cdot)$ is base $2$ and  $\rho$ is defined as
\begin{equation} \label{definition_rho}
\rho \triangleq \max \left\{\frac{|g_1|^2}{|h_{12}|^2},
\frac{|g_2|^2}{|h_{21}|^2} \right\}.
\end{equation}
This paper considers a weak-relay regime where
$\rho$ is a {\em finite} constant.

\subsection{Outer Bounds and Achievable Rate Region}

We first present outer bounds and achievability results that are
applicable to the relay-interference channel model with two
independent digital relays as shown in Fig.~\ref{channel_model}.

\begin{theo}[Capacity Region Outer Bounds] \label{outerbound_theorem_twolinks}
The capacity region of the Gaussian relay-interference channel as
depicted in Fig.~\ref{channel_model} is contained in the outer bound
$\overline{\mathcal{C}}$ given by the set of $(R_1,R_2)$ for which
\begin{eqnarray}
R_1 &\le& \frac{1}{2}\log(1 + \mathsf{SNR}_1) \nonumber \\
	& & + \min \left\{\C_1, \frac{1}{2} \log \left (1 + \frac{\mathsf{SNR}_{r1}}{1 + \mathsf{SNR}_1}  \right) \right\} \label{R1_bound_twolinks}\\
R_2 &\le& \frac{1}{2}\log(1 + \mathsf{SNR}_2) \nonumber \\
	& & + \min \left\{\C_2, \frac{1}{2} \log \left (1 +
\frac{\mathsf{SNR}_{r2}}{1 + \mathsf{SNR}_2}  \right) \right\}
\label{R2_bound_twolinks} \\
R_1 + R_2 &\le& \frac{1}{2}\log(1 + \mathsf{SNR}_2 + \mathsf{INR}_1)
\nonumber \\
	& & + \frac{1}{2}\log \left(1+  \frac{ \mathsf{SNR}_1}{1 + \mathsf{INR}_1}\right) + \C_1 + \C_2 \label{sumrate_bound_1_twolinks}\\
R_1 + R_2 &\le& \frac{1}{2}\log(1 + \mathsf{SNR}_1 + \mathsf{INR}_2)
\nonumber \\
	& & + \frac{1}{2}\log \left(1+ \frac{\mathsf{SNR}_2}{1 + \mathsf{INR}_2}\right) + \C_1 + \C_2 \label{sumrate_bound_2_twolinks}\\
R_1 + R_2 &\le& \frac{1}{2}\log \left(1 + \mathsf{INR}_2 +
\frac{\mathsf{SNR}_1}{1 + \mathsf{INR}_1} \right) \nonumber \\
&&	+\frac{1}{2}\log \left(1 + \mathsf{INR}_1 + \frac{\mathsf{SNR}_2}{1 + \mathsf{INR}_2} \right )  + \C_1 + \C_2 \nonumber \\
&& \label{sumrate_bound_3_twolinks}  \\
R_1 + R_2 &\le& \frac{1}{2}\log \left( 1 + \frac{\mathsf{SNR}_1}{1 + \mathsf{INR}_1 +  \mathsf{SNR}_{r1}} \right) \nonumber  \\
 && + \hf \log(1 + \SNR_2(1 + \phi_2^2 \SNR_{r1}) \nonumber \\
&& \qquad \quad +	\SNR_{r2} + \INR_1 + \SNR_{r1} )+\C_1
\label{sumrate_bound_4_twolinks} \\
R_1 + R_2 &\le& \frac{1}{2}\log(1 + \mathsf{SNR}_1 + \mathsf{INR}_2) \nonumber \\
&& + \frac{1}{2}\log \left(1 + \frac{\mathsf{SNR}_2 + \mathsf{SNR}_{r2}}{1 + \mathsf{INR}_2} \right) + \C_1 \label{sumrate_bound_5_twolinks}\\
R_1 + R_2 &\le& \hf \log\left(1 + \frac{\SNR_1}{1 + \INR_1 + \SNR_{r1}} +\INR_2 \right) \nonumber \\
&& + \hf \log \left( 1 + \frac{\SNR_2(1 + \phi_2^2 \SNR_{r1}) +
\SNR_{r2}}{1 + \INR_2} \right. \nonumber \\
&& \qquad \quad \left. + \INR_1 + \SNR_{r1} \right)  +\C_1  \label{sumrate_bound_6_twolinks}\\
R_1 + R_2 &\le& \frac{1}{2}\log \left( 1 + \frac{\mathsf{SNR}_2}{1 + \mathsf{INR}_2+  \mathsf{SNR}_{r2} } \right)  \nonumber \\
&& + \hf \log(1 + \SNR_1(1 + \phi_1^2 \SNR_{r2}) + \SNR_{r1} \nonumber \\
&& \qquad \quad  + \INR_2 + \SNR_{r2} )+\C_2 \label{sumrate_bound_7_twolinks} \\
R_1 + R_2 &\le& \frac{1}{2}\log(1 + \mathsf{SNR}_2 + \mathsf{INR}_1)
\nonumber \\
&&  + \frac{1}{2}\log \left(1+ \frac{ \mathsf{SNR}_1 + \mathsf{SNR}_{r1}}{1 + \mathsf{INR}_1} \right) + \C_2 \label{sumrate_bound_8_twolinks}\\
R_1 + R_2 &\le& \hf \log\left(1 + \frac{\SNR_2}{1 + \INR_2 + \SNR_{r2}} +\INR_1 \right) \nonumber \\
&& + \hf \log \left( 1 + \frac{\SNR_1(1 + \phi_1^2 \SNR_{r2}) +
\SNR_{r1}}{1 + \INR_1} \right. \nonumber \\
&& \qquad \quad \left. + \INR_2 + \SNR_{r2} \right) \label{sumrate_bound_9_twolinks} +\C_2
\end{eqnarray}
\begin{eqnarray}
R_1 + R_2 &\le& \hf \log \left(1 + \frac{\SNR_1 + \SNR_{r1}}{1 + \INR_1 + \SNR_{r1}} \right) \nonumber \\
&& + \hf \log(1 + \SNR_2(1 + \phi_2^2 \SNR_{r1}) + \SNR_{r2} \nonumber \\
&& \qquad \quad + \INR_1 + \SNR_{r1} ) \label{sumrate_bound_10_twolinks} \\
R_1 + R_2 &\le& \hf \log \left(1 + \frac{\SNR_2 + \SNR_{r2}}{1 + \INR_2 + \SNR_{r2}} \right) \nonumber \\
&& + \hf \log(1 + \SNR_1(1 + \phi_1^2 \SNR_{r2}) + \SNR_{r1} \nonumber \\
&& \qquad \quad + \INR_2 + \SNR_{r2} ) \label{sumrate_bound_11_twolinks} \\
R_1 + R_2 &\le& \hf \log \left( 1 + \frac{\SNR_1(1 + \phi_1^2
\SNR_{r2}) + \SNR_{r1}}{1 + \INR_1+\SNR_{r1}} \right. \nonumber \\
&& \qquad \quad  \left. + \INR_2 + \SNR_{r2} \right) \nonumber \\
&&+ \hf \log \left( 1 + \frac{\SNR_2(1 + \phi_2^2 \SNR_{r1}) +
\SNR_{r2}}{1 + \INR_2 +\SNR_{r2}} \right. \nonumber \\
&& \qquad \quad  \left. + \INR_1 + \SNR_{r1} \right) \label{sumrate_bound_12_twolinks} \\
2R_1 + R_2 &\le& \hf \log \left(1 + \SNR_1 + \INR_2 \right) \nonumber \\
&& + \hf \log\left(1 + \INR_1 + \frac{\SNR_2}{1 + \INR_2} \right) \nonumber \\
&&+ \hf \log \left(1+ \frac{ \SNR_1}{1 + \INR_1}\right) + 2\C_1 + \C_2 \label{2R1R2_bound_1_twolinks} \\
2R_1 + R_2 &\le& \hf \log \left(1 + \frac{\SNR_1 + \SNR_{r1}}{1 + \INR_1 + \SNR_{r1}} \right) \nonumber \\
&&+ \hf \log\left(1 + \SNR_1(1 + \phi_1^2 \SNR_{r2}) + \SNR_{r1}
\right. \nonumber \\
&& \qquad \quad  \left. + \INR_2 + \SNR_{r2} \right) \nonumber \\
&&+ \hf \log \left( 1 + \frac{\SNR_2(1 + \phi_2^2 \SNR_{r1}) +
\SNR_{r2}}{1 + \INR_2 + \SNR_{r2}} \right. \nonumber \\
&& \qquad \quad  \left. + \INR_1 + \SNR_{r1} \right)  \label{2R1R2_bound_2_twolinks} \\
2R_1 + R_2 &\le& \hf \log(1 + \SNR_1 + \INR_2) \nonumber \\
&& + \hf \log \left(1+ \frac{\SNR_1 + \SNR_{r1}}{1 + \INR_1} \right) \nonumber \\
&& + \hf \log \left( 1 + \frac{\SNR_2(1 + \phi_2^2 \SNR_{r1}) +
\SNR_{r2}}{1 + \INR_2} \right. \nonumber \\
&& \qquad \quad \left. + \INR_1 + \SNR_{r1} \right)  + \C_1 \label{2R1R2_bound_3_twolinks} \\
2R_1 + R_2 &\le& \hf \log \left(1 + \frac{\SNR_1 + \SNR_{r1}}{1 +
\INR_1} \right) \nonumber \\
&& + \hf \log \left(1 + \frac{\SNR_2}{1 + \INR_2 + \SNR_{r2} }+ \INR_1 \right) \nonumber \\
&&+ \hf \log \left(1 + \SNR_1 (1 + \phi_1^2 \SNR_{r2}) + \SNR_{r1}
\right. \nonumber \\
&& \qquad \quad \left. + \INR_2 + \SNR_2 \right) + \C_2 \label{2R1R2_bound_4_twolinks} \\
2R_1 + R_2 &\le& \hf \log (1 + \SNR_1 + \INR_2) \nonumber \\
&& + \hf \log\left(1 + \frac{\SNR_1}{1 + \INR_1 + \SNR_{r1}} \right) \nonumber \\
&& + \hf \log \left( 1 + \frac{\SNR_2(1 + \phi_2^2 \SNR_{r1}) +
\SNR_{r2}}{1 + \INR_2} \right. \nonumber \\
&& \qquad \quad \left. + \INR_1 + \SNR_{r1} \right) + 2\C_1  \label{2R1R2_bound_5_twolinks}
\end{eqnarray}
\begin{eqnarray}
2R_1 + R_2 &\le&  \hf \log (1 + \SNR_1 + \INR_2) \nonumber \\
&& + \hf \log \left(1 + \INR_1 + \frac{\SNR_2}{1 + \INR_2} \right) \nonumber \\
&&+ \hf \log \left(1 +  \frac{\SNR_1 + \SNR_{r1}}{1 + \INR_1} \right)
\nonumber \\
&& + \C_1 + \C_2 \label{2R1R2_bound_6_twolinks},
\end{eqnarray}
and $R_1 + 2R_2$ bounded by (\ref{2R1R2_bound_1_twolinks})-(\ref{2R1R2_bound_6_twolinks}) with indices $1$ and $2$ switched, where $\phi_1^2$ and $\phi_2^2$ are defined as
\begin{equation}
\phi_1^2 = \left|\frac{g_1 h_{21}}{g_2 h_{11}} - 1 \right|^2,  \quad \phi_2^2 = \left|\frac{g_2 h_{12}}{g_1 h_{22}} - 1 \right|^2.
\end{equation}
\end{theo}

\begin{IEEEproof}
The above outer bounds can be proved in a genie-aided approach. See Appendix~\ref{proof_outer_bound_twolinks} for details.
\end{IEEEproof}

\begin{theo}[Achievable Rate Region] \label{achievable_theorem_twolinks}
Let $\mathcal{P}$ denote the set of probability distributions $P(\cdot)$ that factor as
\begin{eqnarray}
\lefteqn{P(q, w_1, w_2, x_1, x_2, y_1, y_2, y_R, \hat{y}_{R1}, \hat{y}_{R2})} \nonumber \\
&& =p(q)p(x_1, w_1|q)p(x_2, w_2|q)p(y_1, y_2, y_R|x_1, x_2, q) \nonumber \\
&& \qquad p(\hat{y}_{R1},\hat{y}_{R2}|y_R, q).
\end{eqnarray}
For a fixed distribution $P \in \mathcal{P}$, let $\mathcal{R}(P)$ be the set of all rate pairs $(R_1, R_2)$ satisfying
\begin{eqnarray}
0 \le R_1 &\le& d_1 + \min\left\{(\C_1- \xi_1)^+, \Delta {d}_1  \right\}, \label{achievable_R1} \\
0 \le R_2 &\le& d_2 + \min\left\{(\C_2- \xi_2)^+, \Delta {d}_2  \right\},  \label{achievable_R2}\\
R_1 + R_2 &\le& a_1 + g_2 + \min \left\{(\C_1- \xi_1)^+, \Delta
{a}_1 \right\}   \nonumber \\
&& + \min \left\{(\C_2- \xi_2)^+, \Delta {g}_2 \right\}, \label{achievable_R1R2_1}\\
R_1 + R_2 &\le& a_2 + g_1 +  \min \left\{(\C_1- \xi_1)^+, \Delta
{g}_1 \right\} \nonumber \\
&& + \min \left\{(\C_2- \xi_2)^+, \Delta {a}_2 \right\} , \label{achievable_R1R2_2} \\
R_1 + R_2 &\le& e_1 + e_2 + \min \left\{(\C_1- \xi_1)^+, \Delta
{e}_1 \right\} \nonumber \\
&&  + \min \left\{(\C_2- \xi_2)^+, \Delta {e}_2 \right\}, \label{achievable_R1R2_3}\\
2R_1 + R_2 &\le& a_1 + g_1 + e_2 + \min\left\{ (\C_1- \xi_1)^+, \Delta {a}_1\right\} \nonumber \\
&&+ \min\left\{ (\C_1- \xi_1)^+, \Delta {g}_1\right\} \nonumber \\
&& + \min\left\{ (\C_2- \xi_2)^+, \Delta {e}_2\right\}, \label{achievable_2R1R2}\\
R_1 + 2R_2 &\le& a_2 + g_2 + e_1 + \min\left\{ (\C_2- \xi_2)^+, \Delta {a}_2\right\} \nonumber \\
&&+ \min\left\{ (\C_2- \xi_2)^+, \Delta {g}_2\right\} \nonumber \\
&& + \min\left\{
(\C_1- \xi_1)^+, \Delta {e}_1\right\}, \label{achievable_R12R2}
\end{eqnarray}
where
\begin{eqnarray}
    a_1 &=& I(X_{1}; Y_1|W_{1}, W_{2}, Q),  \label{a2g_begin}\\
    d_1 &=& I(X_1; Y_1|W_{2}, Q),  \\
    e_1 &=& I(X_{1},W_{2}; Y_1| W_{1}, Q),  \\
    g_1 &=& I(X_1, W_2; Y_1|Q), \\
\Delta {a}_1 &=& I(X_{1}; \hat{Y}_{R1}|Y_1, W_{1}, W_{2}, Q), \\
\Delta {d}_1 &=& I(X_1; \hat{Y}_{R1}|Y_1, W_{2}, Q), \\
\Delta {e}_1 &=& I(X_{1},W_{2}; \hat{Y}_{R1}|Y_1, W_{1}, Q), \\
\Delta {g}_1 &=& I(X_1, W_{2}; \hat{Y}_{R1}|Y_1, Q), \\
    \xi_1 &=& I(Y_{R}; \hat{Y}_{R1} | Y_1, X_1, W_{2}, Q),  \label{a2g_end}
\end{eqnarray}
and $a_2, \Delta {a}_2, d_2, \Delta {d}_2, e_2, \Delta {e}_2, g_2, \Delta {g}_2$, and $\xi_2$ are defined by (\ref{a2g_begin})-(\ref{a2g_end}) with
indices $1$ and $2$ switched. Then
\begin{equation}
\R = \bigcup_{P \in \mathcal{P}} \R(P)
\end{equation}
is an achievable rate region for the Gaussian relay-interference channel as shown in Fig.~\ref{channel_model}.
\end{theo}

\begin{proof}
The achievable scheme consists of a Han-Kobayashi strategy at the
transmitters and a generalized hash-and-forward strategy at the relay.
They are the same strategies as adopted in \cite{peyman_ghf_jnl} except
that unlike the GHF relaying scheme in
\cite[Theorem~2]{peyman_ghf_jnl}, where the relay quantizes the
received signal and broadcasts its bin index to both receivers through
a shared digital link, the relay here quantizes the received
signal with two different quantization resolutions, then sends the bin
indices of the quantized signals to the receivers through separated
digital links of rates $\C_1$ and $\C_2$. The following is a
sketch of the encoding/decoding process.

{\it{Encoding}}: Each transmit signal is comprised of a
common message of rate $T_i$ and a private message of rate $S_i$.
The common message codewords $W_i^n(j)$, $j=1,2,\cdots, 2^{nT_i}$ of length $n$ are
generated according to the probability distribution $\Pi_{i=1}^{n} p(w_i|q)$, where $q \sim p(q)$ serves as the time-sharing random variable. Based on the common message codewords, user $i$ generates codewords
$X_i^n(j, k),  k=1,2,\cdots,2^{nS_i}$ of length $n$ following the conditional distribution
$\Pi_{i=1}^{n} p(x_i|w_i, q)$. Each input message $\theta_i \in [1, 2, \cdots, 2^{S_i+T_i}], i=1,2$ is mapped to a message pair $(s_i, t_i) \in [1 , \cdots, 2^{S_i}]\times [1, \cdots, 2^{T_i}]$, then sent to the destinations as $X_i^n(s_i, t_i)$.  At the relay, the quantization codebook is generated according to the probability distribution $p(\hat{y}_{R1},\hat{y}_{R2}|y_R, q)$. After receiving $Y_R^n$, the relay
quantizes $Y_R^n$ into $\hat{Y}_{R1}^n$ and $\hat{Y}_{R2}^n$,
then bins $\hat{Y}_{R1}^n$ to $2^{nC_1}$ bins, and bins
$\hat{Y}_{R1}^n$ to $2^{nC_1}$ bins, and sends the bin indices to the
receivers through the digital links.

{\it{Decoding}}: The decoding process follows the Han-Kobayashi
framework: $X_1^n$ and $W_2^n$ are decoded by receiver $1$ with the help of the
index of the relayed message $\hat{Y}_{R1}^n$; $X_2^n$ and $W_1^n$ are decoded
by receiver $2$ with the help of the index of the relayed message
$\hat{Y}_{R2}^n$. To decode, receiver $1$ first constructs a list of candidates
for the relayed message $\hat{Y}_{R1}^n$, then jointly decodes $X_1^n$, $W_2^n$
and $\hat{Y}_{R1}^n$ using typicality decoding. Similarly, receiver $2$
jointly decodes $X_2^n$, $W_1^n$ and $\hat{Y}_{R2}^n$. Following the error probability analysis in \cite[Theorem~2]{peyman_ghf_jnl},  the rate tuple
$(S_1, T_1, S_2, T_2)$ satisfying the following constraints is achievable:

Constraints at receiver 1:
\begin{eqnarray}
S_1 &\le& \min \{I(X_{1}; Y_1|W_{1}, W_{2}, Q) +  (\C_1 - \xi_1)^+, \nonumber \\
	& & \qquad \; I(X_{1}; Y_1, \hat{Y}_{R1}|W_{1}, W_{2}, Q) \}  \\
S_1 + T_1 &\le& \min \{I(X_1; Y_1|W_{2}, Q)+ (\C_1 - \xi_1)^+, \nonumber \\
	& & \qquad \; I(X_1; Y_1, \hat{Y}_{R1}|W_{2}, Q) \} \\
S_1 + T_2 &\le& \min \{I(X_{1},W_{2}; Y_1|W_{1}, Q) + (\C_1 - \xi_1)^+,
	\nonumber \\
	& & \qquad \; I(X_{1},W_{2}; Y_1 ,\hat{Y}_{R1}|W_{1}, Q)\} \\
S_1 + T_1 + T_2 &\le& \min \{I(X_1, W_{2}; Y_1|Q) + (\C_1 - \xi_1)^+,
	\nonumber \\
	& & \qquad \; I(X_1, W_{2}; Y_1, \hat{Y}_{R1}|Q) \}
\end{eqnarray}

Constraints at receiver 2:
\begin{eqnarray}
S_2 &\le& \min \{I(X_{2}; Y_2|W_{1}, W_{2}, Q) +  (\C_2 - \xi_2)^+,  \nonumber \\
&& \qquad \; I(X_{2}; Y_2, \hat{Y}_{R2}|W_{1}, W_{2}, Q) \} \\
S_2 + T_2 &\le& \min \{I(X_2; Y_2|W_{1}, Q)+ (\C_2 - \xi_2)^+, \nonumber \\
&& \qquad \;  I(X_2; Y_2, \hat{Y}_{R2}|W_{1}, Q) \} \\
S_2 + T_1 &\le& \min \{I(X_{2},W_{1}; Y_2|W_{2}, Q) + (\C_2 - \xi_2)^+,  \nonumber \\
&& \qquad \; I(X_{2},W_{1}; Y_2 ,\hat{Y}_{R2}|W_{2}, Q)\} \\
S_2 + T_2 + T_1 &\le& \min \{I(X_2, W_{1}; Y_2|Q) + (\C_2 - \xi_2)^+,  \nonumber \\
&& \qquad \;  I(X_2, W_{1}; Y_2, \hat{Y}_{R2}|Q) \} \label{achievable_S_2_T_2_T_1}
\end{eqnarray}

The achievable rate region consists of all rate pairs $(R_1, R_2)$ such that
$R_1= S_1 + T_1$ and $R_2 =S_2 + T_2$. 
Applying the
Fourier-Motzkin elimination procedure \cite{HK2007} gives the
achievable rate region (\ref{achievable_R1})-(\ref{achievable_R12R2}).
\end{proof}

We remark here that although both Theorem~\ref{outerbound_theorem_twolinks} and Theorem~\ref{achievable_theorem_twolinks} are stated
for the digital noise-free relay-destination links, it can be easily
verified that both results continue to hold when the digital links are
replaced by analog additive Gaussian noise channels. The fact that the achievable rate region for
the analog channel is at least as large as the rate region for the digital channel is obvious since one can
always digitize the analog link. The fact that the outer bound continues to hold can be verified by going through the proof of that converse in Appendix~\ref{proof_outer_bound_twolinks}. The outer bounds in the converse involve terms like $I(X_1^n; Y_1^n, V_1^n)$, which is in turn upper bounded by $I(X_1^n; Y_1^n) + n\mathsf{C}_1$. It is easy to show that when the digital link $\mathsf{C}_1$ is replaced by an analog link with input $X_{a1}$ and output $Y_{a1}$, the mutual information term is upper bounded by $I(X_1^n; Y_1^n) + I(X_{a1}^n; Y_{a1}^n)$. As a result, all the outer bounds in Theorem~\ref{outerbound_theorem_twolinks} continue to hold in the case of the analog relay link with $\mathsf{C}_1$ replaced by $I(X_{a1}; Y_{a1})$ and $\mathsf{C}_2$ replaced by $I(X_{a2}; Y_{a2})$.

\subsection{Constant Gap in the Weak-Relay Regime}

We now specialize to the Gaussian case, and show that under the weak-relay condition (\ref{weak_relay_regime}), the achievable rate region and the outer bounds of
the Gaussian relay-interference channel with independent relay links
can be made to be within a constant gap to each other. The relaying
strategy that achieves this capacity to within a constant gap
turns out to be naturally suited for
the Gaussian relay-interference channel with a degraded broadcasting
relay, thus establishing the constant-gap result for the
broadcasting-relay case as well.

Assuming Gaussian codebooks and a Gaussian quantization scheme, the
key design parameters are the choice of common-private power splitting
ratio at the transmitters and the quantization level at the relay.
Our choice of design parameters is inspired by that of Wang and Tse
\cite{Wang_ReceiverCooperation}, where the capacity region of a
Gaussian interference channel with rate-limited receiver cooperation
is characterized to within a constant gap. Two key observations are
made in \cite{Wang_ReceiverCooperation}. First, the Etkin-Tse-Wang
strategy \cite{Tse2007} of setting the private power
to be at the noise level at the opposite receiver is used.
Second, the relay quantizes its observation at the private
signal level in order to preserve all the
information of interest to the destinations. At the destinations,
a joint decoding (see \cite{Kim_primitive, Dabora_jointdecoding,
Peyman_GHF,Avestimehr_relaynetwork}) is performed to recover the
source messages.

Consider now the optimal power splitting in a Gaussian
relay-interference channel with independent relay links.  The
Etkin-Tse-Wang strategy, i.e., setting private powers $P_{ip}$ as
\begin{equation} \label{ETW_SISO}
P_{1p} = \min \{1, h_{12}^{-2} \}, \quad P_{2p} = \min \{1, h_{21}^{-2} \}.
\end{equation}
is near optimal for the Gaussian interference channel with
conferencing receivers, but is not necessarily so for
relay-interference channel shown in
Fig.~\ref{channel_model} in its most general form.  Consider an
extreme scenario of $\C_1, \C_2 \rightarrow \infty$. In this case, the
relay fully cooperates with both receivers, so the relay-interference
channel becomes a single-input multiple-output (SIMO) interference
channel with two antennas at the receivers. Thus, the private powers
at the transmitters 
must be set at the effective noise level for the
two-antenna output in order to achieve capacity to within constant
bits \cite{Wang_Allerton} \cite{MIMO_IC}, i.e.,
\begin{equation} \label{ETW_SIMO}
P_{1p} = \min \{1, (g_1^2 + h_{12}^2)^{-1} \}, \quad P_{2p} = \min \{1, (g_2^2 + h_{21}^2)^{-1} \}.
\end{equation}
When $\C_1$ and $\C_2$ are finite, the optimal power splitting strategy is
expected to be a function of not only $h_{12}$ and $h_{21}$ but also
$g_1$, $g_2$, $\C_1$ and $\C_2$, lying somewhere between
(\ref{ETW_SISO}) and (\ref{ETW_SIMO}).

This complication can be avoided, however, if we focus on the weak-relay regime (\ref{weak_relay_regime}), namely $|g_1| \le
\sqrt{\rho}|h_{12}|$ and $|g_2| \le \sqrt{\rho}|h_{21}|$ for some
finite constant $\rho$. In this case, the power splittings
(\ref{ETW_SISO}) and (\ref{ETW_SIMO}) differ by at most a constant
factor. The main result of this section shows that in this weak-relay
regime, the Etkin-Tse-Wang's original power splitting (\ref{ETW_SISO})
is sufficient for achieving the capacity of the Gaussian
relay-interference channel to within a constant gap (which is a function of
$\rho$).



Consider next the optimization of the quantization level.  Applying the insight of
\cite{Wang_ReceiverCooperation} to the Gaussian relay-interference
channel with independent relay links shown in Fig.~\ref{channel_model},
the quantized messages for two receivers can be
expressed as
\begin{eqnarray}
\hat{Y}_{R1} = g_1 U_1 + g_1 W_1 + g_2 W_2 + \overbrace{g_2 U_2 +
Z_R}^\text{of no interest at $Y_1$}  + \eta_1 \\
\hat{Y}_{R2} = g_1 W_1 + g_2 U_2 + g_2 W_2 + \underbrace{g_1 U_1  +
Z_R}_\text{of no interest at $Y_2$}   + \eta_2
\end{eqnarray}
where $W_i$ and $U_i$ are common message and private message respectively,
and $\eta_i \sim \mathcal{N}(0, \mathsf{q}_i)$ is the quantization noise,
$i=1,2$. Therefore, a reasonable choice of the quantization levels for
receiver $1$ and receiver $2$ is
\begin{equation} \label{adaptive_q_level}
\mathsf{q}_1 = 1 + g_2^2 P_{2p}, \quad \mathsf{q}_2 = 1 + g_1^2 P_{1p}.
\end{equation}
Now observe that in the weak-relay regime, i.e., $|g_1| \le
\sqrt{\rho}|h_{12}|, |g_2| \le \sqrt{\rho}|h_{21}|$, the above
quantization levels (with Etkin-Tse-Wang power splitting) are
between $1$ and the constant $\rho+1$. Thus, we can choose
the quantization levels to be a constant and optimize it between $1$ and $\rho + 1$.


%

\begin{theo}[Constant Gap in the Weak-Relay Regime] \label{constantgap_theorem_twolinks}
For the Gaussian relay-interference channel with independent relay
links as depicted in Fig.~\ref{channel_model}, in the weak-relay regime,
using the generalized hash-and-forward relaying scheme with
quantization levels $\q_1=\q_2=\frac{\sqrt{\rho^2 + 16\rho +
16}-\rho}{4}$, where $\rho$ is defined in (\ref{definition_rho}),
and using the Han-Kobayashi scheme with Etkin-Tse-Wang power splitting strategy,
$X_i = U_i + W_i, i=1, 2$, where $U_i$ and $W_i$ are both Gaussian distributed with
the powers of $U_1$ and $U_2$ set according to $P_{1p}=\min\{1, h_{12}^{-2}\}$ and $P_{2p}
= \min \{1, h_{21}^{-2} \}$, respectively, the achievable rate region given in
Theorem~\ref{achievable_theorem_twolinks} is within
\begin{equation}
\delta = \hf \log\left(2 + \frac{\rho+\sqrt{\rho^2 + 16\rho + 16}}{2} \right)
\label{gap_to_capacity}
\end{equation}
bits of the capacity region outer bound in
Theorem~\ref{outerbound_theorem_twolinks}.
\end{theo}

\begin{IEEEproof}
The main step is to show that
using superposition coding $X_i = U_i + W_i, i=1, 2$, where $U_i \sim
\mathcal{N}(0, P_{ip})$ and $W_i \sim \mathcal{N}(0, P_{ic})$ with $P_{ip} +
P_{ic} =1$, $P_{1p}=\min\{1,
h_{12}^{-2}\}$, and $P_{2p} = \min \{1, h_{21}^{-2} \}$, each of the achievable
rate constraints in (\ref{achievable_R1})-(\ref{achievable_R12R2}) is within a finite gap to the corresponding upper bound in
(\ref{R1_bound_twolinks})-(\ref{2R1R2_bound_6_twolinks}). Specifically,
it is shown  in Appendix~\ref{proof_constant_theorem_twolinks}  that

(i) Individual rate (\ref{achievable_R1}) is within
\begin{equation}
\delta_{R_1} = \max \left\{\alpha(\q_1), \beta(\q_1) \right\}
\end{equation}
bits of the upper bound (\ref{R1_bound_twolinks}), where $\alpha(\cdot)$ and $\beta(\cdot)$ are as defined in (\ref{alpha_beta});

(ii) Individual rate (\ref{achievable_R2}) is within
\begin{equation}
\delta_{R_2} = \max \left\{\alpha(\q_2), \beta(\q_2) \right\}
\end{equation}
bits of the upper bound (\ref{R2_bound_twolinks});

(iii) Sum rates (\ref{achievable_R1R2_1}), (\ref{achievable_R1R2_2}), and
(\ref{achievable_R1R2_3}) are within
\begin{eqnarray}
\delta_{R_1+R_2} & = & \max \left\{\alpha(\q_1) + \alpha(\q_2), \alpha(\q_1) +
\beta(\q_2), \right. \nonumber \\
&& \quad \left. \beta(\q_1) + \alpha(\q_2),  \beta(\q_1) + \beta(\q_2) \right\}
\end{eqnarray}
bits of the upper bounds
(\ref{sumrate_bound_1_twolinks})-(\ref{sumrate_bound_12_twolinks});

(iv) $2R_1 +R_2$ rate (\ref{achievable_2R1R2}) is within
\begin{eqnarray}
\delta_{2R_1 +R_2} &=& \max \left\{ 2\alpha(\q_1)+ \alpha(\q_2),
2\beta(\q_1)+ \alpha(\q_2), \right. \nonumber \\
&& \qquad \; \; \alpha(\q_1)+ \beta(\q_1)+ \alpha(\q_2), \nonumber \\
&& \qquad \; \; 2\alpha(\q_1)+ \beta(\q_2), 2\beta(\q_1)+ \beta(\q_2) \nonumber \\
&& \qquad \; \; \left. \alpha(\q_1)+ \beta(\q_1)+ \beta(\q_2) \right\}
\end{eqnarray}
bits of the upper bounds
(\ref{2R1R2_bound_1_twolinks})-(\ref{2R1R2_bound_6_twolinks});

(v) $R_1 +2R_2$ rate (\ref{achievable_R12R2}) is within
\begin{eqnarray}
\delta_{R_1 +2R_2} &=& \max \left\{ \alpha(\q_1)+ 2\alpha(\q_2),
	\alpha(\q_1)+ 2\beta(\q_2), \right. \nonumber \\
&& \qquad \; \; \alpha(\q_1)+ \beta(\q_2)+ \alpha(\q_2), \nonumber \\
&& \qquad \; \; \beta(\q_1)+ 2\alpha(\q_2), \beta(\q_1)+ 2\beta(\q_2), \nonumber \\
&& \qquad \; \; \left. \beta(\q_1)+ \beta(\q_2)+ \alpha(\q_2) \right\}
\end{eqnarray}
bits of the upper bounds not shown explicitly in Theorem~\ref{outerbound_theorem_twolinks} but can be obtained by switching the indices $1$ and $2$ of
(\ref{2R1R2_bound_1_twolinks})-(\ref{2R1R2_bound_6_twolinks}).

Since $\alpha(\cdot)$ is a monotonically increasing
function and $\beta(\cdot)$ is a monotonically decreasing function. In order to
minimize the above gaps over $\q_1$ and $\q_2$, the quantization levels should
be set such that
\begin{equation}
\alpha(\q_1^*) = \beta(\q_1^*) = \alpha(\q_2^*) = \beta(\q_2^*),
\end{equation}
which results in $\q_1^* = \q_2^* = \frac{\sqrt{\rho^2 + 16\rho +
16}-\rho}{4}$. Substituting $\q_1^*$ and $\q_2^*$ into the above gaps, we prove
that the constant gap is $\delta$ bits per dimension, where $\delta$ is given in
(\ref{gap_to_capacity}).
\end{IEEEproof}

Note that the finite capacity gap is an increasing function of
$\rho$: smaller $\rho$ results in a smaller gap. In the case that
$\rho=1$, i.e., $|g_1| \le |h_{12}|$ and $|g_2| \le |h_{21}|$, the
optimal quantization levels are $\q_1^*=\q_2^*=\frac{\sqrt{33}-1}{4}$,
and the gap to the capacity is given by $\hf \log \left(
\frac{5+\sqrt{33}}{4} \right) =1.2128$ bits.


\subsection{Gaussian Relay-Interference Channel with a Broadcasting Relay}

\begin{figure}[t]
\centering
\subfigure[GHF]{ \psfrag{YR}{$Y_{R}^n$} \psfrag{q1}{$\mathsf{q}_1$}  \psfrag{q2}{$\mathsf{q}_2$}  \psfrag{YR1}{$\hat{Y}_{R1}^n$}   \psfrag{YR2}{$\hat{Y}_{R2}^n$} \psfrag{one}{$1$}  \psfrag{two}{$2$} \psfrag{cdots}{$\cdots$} \psfrag{nC1}{$2^{n\mathsf{C}_1}$}  \psfrag{nC2}{$2^{n\mathsf{C}_2}$}  \psfrag{index1}{$\mathrm{index}_1$} \psfrag{index2}{$\mathrm{index}_2$}
\includegraphics[width=3.1in]{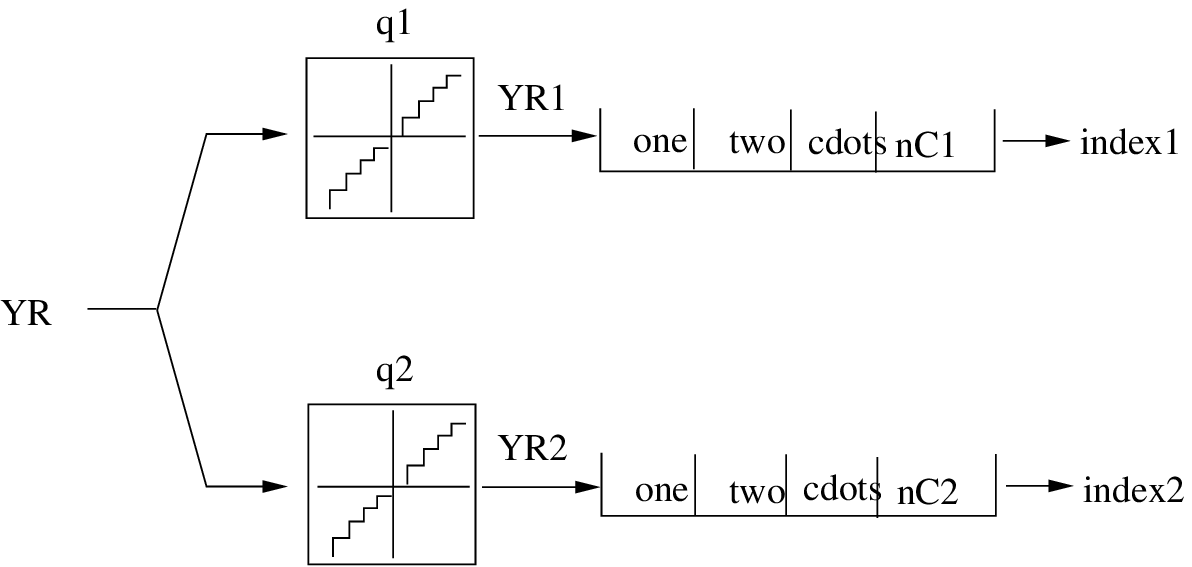}
\label{ghf_twolinks}
}
\subfigure[Incremental GHF with a refinement process]{
\psfrag{YR}{$Y_{R}^n$} \psfrag{q}{$\mathsf{q}$} \psfrag{YRhat}{$\hat{Y}_{R}^n$}   \psfrag{one}{$1$}  \psfrag{two}{$2$} \psfrag{cdots}{$\cdots$} \psfrag{nC1}{$2^{n\mathsf{C}_1}$}  \psfrag{nC2C1}{$2^{n(\mathsf{C}_2-\mathsf{C}_1)}$}  \psfrag{index1}{$\mathrm{index}_1$} \psfrag{index21}{$\mathrm{index}_{21}$}  \psfrag{index22}{$\mathrm{index}_{22}$}
\includegraphics[width=3.1in]{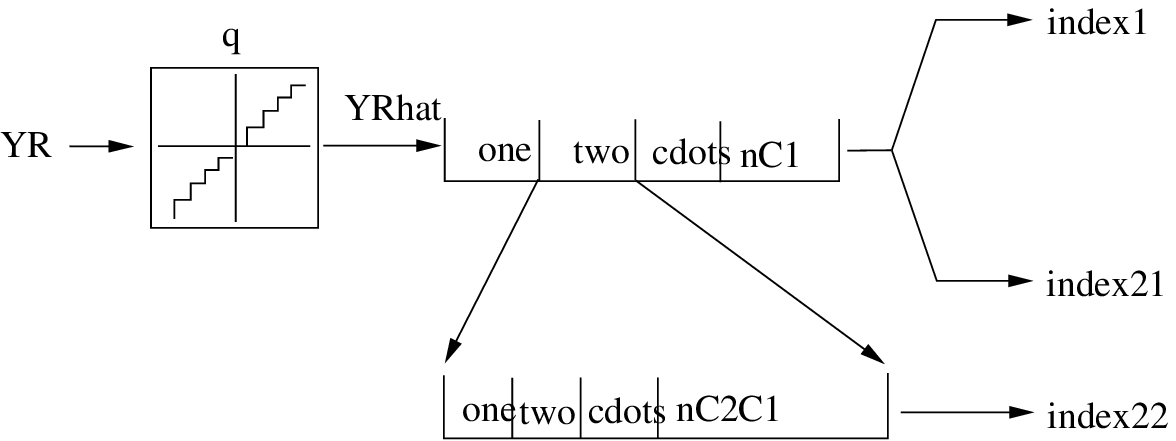}
\label{universal_ghf_twolinks}
}
\subfigure[Universal GHF]{
\psfrag{YR}{$Y_{R}^n$} \psfrag{q}{$\mathsf{q}$} \psfrag{YRhat}{$\hat{Y}_{R}^n$}   \psfrag{one}{$1$}  \psfrag{two}{$2$} \psfrag{cdots}{$\cdots$} \psfrag{nC1}{$2^{n\mathsf{C}}$} \psfrag{index1}{$\mathrm{index}_1$} \psfrag{index2}{$\mathrm{index}_{2}$}
\includegraphics[width=3.1in]{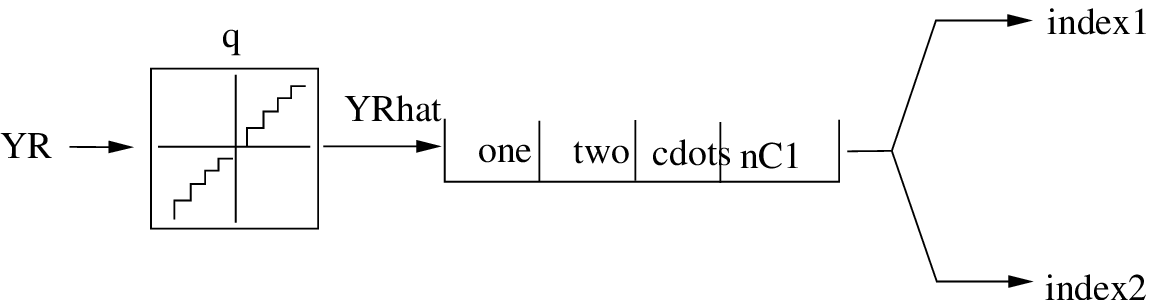}
\label{universal_ghf}
}
\label{fig:subfigureExample}
\caption{Evolution of the generalized hash-and-forward relay scheme}
\end{figure}

The GHF relaying scheme originally stated in
Theorem~\ref{achievable_theorem_twolinks} requires independent relay
links. As shown
in Fig.~\ref{ghf_twolinks}, the relay observation $Y_R^n$ undergoes
two separate quantization and binning processes to obtain the two
messages 
for the two receivers. However, in the weak-relay regime,
Theorem~\ref{constantgap_theorem_twolinks} shows that using an
identical quantization level for the two receivers is without loss
of approximate optimality, thus a common quantization process can be
shared between the two receivers.  Further, since the same
$\hat{Y}_R^n$ is binned into bins of sizes $2^{n\C_1}$ and
$2^{n\C_2}$, this is equivalent to first binning $\hat{Y}_R^n$ into
$2^{n\C_1}$ bins (assuming $\C_1 \le \C_2$) then further binning
each bin into $2^{n(\C_2 - \C_1)}$ sub-bins, as shown in
Fig.~\ref{universal_ghf_twolinks}. The message sent to receiver 2
can be thought of as the refinement of the message sent to receiver 1.
This is exactly the incremental relaying strategy we seek for the
Gaussian interference channel with a broadcasting relay, where the
message to receiver 1 is a degraded version of the message to receiver
2. Finally, if $\C_1=\C_2=\C$, the relay-interference channel reduces
to the universal relaying scheme studied in \cite{peyman_ghf_jnl},
where a digital link is shared between the relay and the receivers, as
shown in Fig.~\ref{universal_ghf}. We note here that the outer bounds for the independent relay link case (Theorem~\ref{outerbound_theorem_twolinks}) continues to hold for the degraded broadcast relay case.


\begin{corollary}
\label{incremental_relaying}
The constant-gap-to-capacity result stated in Theorem
\ref{constantgap_theorem_twolinks} holds also for the Gaussian
relay-interference channel with degraded broadcasting relay links,
where (assuming $\C_1 \le \C_2$) the message sent through the link with
capacity $\C_1$ must be a degraded version of the message sent through
the link with capacity $\C_2$.
\end{corollary}

\subsection{Comments on the Strong-Relay Regime}

The constant-gap result in this paper holds only in the weak-relay
regime of $|g_1| \le \sqrt{\rho}|h_{12}|$ and $|g_2| \le
\sqrt{\rho}|h_{21}|$, where $\rho$ is finite. The main difficulty in
extending this result to the general case is that both the
choice of the Han-Kobayashi power splitting ratio and the GHF relay
strategy are no longer optimal in the strong-relay regime. As mentioned earlier, the Etkin-Tse-Wang power splitting is
not optimal when the relay links $g_i, i=1, 2$ grow unboundedly
stronger than the interference links $h_{12}$ and $h_{21}$. Further, GHF may not be an appropriate relay
strategy. To see this, assume a channel model with separate relay
links and consider an extreme scenario where the relay
links $g_i, i=1,2$ go to infinity, while all other channel parameters
are kept constant. This special case is known as the cognitive
relay-interference channel \cite{Sahin_Erkip_2}. The capacity region outer bound of
Theorem~\ref{outerbound_theorem_twolinks} for this case reduces to
\begin{eqnarray}
R_1 &\le& \frac{1}{2}\log(1 + \mathsf{SNR}_1) + \C_1 \\
R_2 &\le& \frac{1}{2}\log(1 + \mathsf{SNR}_2) + \C_2 \\
R_1 + R_2 &\le& \frac{1}{2}\log(1 + \mathsf{SNR}_2 + \mathsf{INR}_1)  \nonumber \\
&& + \frac{1}{2}\log \left(1+  \frac{ \mathsf{SNR}_1}{1 + \mathsf{INR}_1}\right) + \C_1 + \C_2 \\
R_1 + R_2 &\le& \frac{1}{2}\log(1 + \mathsf{SNR}_1 + \mathsf{INR}_2) \nonumber \\
&& + \frac{1}{2}\log \left(1+ \frac{\mathsf{SNR}_2}{1 + \mathsf{INR}_2}\right) + \C_1 + \C_2 \\
R_1 + R_2 &\le& \frac{1}{2}\log \left(1 + \mathsf{INR}_2 + \frac{\mathsf{SNR}_1}{1 + \mathsf{INR}_1} \right) \nonumber \\
&& +\frac{1}{2}\log \left(1 + \mathsf{INR}_1 + \frac{\mathsf{SNR}_2}{1 + \mathsf{INR}_2}  \right) + \C_1 + \C_2 \nonumber \\
&& \\
2R_1 + R_2 &\le& \hf \log \left(1 + \SNR_1 + \INR_2 \right)  \nonumber \\
&& + \hf \log\left(1 + \INR_1 + \frac{\SNR_2}{1 + \INR_2} \right) \nonumber \\
&&+ \hf \log \left(1+ \frac{ \SNR_1}{1 + \INR_1}\right) + 2\C_1 + \C_2 \\
R_1 + 2R_2 &\le& \hf \log \left(1 + \SNR_2 + \INR_1 \right) \nonumber \\
&& + \hf \log\left(1 + \INR_2 + \frac{\SNR_1}{1 + \INR_1} \right) \nonumber \\
&&+ \hf \log \left(1+ \frac{ \SNR_2}{1 + \INR_2}\right) + \C_1 + 2\C_2,
\end{eqnarray}
which is in fact the outer bound of the underlying interference
channel expanded by $\C_1$ bits in the $R_1$ direction and $\C_2$ in
the $R_2$ direction. In this special case, a decode-and-forward
strategy can easily achieve the capacity region to within a constant gap.
This is because the relay is capable of decoding all the source messages,
so it can simply forward the bin indices of the privates messages to achieve $(R_1 + \C_1, R_2 + \C_2)$
for any achievable rate pair $(R_1, R_2)$ in the absence of the relay.
Etkin-Tse-Wang power splitting with decode-and-forward then achieves
the outer bound to within a constant gap. In contrast, GHF cannot achieve the capacity region to within a constant gap in this case.

\subsection{Generalized Degrees of Freedom}

We can gain further insights into the effect of relaying on the Gaussian
interference channel by analyzing the GDoF
of the sum rate in the symmetric channel setting.
Consider the case where $\INR_1 = \INR_2 = \INR$,
$\SNR_1 = \SNR_2 = \SNR$, $\SNR_{r1} = \SNR_{r2} = \SNR_r$, and $\C_1 = \C_2 =
\C$. In the high SNR regime, similar to \cite{Tse2007, Wang_ReceiverCooperation}, define
\begin{eqnarray}
\alpha &:=& \lim_{\SNR \rightarrow \infty} \frac{\log \INR}{\log \SNR}, \\
\beta  &:=& \lim_{\SNR \rightarrow \infty}\frac{\log \SNR_{r}}{\log \SNR},\\
\kappa &:=& \lim_{\SNR \rightarrow \infty}\frac{\C}{\hf \log \SNR}.
\end{eqnarray}
The GDoF of the sum capacity is defined as
\begin{equation}
d_{\mathrm{sum}} = \left. \lim_{\SNR \rightarrow \infty}\frac{C_{\mathrm{sum}}}{\hf \log \SNR} \right|_{\mathrm{fixed} \; \alpha, \beta, \kappa }
\end{equation}
As a direct consequence of the constant-gap result, $d_{\mathrm{sum}}$
can be characterized in the weak-relay regime as follows.

\begin{corollary}
For the symmetric Gaussian relay-interference channel in the weak-relay regime (i.e., $\beta \le \alpha$), the GDoF
of the sum capacity is given by the following. When $0 \le \alpha < 1$
\begin{eqnarray} \label{d_sym_twolinks_1}
d_{\mathrm{sum}} &=&  \min \left\{2-\alpha + \min\{\beta,
\kappa\}, 2\max\{\alpha, 1 -\alpha\}+ 2\kappa,  \right. \nonumber \\
&& \left. \qquad \;2 \max \{\alpha, 1 + \beta -\alpha \} \right\}.
\end{eqnarray}
When $\alpha \ge 1$
\begin{eqnarray} \label{d_sym_twolinks_2}
d_{\mathrm{sum}} = \min\left\{\alpha + \kappa, \alpha + \beta, 2(1 + \kappa), 2 \max\{1, \beta\} \right\}.
\end{eqnarray}
\end{corollary}
\vspace{1em}

Note that when $\alpha=1$, the GDoF of the sum capacity is in fact not
well defined. This is because both $\INR = \gamma \SNR$ (where $\gamma \neq 1$
is finite) and $\INR=\SNR$ result in the same $\alpha=1$. However, in the case of
$\INR=\SNR$, the channel becomes ill conditioned, i.e. $\phi_1 = \phi_2
= 0$, which results in a $d_{\mathrm{sum}}$ other than the one in
(\ref{d_sym_twolinks_2}). In other words, multiple values of $d_{\mathrm{sum}}$
correspond to the same $\alpha=1$. This is similar to the situation of
\cite[Theorem 7.3]{Wang_ReceiverCooperation}. Applying the similar argument
that the event $\left\{\INR = \SNR \right\}$ is of zero measure, we have the
GDoF of the sum capacity as shown in (\ref{d_sym_twolinks_2}) almost
surely.

When the relay links and the interference links share the same channel
gain, i.e. $\alpha=\beta$, the GDoF of the sum capacity reduces to
\begin{eqnarray} \label{sum_capacity_1}
d_{\mathrm{sum}} = \min \left\{2 + \kappa -\alpha, 2\max\{\alpha, 1 -\alpha\}+ 2\kappa, 2 \right\}
\end{eqnarray}
for $0 \le \alpha < 1$, and
\begin{equation} \label{sum_capacity_2}
d_{\mathrm{sum}} = \min \left\{\alpha + \kappa, 2(1 + \kappa), 2\alpha
\right\},
\end{equation}
for $\alpha \ge 1$. Interestingly, this is the same as the sum capacity
(in GDoF) of the Gaussian interference channel with rate-limited
receiver cooperation \cite{Wang_ReceiverCooperation}. Therefore, the
same sum capacity GDoF gain can be achieved with either receiver
cooperation or with an independent in-band-reception and
out-of-band-transmission relay assuming that the source-relay links
are the same as the interfering links of the underlying interference
channel (i.e. $\alpha=\beta$).

Fig.~\ref{comparison_impact_of_kappa} shows the GDoF gain due to
the relay for the $\alpha=\beta$
case.  There are several interesting features.
When $\kappa = 0.2$, the GDoF curve remains the
``W'' shape for the conventional Gaussian interference channel \cite{Tse2007}.
The sum-capacity gain is $2\kappa$ in the very and moderately weak
interference regimes (when $0.2 \le \alpha \le 0.6$) or the very strong
interference regime ($\alpha \ge 2.2$), and is $\kappa$ in other regimes ($\frac{2}{3} \le
\alpha \le 2$).  
As $\kappa$ gets larger, the left ``V'' branch of the ``W'' curve becomes
smaller, and it disappears completely at the critical point of $\kappa=0.5$.
As $\kappa$ keeps increasing, the right ``V'' of the ``W'' curve also
eventually disappears. 
The detailed sum-capacity gains for different values of
$\alpha$ are listed in Table~\ref{table_sum_gain_type_I}.




\begin{table*} [t]
    \center
    \caption{Sum-capacity GDoF gain due to the relay for the symmetric
	Gaussian relay-interference channel for the $\alpha=\beta$ and
	$\kappa \le \hf$ case}
    \label{table_sum_gain_type_I}
    \begin{tabular}{|c|c|c|c|c|c|c|}
    \hline
     Range of $\alpha$ & $\alpha \le \kappa$ &$\kappa \le \alpha \le \frac{2-\kappa}{3}$& $\frac{2-\kappa}{3} \le \alpha \le \frac{2}{3}$& $\frac{2}{3} \le \alpha \le 2$ & $2 \le \alpha \le 2+\kappa$ &$\alpha \ge 2 +\kappa$ \\ \hline
    Gain   & $2\alpha$ &$2\kappa$ &$2-3\alpha+\kappa$ & $\kappa$ &$\alpha+\kappa -2$ & $2\kappa$   \\ \hline
    \end{tabular}
\end{table*}

\begin{figure} [t]
\centering
\includegraphics[width=3.4in]{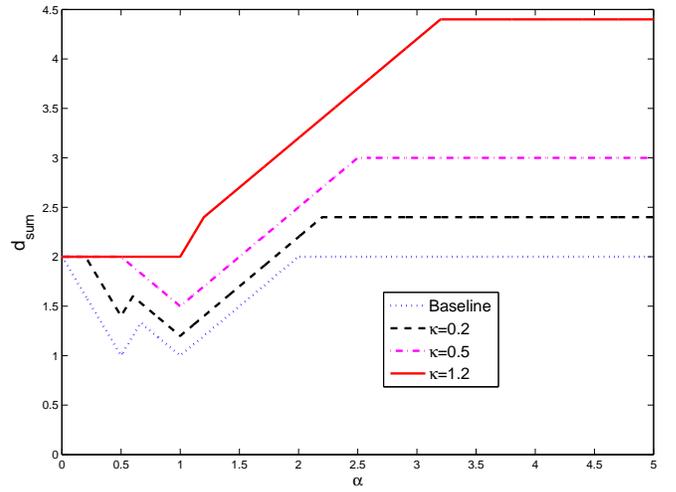}
\caption{The GDoF gain due to the relay in a symmetric Gaussian
relay-interference channel for the $\alpha=\beta$ case}
\label{comparison_impact_of_kappa}
\end{figure}

\subsection{Interpretation via the Deterministic Relay Channel}

In the Han-Kobayashi framework, each input signal of the interference
channel consists of both a common message and a private message. The
sum-capacity gain due to the relay in the relay-interference channel
therefore in general includes improvements in both the common
and the private message rates. This section illustrates that in the
asymptotic high SNR regime, the rate improvement can be
interpreted as either a private rate gain alone, or a common rate gain
alone. Further, the one-bit-per-relay-bit or the
two-bits-per-relay-bits GDoF  improvement shown in the previous
section can be interpreted using a deterministic relay model. The rest
of this section illustrates this point for the symmetric Gaussian
relay-interference channel in the $\alpha=\beta$ and $\kappa \le \hf$
case as an example.

\subsubsection{Very Weak Interference Regime}

For the symmetric Gaussian interference channel, in the very weak
interference regime of $0 \le \alpha \le \hf$, common messages do not
carry any information (although it can be assigned nonzero powers
as in the Etkin-Tse-Wang power splitting strategy). Setting $X_1$ and
$X_2$ to be private messages only is capacity achieving in terms of
GDoF  (\cite{Tse2007, Biao, Khandani08, VVV}).

Assigning $X_1$ and $X_2$ to be private only is also optimal for
GDoF  for the symmetric Gaussian relay-interference channel in the
very weak interference regime. This is because when $X_1$ and $X_2$ are
both private messages and are treated as noises at $Y_2$ and $Y_1$
respectively, the relay-interference channel asymptotically becomes
two deterministic relay channels in the high SNR regime.
Consider the relay operation for $Y_1$ as
illustrated in Fig.~\ref{weak_determin_relay_1}. When noise variances
of $Z_1$ and $Z_R$ go down to zero, the observation at the relay
becomes $Y_R = gX_1 + gX_2$ and the received signal at receiver $1$
becomes $Y_1 = h_d X_1 + h_c X_2$. In this case, the relay's
observation is a deterministic function of $X_1$ and $Y_1$, i.e. $Y_R
= gX_1 + \frac{g}{h_c}(Y_1
- h_dX_1)$. Thus $X_1$ and $Y_1$, along with the relay $Y_R$ form a
  deterministic relay channel of the type studied in
\cite{Kim2008}.  According to \cite{Kim2008}, the achievable rate of
user $1$ is given by
    \begin{eqnarray}
    R_1 &=& \min \left\{I(X_1; Y_1, Y_R), I(X_1; Y_1) + \C  \right\} \nonumber \\
    &=& \min \left\{\hf \log (1 + h_d^2), \hf \log \left(1 + \frac{h_d^2}{h_c^2} \right) + \C \right\}\nonumber \\
    &\rightarrow& \min\{1, 1 -\alpha + \kappa\},
    \end{eqnarray}
resulting in one-bit improvement for each relay bit in the regime
$\kappa \le \alpha \le \hf$.
Similarly, as illustrated in Fig.~\ref{weak_determin_relay_2}, $X_2$,
$Y_2$, and $Y_R$ form another deterministic relay channel with $X_2$
as the input, $Y_2$ as the output, and $Y_R$ as the relay. Thus, the
achievable rate of user $2$ is the same as user $1$, resulting in the
same one-bit-per-relay-bit improvement.  Further, as shown in
\cite{Kim2008}, a hash-and-forward relay strategy achieves the
capacity for deterministic relay channels.  As the hashing operation is the same for both case, the
same relay bit can therefore benefit both receivers at the same time,
resulting in two-bit increase in sum capacity for one relay bit, as
first pointed out in \cite{Peyman_GHF}.



\begin{figure*}[t]
\centering \subfigure[At receiver $1$, $X_2$ is treated as noise.]{ \psfrag{X1}{$X_1$} \psfrag{X2}{$X_2$} \psfrag{Y1}{$Y_1$}
\psfrag{Z1}{$Z_1$} \psfrag{ZR}{$Z_R$}
 \psfrag{h11}{$h_{d}$}
 \psfrag{h21}{$h_{c}$} \psfrag{G1}{$g$} \psfrag{G2}{$g$}  \psfrag{Relay}{$\mathrm{relay}$} \psfrag{C1}{$\mathsf{C}$} \psfrag{YR}{$Y_R$} \psfrag{Z2}{$ $}
\includegraphics[width=3.1in]{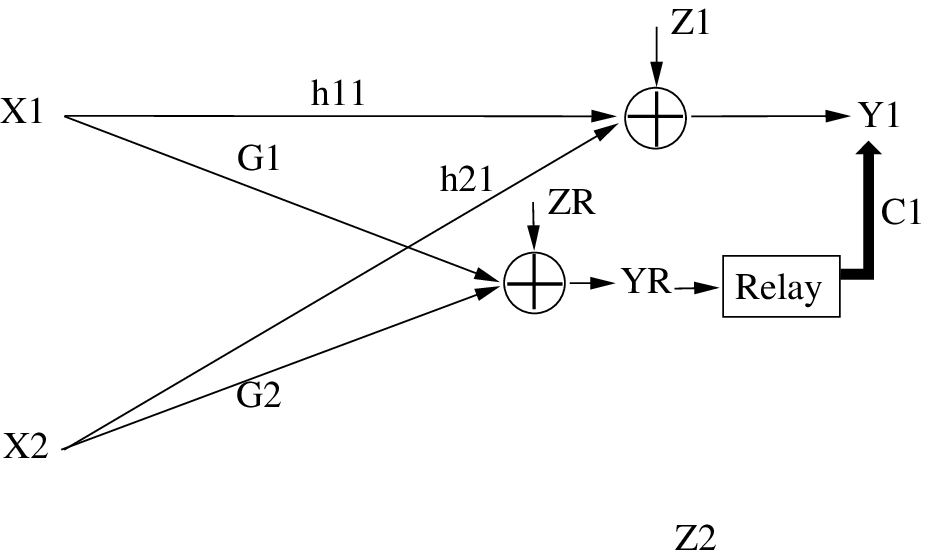}
\label{weak_determin_relay_1} }
\subfigure[At
receiver $2$, $X_1$ is treated as noise.]{ \psfrag{X1}{$X_1$} \psfrag{X2}{$X_2$} \psfrag{Y2}{$Y_2$}
\psfrag{Z2}{$Z_2$} \psfrag{ZR}{$Z_R$}
 \psfrag{h22}{$h_{d}$}
 \psfrag{h12}{$h_{c}$} \psfrag{G1}{$g$} \psfrag{G2}{$g$}  \psfrag{Relay}{$\mathrm{relay}$} \psfrag{C2}{$\mathsf{C}$} \psfrag{YR}{$Y_R$}
\psfrag{Z1}{$ $}
\includegraphics[width=3.1in]{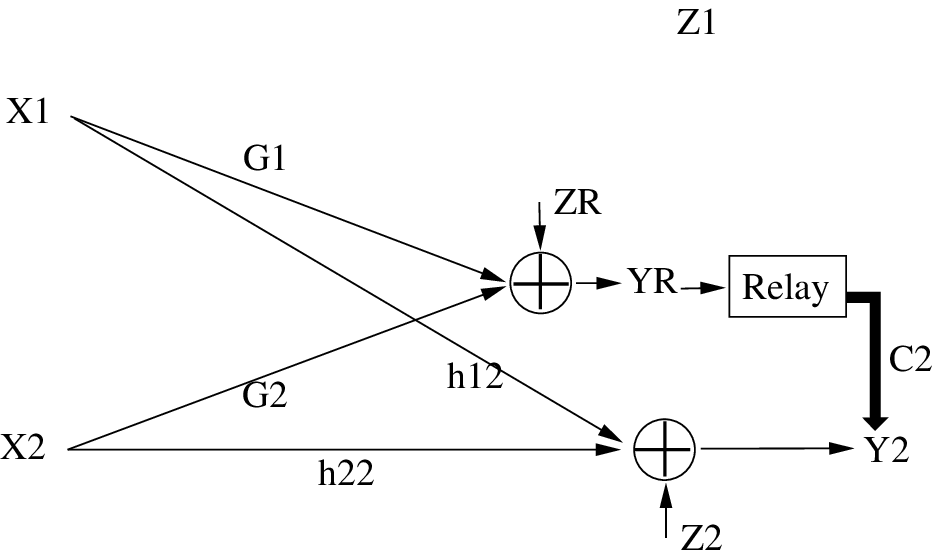}
\label{weak_determin_relay_2} } \caption{Asymptotic deterministic relay
channels in the very weak interference regime $\kappa \le \alpha \le \hf$.}
\end{figure*}


\subsubsection{Moderately Weak and Strong Interference Regimes}

The above interpretation, which states that the GDoF improvement in
the very weak interference regime comes solely from the private rate gain,
is not the only possible interpretation. The rate gain can also be
interpreted as improvement in common information rate --- an
interpretation that applies not only to the very weak
interference regime, but in fact to all regimes (for the symmetric rate with symmetric channels).
In the following, we illustrate this point by focusing on a two-stage
Han-Kobayashi strategy, where common messages are decoded first,
then the private messages. This is the same two-stage
Han-Kobayashi scheme used in \cite{Tse2007} for the Gaussian
interference channel without the relay.

Specifically, the relay uses the same GHF relaying strategy as in
Theorem~\ref{constantgap_theorem_twolinks}, but it is now designed to
help the common messages only. Here, both common messages $W_1^n$ and
$W_2^n$ are decoded and subtracted at both receivers
with the help of the GHF relay first (while treating private
messages as noise), the private
messages are then decoded at each receiver treating each other
as noise. The decoding of the private message at the second stage
results in
\begin{eqnarray} \label{simplified_HK_private_rates}
R_{u} &=& \hf \log \left(1 + \frac{\SNR_p}{1 + \INR_p} \right) \nonumber \\
&\rightarrow& \max\{0, 1 - \alpha\},
\end{eqnarray}
Note that the relay does not help the private rate.

In the common-message decoding stage, $W_1^n$ and $W_2^n$ are jointly
decoded at both receiver $1$ and receiver $2$
with the help of the GHF relay. As a
result, $(W_1^n, W_2^n, Y_1^n, Y_R^n)$ forms a multiple-access relay
channel at receiver $1$ with $W_1^n, W_2^n$ as the inputs, $Y_1^n$ as
the output and $Y_R^n$ as the relay. The achievable rate region of
such a multiple-access channel with a GHF relay is given by
\begin{eqnarray}
R_{w1} &\le& I(W_1; Y_1|W_2)  \nonumber \\
&& + \min \left\{(\C - \xi)^+, I(W_1; \hat{Y}_R |Y_1,
W_2) \right\} \nonumber \\
R_{w2} &\le& I(W_2; Y_1|W_1) \nonumber \\
&& + \min \left\{(\C - \xi)^+, I(W_2; \hat{Y}_R |Y_1,
W_1) \right\} \nonumber \\
R_{w1} + R_{w2} &\le& I(W_1, W_2; Y_1) \nonumber \\
&& + \min \left\{(\C - \xi)^+, I(W_1, W_2;
\hat{Y}_R |Y_1) \right\}. \nonumber
\end{eqnarray}
With the Etkin-Tse-Wang input strategy (i.e. $P_{1p} = \min \{1,
h_{12}^{-2}\}, P_{2p} = \min \{1, h_{21}^{-2}\}$) and the GHF relaying
scheme with $\q_1=\q_2 =\frac{\sqrt{\rho^2 + 16\rho + 16}-\rho}{4}$,
it can be shown that the common-message rate region for the receiver 1
in the high SNR regime in term of GDoF is given as follows.
When $0 \le \alpha \le 1$
\begin{eqnarray}
R_{w1} &\le&  \alpha \nonumber \\
R_{w2} &\le& \min \{\alpha, \kappa + \max\{2\alpha -1, 0\} \} \nonumber \\
R_{w1}+R_{w2} &\le& \alpha + \min\{\alpha, \kappa\} \nonumber.
\end{eqnarray}
When $\alpha \ge 1$
\begin{eqnarray}
R_{w1} &\le&  \min \{\alpha, 1 + \kappa\} \nonumber \\
R_{w2} &\le& \alpha  \nonumber \\
R_{w1}+R_{w2} &\le& \alpha + \kappa \nonumber.
\end{eqnarray}
Due to symmetry, the rate region for the multiple-access relay channel
at receiver $2$ can be obtained by switching the indices $1$ and $2$.

Note that in suitable interference regimes, both the individual rate
and the sum rate can potentially be increased by one bit for each
relay bit. This is again a consequence of the fact that the relay
operation has a deterministic relay
channel interpretation in the high SNR regime. For example, in the
strong interference regime where $1 \le \alpha \le 2 + \kappa$, the
sum rate of the multiple-access relay channel benefits by one bit for
each relay bit in the high SNR regime as shown in
Fig.~\ref{strong_determin_relay}. In the very strong interference
regime, the interference can be decoded, subtracted or can serve as side
information, therefore the individual rate increases by one bit for
each relay bit as shown in Fig.~\ref{verystrong_determin_relay}.

\begin{figure*}[t]
\centering \subfigure[$Y_1$ decodes both $X_1$ and $X_2$.]{ \psfrag{X1}{$X_1$} \psfrag{X2}{$X_2$} \psfrag{Y1}{$Y_1$}
\psfrag{Z1}{$Z_1$} \psfrag{ZR}{$Z_R$}
 \psfrag{h11}{$h_{d}$}
 \psfrag{h21}{$h_{c}$} \psfrag{G1}{$g$} \psfrag{G2}{$g$}  \psfrag{Relay}{$\mathrm{relay}$} \psfrag{C1}{$\mathsf{C}$} \psfrag{YR}{$Y_R$} \psfrag{Z2}{$ $}
\includegraphics[width=3.1in]{./figures/Weak_determin_Y1.eps}
\label{strong_determin_relay} }
\subfigure[
$X_2$ is decoded and serves as side information.]{ \psfrag{X1}{$X_1$} \psfrag{Y1}{$Y_1$} \psfrag{Z1}{$Z_1$} \psfrag{ZR}{$Z_R$}
 \psfrag{h11}{$h_{d}$} \psfrag{G1}{$g$} \psfrag{Relay}{$\mathrm{relay}$} \psfrag{C1}{$\mathsf{C}$} \psfrag{YR}{$Y_R$}  \psfrag{Z2}{$ $}
\includegraphics[width=3.1in]{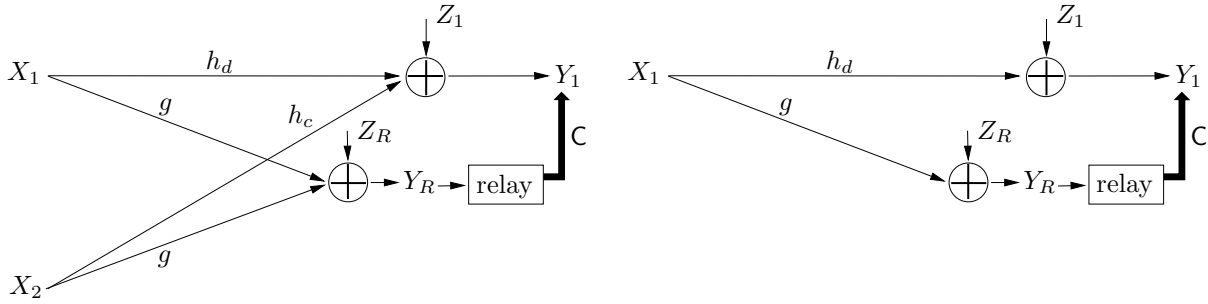}
\label{verystrong_determin_relay} } \caption{Asymptotic deterministic relay
channels in the strong and very strong interference regimes.}
\end{figure*}

Now, the achievable rates of common messages can be obtained by
intersecting the two rate regions. Taking the achievable rates of
private messages in (\ref{simplified_HK_private_rates}) into account,
it is easy to verify that this two-stage Han-Kobayashi scheme achieves
the sum capacity in (\ref{sum_capacity_1}) and (\ref{sum_capacity_2}).
As depicted in Fig.~\ref{comparison_impact_of_kappa}, the sum-capacity gain due to the
relay can be one-bit-per-bit or two-bits-per-bit. In the following, we
demonstrate in Fig.~\ref{fig:common_msg_intersection} how these gains
are obtained by pictorially showing the intersection of the two
common-message rate regions for different values of $\alpha$.

\begin{figure*}[ht]
\centering \subfigure[$0 \le \alpha \le \kappa$]{ \psfrag{Rw1}{$R_{w1}$}  \psfrag{Rw2}{$R_{w2}$}  \psfrag{alpha}{$\alpha$}  \psfrag{kappa}{$\kappa$}  \psfrag{0.5}{$0.5$} \psfrag{1}{$1$}  \psfrag{zero}{$0$}   \psfrag{DeltaRw}{\small $\Delta R_w = 2\alpha$} \psfrag{Withrelay}{\small $ \mathrm{With\; relay}$} \psfrag{Withoutrelay}{\small $\mathrm{Without \; relay}$}
\includegraphics[width=2.2in]{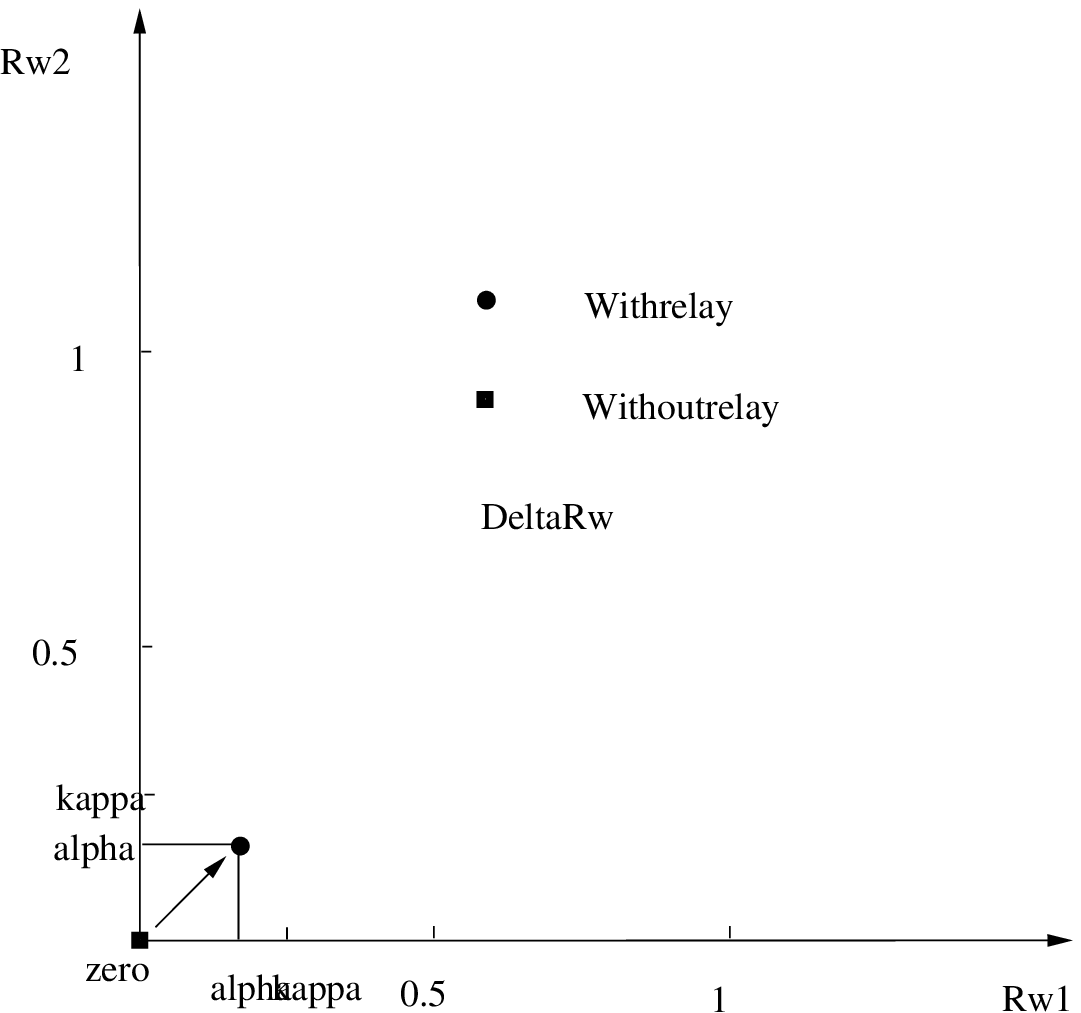}
\label{common_msg_intersection_1} } \subfigure[$\kappa \le \alpha \le \hf$]{ \psfrag{Rw1}{$R_{w1}$}  \psfrag{Rw2}{$R_{w2}$}  \psfrag{alpha}{$\alpha$}  \psfrag{kappa}{$\kappa$}  \psfrag{0.5}{$0.5$} \psfrag{1}{$1$}  \psfrag{zero}{$0$}   \psfrag{DeltaRw}{\small $\Delta R_w = 2\kappa$} \psfrag{Withrelay}{\small $ \mathrm{With\; relay}$} \psfrag{Withoutrelay}{\small $\mathrm{Without \; relay}$}
\includegraphics[width=2.2in]{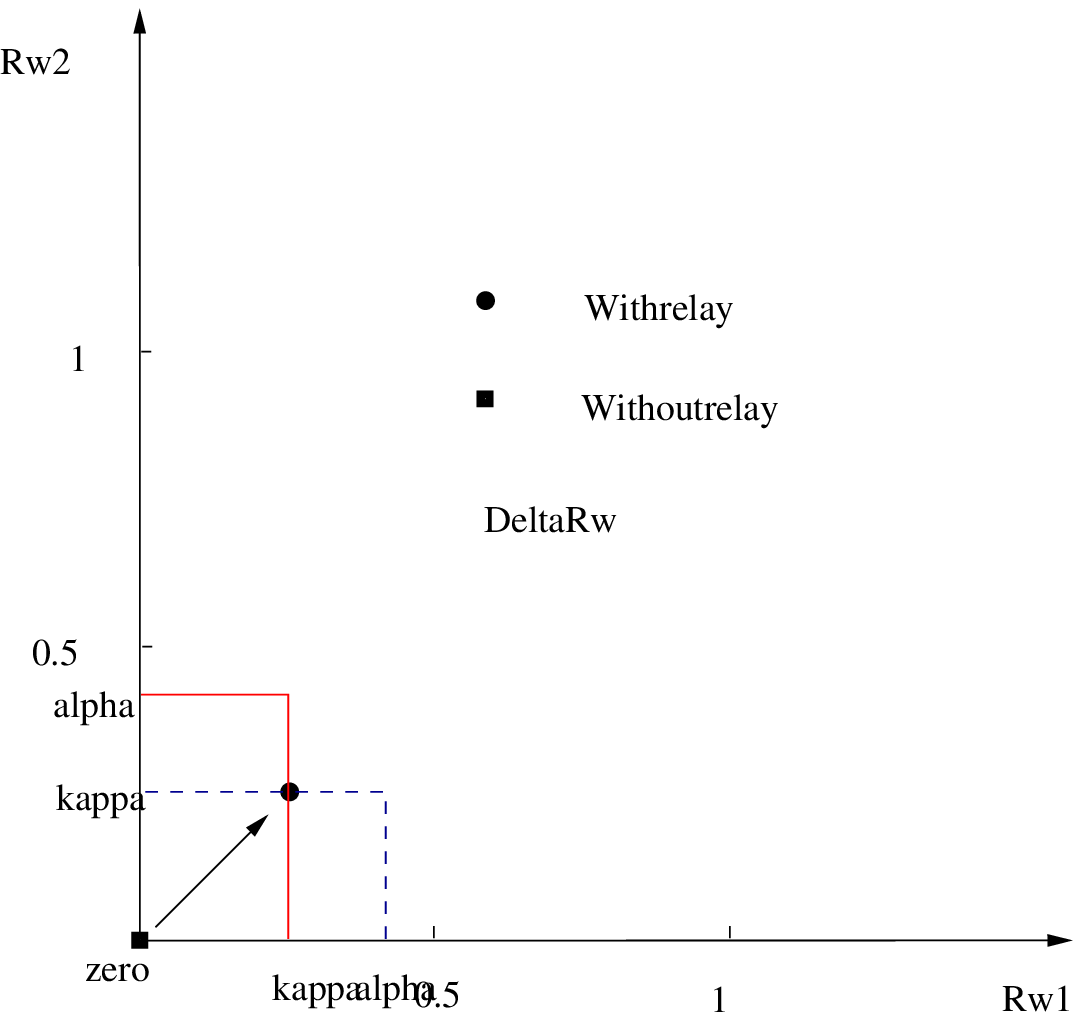}
\label{common_msg_intersection_2} } \subfigure[$\hf \le \alpha \le
\frac{2-\kappa}{3}$]{ \psfrag{Rw1}{$R_{w1}$}  \psfrag{Rw2}{$R_{w2}$}  \psfrag{alpha}{$\alpha$}  \psfrag{kappa}{$\kappa$}  \psfrag{0.5}{$0.5$} \psfrag{1}{$1$}  \psfrag{zero}{$0$}   \psfrag{DeltaRw}{\small $\Delta R_w = 2\kappa$} \psfrag{Withrelay}{\small $ \mathrm{With\; relay}$} \psfrag{Withoutrelay}{\small $\mathrm{Without \; relay}$}
\includegraphics[width=2.2in]{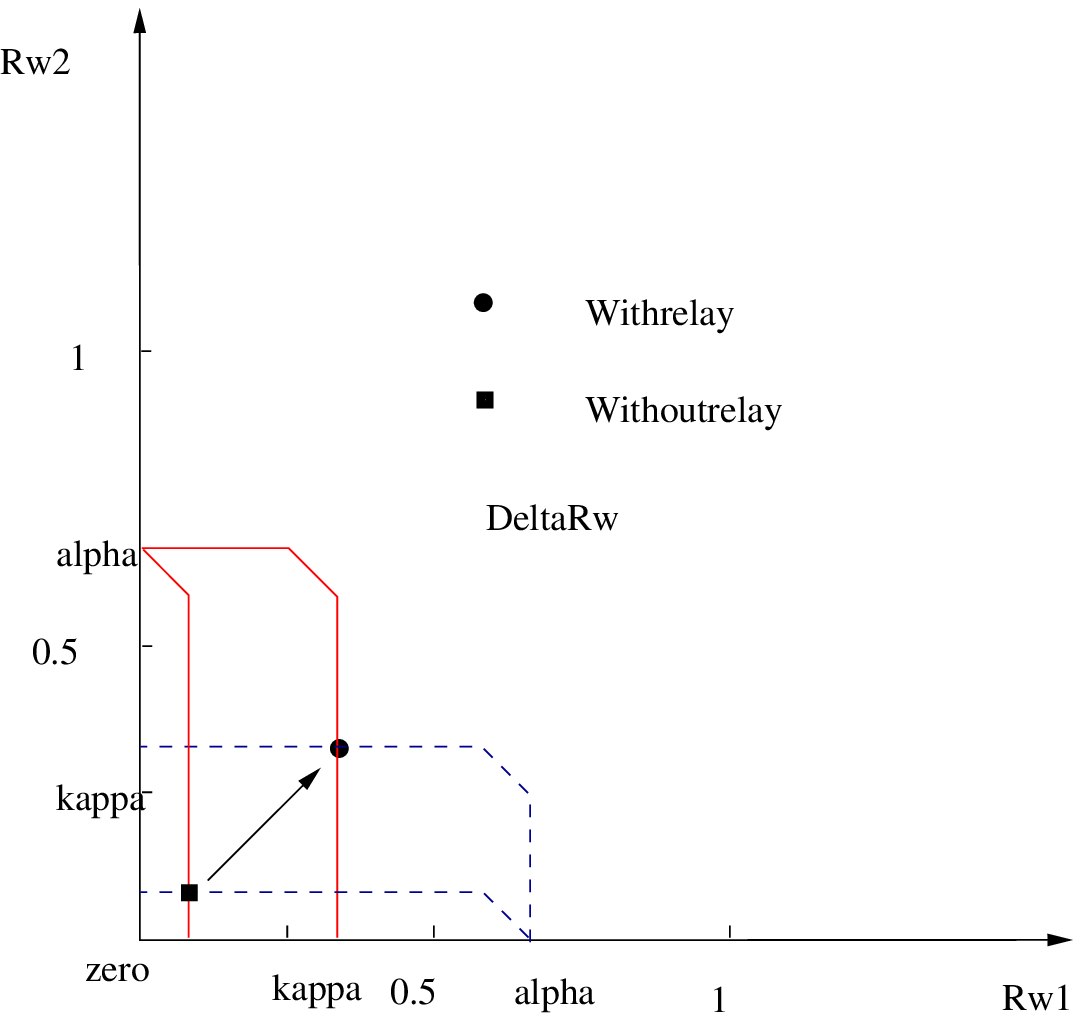}
\label{common_msg_intersection_3} } \subfigure[$\frac{2}{3} \le \alpha
\le 1$]{ \psfrag{Rw1}{$R_{w1}$}  \psfrag{Rw2}{$R_{w2}$}  \psfrag{alpha}{$\alpha$}  \psfrag{kappa}{$\kappa$}  \psfrag{0.5}{$0.5$} \psfrag{1}{$1$}  \psfrag{zero}{$0$}   \psfrag{DeltaRw}{\small $\Delta R_w = \kappa$} \psfrag{Withrelay}{\small $ \mathrm{With\; relay}$} \psfrag{Withoutrelay}{\small $\mathrm{Without \; relay}$}
\includegraphics[width=2.2in]{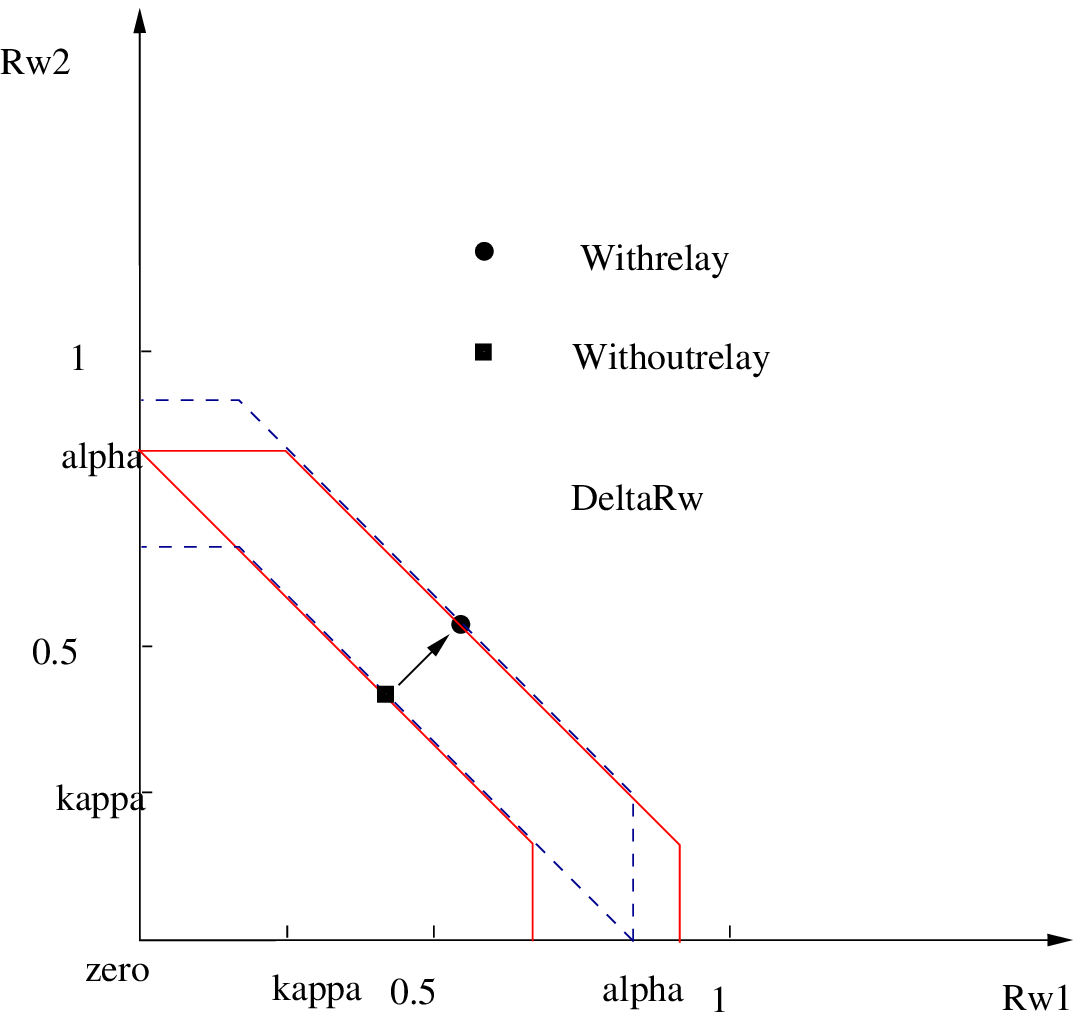}
\label{common_msg_intersection_4} } \subfigure[$1 \le \alpha \le 1 +\kappa$]{
\psfrag{Rw1}{$R_{w1}$}  \psfrag{Rw2}{$R_{w2}$}  \psfrag{alpha}{$\alpha$}  \psfrag{kappa}{$\kappa$}  \psfrag{0.5}{$0.5$} \psfrag{1}{$1$}  \psfrag{zero}{$0$}   \psfrag{DeltaRw}{\small $\Delta R_w = \kappa$} \psfrag{Withrelay}{\small $ \mathrm{With\; relay}$} \psfrag{Withoutrelay}{\small $\mathrm{Without \; relay}$}
\includegraphics[width=2.2in]{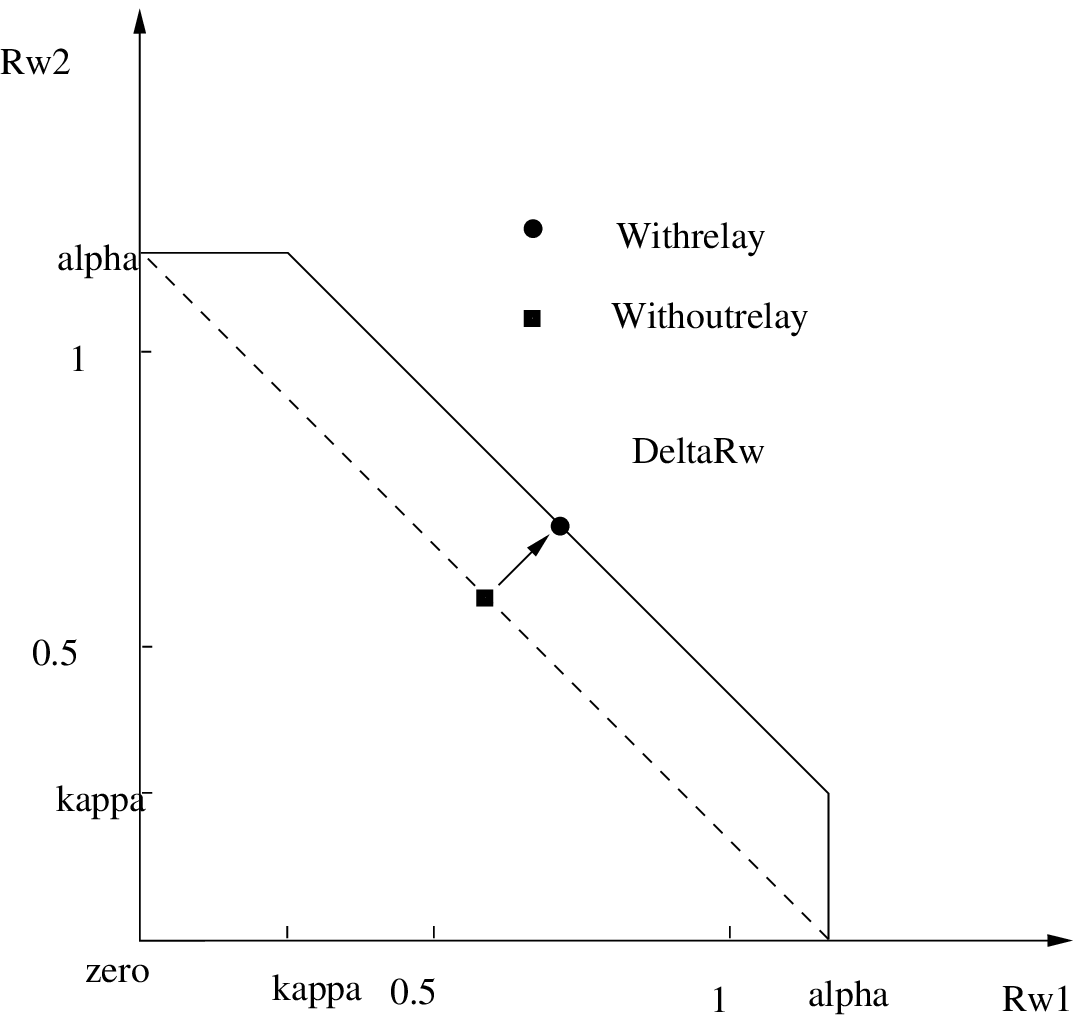}
\label{common_msg_intersection_5} } \subfigure[$ \alpha \ge 2+\kappa $]{ \psfrag{Rw1}{$R_{w1}$}  \psfrag{Rw2}{$R_{w2}$}  \psfrag{alpha}{$\alpha$}  \psfrag{kappa}{$\kappa$}  \psfrag{0.5}{$0.5$} \psfrag{1}{$1$}  \psfrag{zero}{$0$}   \psfrag{DeltaRw}{\small $\Delta R_w = 2\kappa$} \psfrag{Withrelay}{\small $ \mathrm{With\; relay}$} \psfrag{Withoutrelay}{\small $\mathrm{Without \; relay}$} \psfrag{1pluskappa}{$1+\kappa$}
\includegraphics[width=2.2in]{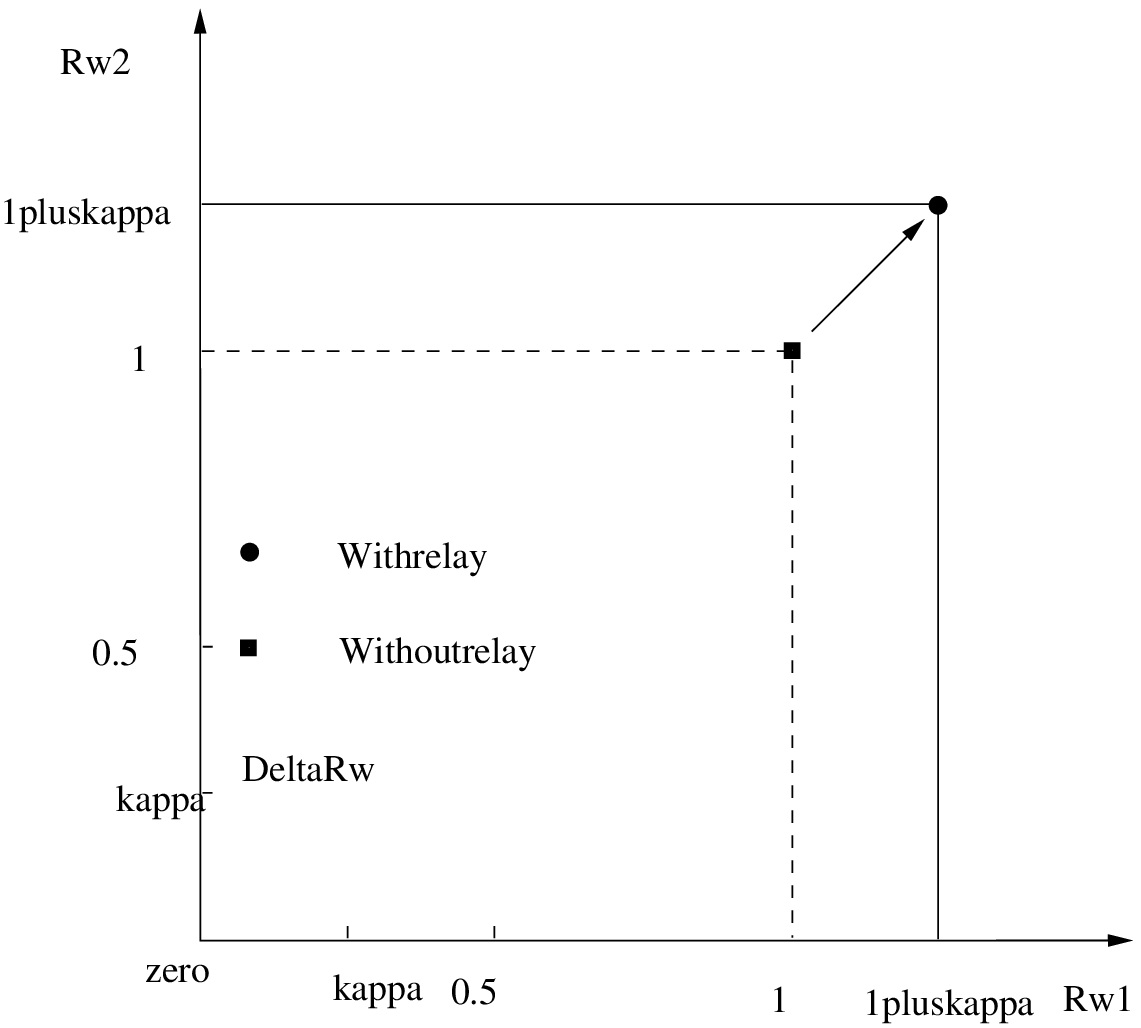}
\label{common_msg_intersection_6} } \caption{Generalized-degree-of-freedom gain due to relaying is roughly
$\kappa$ or $2\kappa$ depending on how the two common-message multiple-access regions are intersected} \label{fig:common_msg_intersection}
\end{figure*}

\begin{itemize}
\item When $\alpha \le \kappa$,
as can be seen from Fig.~\ref{common_msg_intersection_1},
the two rate regions are identical and are both given by
$\left\{(R_{w1}, R_{w2}): R_{w1} \le \alpha, R_{w2} \le \alpha \right\}$. The
intersection of the two is the same rectangle with the top-right corner located at
$(\alpha, \alpha)$. This gives a $2\alpha$-bit gain over the baseline, which is located at the origin.

\item As $\alpha$ increases to $\kappa \le \alpha \le \hf$, the baseline rate pair is still at the origin. With the help of the relay, the two common-message rate
regions become rectangles $\left\{(R_{w1}, R_{w2}): R_{w1} \le \alpha, R_{w2}
\le \kappa \right\}$ and $\left\{(R_{w1}, R_{w2}):R_{w1} \le \kappa, R_{w2} \le
\alpha \right\}$ respectively. As shown in
Fig.~\ref{common_msg_intersection_2}, the intersection of the two gives a
square with the top-right corner located at $(\kappa, \kappa)$. As a result, the sum-capacity
gain is $2\kappa$ bits.

\item As $\alpha$ increases to $\hf \le \alpha \le 1$, the common-message
rate regions at receivers $1$ and $2$ become pentagons. However, depending on
the value of $\alpha$, the sum rate improves by different amounts.
When $\alpha \le \frac{2-\kappa}{3}$, as shown in
Fig.~\ref{common_msg_intersection_3}, the intersection of the two pentagon
regions gives a square shape with the top-right corner located at $(2\alpha-1 + \kappa,
2\alpha-1+\kappa)$. Compared with $(2\alpha-1, 2\alpha-1)$
achieved without the relay, a sum-capacity gain of $2\kappa$ bits is obtained.
However, when $\alpha \ge \frac{2-\kappa}{3}$, as depicted
in Fig.~\ref{common_msg_intersection_4}, the intersection of the two rate
regions is still a pentagon with the sum-capacity limited by
$R_{w1}+R_{w2} \le 2 - \alpha + \kappa$. In this case, depending on the value
of $\alpha$, the sum-rate gain is $2 - 3\alpha + \kappa$ bits when
$\frac{2-\kappa}{3} \le \alpha \le \frac{2}{3}$, and is $\kappa$ bits when
$\frac{2}{3} \le \alpha \le 1$. (The latter case is shown in
Fig.~\ref{common_msg_intersection_4}.)

\item When $1 \le \alpha \le 2+\kappa$, the common-message rate regions
are again pentagons and the interpretation is similar to the
$\frac{2-\kappa}{3} \le \alpha \le 1$ case. Fig.~\ref{common_msg_intersection_5}
shows an example of $1 \le \alpha \le 1 + \kappa$. In this case, the two rate
regions are identical pentagons with the sum capacity limited by $R_{w1}+R_{w2} \le
\alpha+\kappa $. Compared with the baseline sum capacity, a $\kappa$-bits gain is obtained. When $1 + \kappa \le
\alpha \le 2+\kappa$, the intersection of the two common-message rate regions
again gives a sum-capacity of $\alpha +\kappa$. However, since the baseline sum capacity
becomes saturated when when $\alpha \ge 2$ (\cite{Tse2007, Sato, Carleial1978}), the sum-capacity gain over the baseline is $\kappa$ bits when $1 \le
\alpha \le 2$, and is $\alpha+\kappa-2$ bits when $2 \le \alpha \le 2 + \kappa$.

\item Finally, $\alpha \ge 2 + \kappa$ falls into the very strong interference
regime. The intersection of the two common-message rate regions is a rectangle with the top-right corner located at $(1 +\kappa, 1 + \kappa)$ as shown
in Fig.~\ref{common_msg_intersection_6}. The sum-capacity gain is thus
$2\kappa$ bits in the very strong interference regime.
\end{itemize}


\section{Gaussian Relay-Interference Channel With a Single Digital Link}

The result of the previous section shows that for the symmetric
channel, the sum-capacity improvement can be thought as coming solely
from the improvement of the common message rate, or in a very weak
interference regime as coming solely from the
improvement of the private message rates. Thus, the function of
the relay for the symmetric rate in symmetric channel is solely in
forwarding useful signals. This interpretation does not necessarily
hold for the asymmetric cases. In this section, we study a
particular asymmetric channel to illustrate the composition of the
sum-capacity gain. 
We are motivated by the fact that the relay's
observation in a relay-interference channel is a linear combination of
the intended signal and the interfering signal. Clearly, forwarding
the intended signal and the interfering signal can both be beneficial
(e.g.\ \cite{Dabora_Maric}). This section
illustrates that depending on the different channel parameters, the sum-rate gain from forwarding both intended signal and interference signal happens to be the same as that of forwarding intended signal only or forwarding interference signal only.


Specifically, we focus on a particular asymmetric model as
shown in Fig.~\ref{channel_model_single_digital_link}, where the
digital relay link exists only for receiver $1$, and not for receiver
$2$, i.e., $\C_2=0$.  This section first derives a
constant-gap-to-capacity result for this channel.  Note that this
channel is a special case of the general channel model studied in the
previous section, but the constant-gap-to-capacity result can be
established in this special case for a broader set of channels.
Unlike the weak-relay assumption $|g_1| \le \sqrt{\rho}|h_{12}|$ and
$|g_2| \le \sqrt{\rho}|h_{21}|$ made in the previous section, this
section assumes that $|g_2| \le \sqrt{\rho}|h_{21}|$ only with no
constraints on $g_1$ or $h_{12}$. Under this channel setup, it can be
shown that in the high SNR regime, the sum capacity improvement 
can also be obtained as if only the intended signal is forwarded or only the interference signal is forwarded. Note that this conclusion applies to the case of a single relay-destination link only, and not necessarily to the general case with two relay-destination links.

\subsection{Capacity Region to within Constant Gap in the Weak-Relay Regime}

Since the channel model studied in
Fig.~\ref{channel_model_single_digital_link} is a special case of the
general Gaussian relay-interference channel, we first simplify the
achievable rate region in Theorem~\ref{achievable_theorem_twolinks} to
the following corollary by setting $\C_2=0$. The only difference in
the coding scheme is that instead of performing two quantizations as in
the general relay-interference channel, the relay in
Fig.~\ref{channel_model_single_digital_link} does one quantization of
the received signal $Y_R$ into $\hat{Y}_{R1}$ and sends the bin index
of $\hat{Y}_{R1}$ to receiver $1$ through the digital link $\C_1$.

\begin{corollary} \label{achievable_theorem_single_digital_link}
\begin{figure} [t]
\centering
\psfrag{X1}{$X_1$} \psfrag{X2}{$X_2$} \psfrag{Y1}{$Y_1$}
\psfrag{Y2}{$Y_2$}
\psfrag{Z1}{$Z_1$} \psfrag{Z2}{$Z_2$} \psfrag{ZR}{$Z_R$}
 \psfrag{h11}{$h_{11}$} \psfrag{h22}{$h_{22}$}
 \psfrag{h21}{$h_{21}$}  \psfrag{h12}{$h_{12}$} \psfrag{G1}{$g_1$} \psfrag{G2}{$g_2$}  \psfrag{Relay}{$\mathrm{relay}$} \psfrag{C1}{$\mathsf{C}_1$} \psfrag{YR}{$Y_R$}
\includegraphics[width=3.1in]{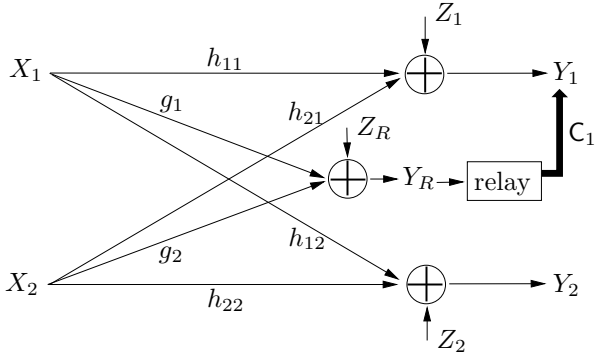}
\caption{Gaussian relay-interference channel with a single digital link}
\label{channel_model_single_digital_link}
\end{figure}

For the Gaussian relay-interference channel with a single digital link as shown in Fig.~\ref{channel_model_single_digital_link}, the following rate region is
achievable:
\begin{eqnarray}
0 \le R_1 &\le& d_1 + \min\left\{(\C_1- \xi_1)^+, \Delta {d}_1  \right\} \label{R1_single_digital_link} \\
0 \le R_2 &\le& d_2  \label{R2_single_digital_link}\\
R_1 + R_2 &\le& a_1 + g_2 + \min \left\{(\C_1- \xi_1)^+, \Delta {a}_1 \right\} \label{sumrate_1_single_digital_link}\\
R_1 + R_2 &\le& a_2 + g_1 + \min \left\{(\C_1- \xi_1)^+, \Delta {g}_1 \right\} \label{sumrate_2_single_digital_link} \\
R_1 + R_2 &\le& e_1+e_2 + \min \left\{(\C_1- \xi_1)^+, \Delta {e}_1 \right\} \label{sumrate_3_single_digital_link}\\
2R_1 + R_2 &\le& a_1 + g_1 + e_2  + \min\left\{ 2(\C_1- \xi_1)^+,  \right. \nonumber \\
&& \left. \quad (\C_1- \xi_1)^+ + \Delta {a}_1, \Delta {a}_1+  \Delta {g}_1\right\} \label{2R1R2_1_single_digital_link}\\
R_1 + 2R_2 &\le& a_2 + g_2 + e_1 + \min \left\{(\C_1- \xi_1)^+, \Delta
{e}_1 \right\},  \nonumber \\
&& \label{R12R2_1_single_digital_link}
\end{eqnarray}
where all the parameters are as defined in Theorem~\ref{achievable_theorem_twolinks}.
\end{corollary}

The proof follows directly from Theorem~\ref{achievable_theorem_twolinks}.
Note that in (\ref{2R1R2_1_single_digital_link}), we apply the fact that
$\Delta {a}_1 \le \Delta {g}_1$. Likewise, the capacity region outer bound in
Theorem~\ref{outerbound_theorem_twolinks} also simplifies when $\C_2=0$.
We can now prove the following constant-gap theorem for the Gaussian
relay-interference channel with a single digital link.

\begin{theo} \label{constantgap_theorem_single_digital_link}
For the Gaussian relay-interference channel with a single digital link as
depicted in Fig.~\ref{channel_model_single_digital_link}, with the same
signaling strategy as in Theorem~\ref{constantgap_theorem_twolinks}, i.e. a
combination of the Han-Kobayashi scheme with Etkin-Tse-Wang power splitting
strategy and the GHF relaying scheme with the fixed quantization level $\q_1 = \frac{\sqrt{\rho^2 + 16\rho +
16}-\rho}{4}$, in the
weak-relay regime of $|g_2| \le \sqrt{\rho}|h_{21}|$, the achievable rate region
in Corollary~\ref{achievable_theorem_single_digital_link} is within $\delta$ bits of the capacity region outer bound in Theorem~\ref{outerbound_theorem_twolinks} (with $\C_2$ set to zero), where $\delta$ is defined in Theorem~\ref{constantgap_theorem_twolinks}.
\end{theo}

\begin{IEEEproof}
Although the signalling scheme and the constant gap result resemble
those of Theorem~\ref{constantgap_theorem_twolinks},
Theorem~\ref{constantgap_theorem_single_digital_link} is not simply
obtained by setting $\C_2 =0$ in
Theorem~\ref{constantgap_theorem_twolinks}, since the weak-relay
condition has been relaxed. 
In the following, we prove the constant-gap result by directly
comparing each achievable rate expression with its corresponding upper bound.

Applying the inequalities of Lemma~\ref{lemma_useful_ineq} and following along the same lines of the proof of Theorem~\ref{constantgap_theorem_twolinks} in
Appendix~\ref{proof_constant_theorem_twolinks}, it is easy to show that each of the achievable rates in (\ref{R1_single_digital_link})-(\ref{R12R2_1_single_digital_link}) achieves to
within a constant gap of its corresponding upper bound in Theorem~\ref{outerbound_theorem_twolinks} (with $\C_2$ set to zero) in the weak-relay regime. The
constant gaps are shown as follows:

(i) Individual rate (\ref{R1_single_digital_link}) is within
\begin{equation}  \label{begin_gap_single_digital_link}
\delta_{R_1} = \max\left\{{\alpha}(\q_1), {\beta}(\q_1)  \right\}
\end{equation}
bits of (\ref{R1_bound_twolinks}).

(ii) Individual rate (\ref{R2_single_digital_link}) is within
\begin{equation}
\delta_{R_2} = \hf
\end{equation}
bits of (\ref{R2_bound_twolinks}).

(iii) Sum rates (\ref{sumrate_1_single_digital_link}),
(\ref{sumrate_2_single_digital_link}) and (\ref{sumrate_3_single_digital_link})
are within
\begin{equation}
\delta_{R_1 + R_2} = \hf + \max \left\{{\alpha}(\q_1), {\beta}(\q_1) \right\}
\end{equation}
bits of their upper bounds (\ref{sumrate_bound_1_twolinks}),
(\ref{sumrate_bound_8_twolinks}), (\ref{sumrate_bound_2_twolinks}),
(\ref{sumrate_bound_7_twolinks}), (\ref{sumrate_bound_3_twolinks}), and
(\ref{sumrate_bound_9_twolinks}). Specifically,

\begin{itemize}
    \item The first term of (\ref{sumrate_1_single_digital_link}) is within
    $\hf + {\beta}(\q_1)$ bits of (\ref{sumrate_bound_1_twolinks}). The second term is within
    $\hf + {\alpha}(\q_1)$ bits of (\ref{sumrate_bound_8_twolinks}).

    \item The first term of (\ref{sumrate_2_single_digital_link}) is within
    $\hf + {\beta}(\q_1)$ bits of (\ref{sumrate_bound_2_twolinks}). The second term is within
    $\hf + {\alpha}(\q_1)$ bits of (\ref{sumrate_bound_7_twolinks}).

    \item The first term of (\ref{sumrate_3_single_digital_link}) is within
    $\hf + {\beta}(\q_1)$ bits of (\ref{sumrate_bound_3_twolinks}). The second term is within
    $\hf + {\alpha}(\q_1)$ bits of (\ref{sumrate_bound_9_twolinks}).
\end{itemize}

Therefore, the achievable sum rate in (\ref{sumrate_1_single_digital_link})-(\ref{sumrate_3_single_digital_link}) is  within a constant gap of the sum-rate upper bound specified by (\ref{sumrate_bound_1_twolinks})-(\ref{sumrate_bound_12_twolinks}) in the weak-relay regime.

(iv) $2R_1+ R_2$ rate (\ref{2R1R2_1_single_digital_link}) is within
\begin{equation}
\delta_{2R_1 + R_2} = \hf + \max \left\{2{\alpha}(\q_1), {\alpha}(\q_1) +
{\beta}(\q_1), 2{\beta}(\q_1) \right\}
\end{equation}
bits of the upper bounds (\ref{2R1R2_bound_1_twolinks}),
(\ref{2R1R2_bound_6_twolinks}), and (\ref{2R1R2_bound_4_twolinks}).
Specifically, the first term of (\ref{2R1R2_1_single_digital_link}) is within
$\hf + 2{\beta}(\q_1)$ bits of (\ref{2R1R2_bound_1_twolinks}). The second term is
within $\hf + {\alpha}(\q_1) + {\beta}(\q_1)$ bits of
(\ref{2R1R2_bound_6_twolinks}). The third term is within $\hf + 2{\alpha}(\q_1)$
bits of (\ref{2R1R2_bound_4_twolinks}).

(v) $R_1 + 2R_2$ rate (\ref{R12R2_1_single_digital_link}) is within
\begin{equation} \label{end_gap_single_digital_link}
\delta_{R_1 + 2R_2} = 1 + \max \left\{{\alpha}(\q_1), {\beta}(\q_1)  \right\}
\end{equation}
bits of the upper bounds
\begin{eqnarray}
2R_1 + R_2 &\le& \hf \log \left(1 + \SNR_2 + \INR_1 \right) \nonumber \\
&& + \hf \log\left(1 + \INR_2 + \frac{\SNR_1}{1 + \INR_1} \right) \nonumber \\
&&+ \hf \log \left(1+ \frac{ \SNR_2}{1 + \INR_2}\right) +  \C_1 \\
2R_1 + R_2 &\le& \hf \log (1 + \SNR_2 + \INR_1)  \nonumber \\
&& + \hf \log\left(1 + \frac{\SNR_2}{1 + \INR_2 + \SNR_{r2}} \right) \nonumber \\
&& + \hf \log \left( 1 + \frac{\SNR_1(1 + \phi_1^2 \SNR_{r2}) + \SNR_{r1}}{1 + \INR_1} \right. \nonumber \\
&& \left. \qquad \qquad  + \INR_2 + \SNR_{r2} \right),
\end{eqnarray}
which are not shown explicitly in Theorem~\ref{outerbound_theorem_twolinks} but
can be obtained by switching the indices $1$ and $2$ of
(\ref{2R1R2_bound_1_twolinks}) and (\ref{2R1R2_bound_5_twolinks}) followed by
setting $\C_2 =0$.

Since ${\alpha}(\cdot)$ is an increasing function and ${\beta}(\cdot)$ is a
decreasing function, to minimize the gaps above, we need
\begin{equation}
{\alpha}(\q_1^*) = {\beta}(\q_1^*),
\end{equation}
which results in the quantization level $\q_1^* =  \frac{\sqrt{\rho^2 + 16\rho +
16}-\rho}{4}$.
With this optimal quantization level applied to the gaps above, we prove
that the achievable rate region
(\ref{R1_single_digital_link})-(\ref{R12R2_1_single_digital_link}) is within
\begin{eqnarray}
 \lefteqn{\max \left\{\hf ,  \hf \log\left(2 + \frac{\rho+\sqrt{\rho^2 + 16\rho + 16}}{2} \right) \right\}}  \nonumber \\
  &=&  \hf \log\left(2 + \frac{\rho+\sqrt{\rho^2 + 16\rho + 16}}{2} \right)
\end{eqnarray}
bits of the capacity region.
\end{IEEEproof}

\subsection{Generalized Degree of Freedom}

We now derive the GDoF of the channel depicted in
Fig.~\ref{channel_model_single_digital_link}, for the case where the
underlying interference channel is symmetric, i.e.,
$\INR_1 = \INR_2 = \INR$ and $\SNR_1 = \SNR_2 = \SNR$. In the high SNR regime,
define
\begin{eqnarray}
\beta_i &:=& \lim_{\SNR \rightarrow \infty}\frac{\log \SNR_{ri}}{\log \SNR}, \;\; i=1, 2, \\
\kappa_1 &:=& \lim_{\SNR \rightarrow \infty}\frac{\C_1}{\hf \log \SNR},
\end{eqnarray}
Applying Theorem~\ref{constantgap_theorem_single_digital_link}, we have
the following result on the GDoF:

\begin{corollary}
In the weak-relay regime where $\beta_2 \le \alpha$, the GDoF sum
capacity of the symmetric relay-interference channel with a single
digital link is given by the following.  For $0 \le \alpha <1$
\iftwocol
\begin{eqnarray} \label{d_sum}
d_{\mathrm{sum}} = \left\{
  \begin{array}{l}
\min \{2-\alpha, 2\max(\alpha, 1-\alpha)+\kappa_1, \max(\alpha, 1-\alpha)\\
\qquad + \max (\beta_1, 1+\beta_2-\alpha, \alpha) \},  \qquad \quad \beta_1 \le 1 \\
\min \{2 - \alpha + \kappa_1, 2\max(\alpha, 1-\alpha)+\kappa_1, \\
\qquad 1 + \beta_1 - \alpha \}, \qquad  \qquad  \qquad    \quad  \qquad \beta_1 \ge 1
  \end{array}
\right.
\end{eqnarray}
\else
\begin{eqnarray} \label{d_sum}
d_{\mathrm{sum}} = \left\{
  \begin{array}{l}
\min \{2-\alpha, 2\max(\alpha, 1-\alpha)+\kappa_1, \max(\alpha, 1-\alpha)  \nonumber \\
\qquad + \max (\beta_1, 1+\beta_2-\alpha, \alpha) \},  \qquad  \qquad \beta_1 \le 1 \\
\min \{2 - \alpha + \kappa_1, 2\max(\alpha, 1-\alpha)+\kappa_1,  \nonumber \\
\qquad  1 + \beta_1 - \alpha \},  \quad  \qquad \qquad  \qquad \qquad \quad  \beta_1 \ge 1
  \end{array}
\right.
\end{eqnarray}
and for $\alpha \ge 1$
\begin{equation}
d_{\mathrm{sum}} = \min \{\alpha, 2+\kappa \}.
\end{equation}
\end{corollary}

\begin{table*} [t]
    \center
    \caption{Sum-capacity GDoF gain due to the relay for the
	symmetric Gaussian relay-interference channel with a single
	digital relay link for $\alpha=\beta_1=\beta_2$}
	    \label{gdof_single_relay_typeII}
    \begin{tabular}{|c|c|c|c|c|c|c|}
    \hline
     Range of $\alpha$ & $\alpha \le \kappa$ &$\kappa \le \alpha \le \frac{2-\kappa}{3}$& $\frac{2-\kappa}{3} \le \alpha \le \frac{2}{3}$& $\frac{2}{3} \le \alpha \le 2$ & $2 \le \alpha \le 2+\kappa$ &$\alpha \ge 2 +\kappa$ \\ \hline
    Gain   & $\alpha$   &$\kappa$ &$2-3\alpha$ & $0$   &$\alpha -2$   & $\kappa$    \\ \hline
    \end{tabular}
\end{table*}

\begin{figure} [t]
\centering
\includegraphics[width=3.4in]{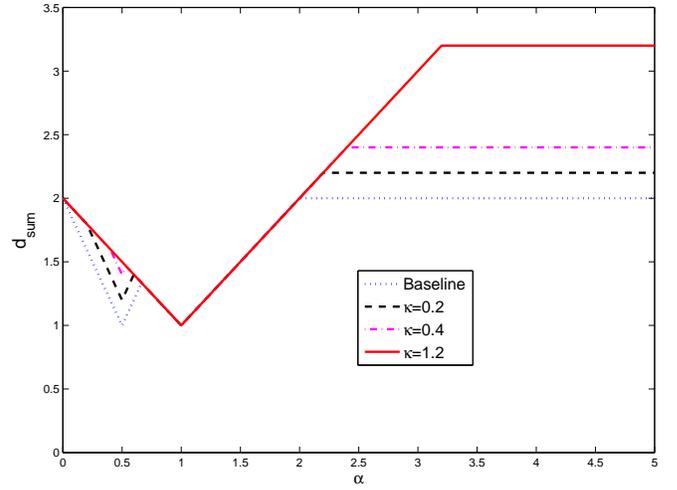}
\caption{Impact of the relay-destination link on sum capacity}
\label{sum_capacity_typeII}
\end{figure}

Table \ref{gdof_single_relay_typeII} and Fig.~\ref{sum_capacity_typeII} illustrate
the GDoF gain due to the relay where the direct links, the interference
links and the links to the relay are symmetric for both users, and where
$\alpha=\beta_1=\beta_2$. The main feature here is that there is no gain in
sum capacity for $\frac{2}{3} \le \alpha \le 2$. In other regimes of
$\alpha$, the sum-capacity gain is roughly one bit per relay bit.

\subsection{Signal Relaying vs. Interference Forwarding}

In the relay-interference channel, the relay observes a corrupted
version of the weighted sum of two source signals $X_1$ and $X_2$, and
forwards a description to the receiver. Intuitively, the observations
about both source signals are helpful. For the receiver $1$, the
observation about $X_1$ helps receiver $1$ reinforce the signal
intended for it; the observation about $X_2$ helps receiver $1$
mitigate the interference. The former can be thought of as signal
relaying, the latter interference forwarding.

In this section, we show that the sum-capacity gain in a Gaussian interference channel due to a single relay link is equivalent to that achievable with signal relaying alone or with interference forwarding alone, depending on the channel parameters. Toward this end, we first set the source-relay link from $X_2$ to zero,
i.e., $g_2 =0$, and compute the GDoF of the sum capacity. In this
case, the sum-capacity gain must be solely due to forwarding intended signal $X_1$.
Similarly, we can also set $g_1 =0$, and compute the GDoF of the
sum-capacity gain due solely to forwarding interference signal $X_2$.  By comparing
these rates we show that interestingly when the relay link of user $1$ is
under certain threshold, i.e., $\beta_1 \le 1-\alpha +\beta_2$, the
sum-capacity gain is equivalent to that achievable by interference forwarding. When
$\beta_1 \ge 1-\alpha +\beta_2$, the sum-capacity gain is equivalent to that achievable by signal relaying.

More specifically, with $g_2=0$, the sum-capacity
can be computed as
\iftwocol
\begin{eqnarray} \label{d_SR}
d_{SR} = \left\{
  \begin{array}{l}
\min \{2-\alpha, 2\max(\alpha, 1-\alpha)+\kappa_1,  \\
\qquad \max(\alpha, 1-\alpha) +\max(\beta_1,1-\alpha, \alpha) \},  \beta_1 \le 1 \\
\min \{2 - \alpha + \kappa_1, 2\max(\alpha, 1-\alpha)+\kappa_1,  \\
\qquad 1 + \beta_1 - \alpha \}, \qquad \qquad \qquad \qquad \quad \;   \beta_1
\ge 1
  \end{array}
\right.
\end{eqnarray}
\else
\begin{eqnarray} \label{d_SR}
d_{SR} = \left\{
  \begin{array}{l}
\min \{2-\alpha, 2\max(\alpha, 1-\alpha)+\kappa_1, \max(\alpha, 1-\alpha) \nonumber \\
\qquad + \max(\beta_1,1-\alpha, \alpha) \}, \qquad  \qquad \qquad  \;\; \beta_1 \le 1 \\
\min \{2 - \alpha + \kappa_1, 2\max(\alpha, 1-\alpha)+\kappa_1, \nonumber \\
\qquad \;1 + \beta_1 -
\alpha \}. \qquad \qquad \qquad \qquad \qquad   \beta_1 \ge 1
  \end{array}
\right.
\end{eqnarray}
Similarly, let $g_1=0$. The sum-capacity GDoF obtained by forwarding interference
signal is
\begin{eqnarray} \label{d_IF}
d_{IF} &=& \min \{2-\alpha, 2\max (\alpha, 1-\alpha)+\kappa_1,
  \max(\alpha, 1-\alpha)  \nonumber \\
  && \qquad \;+ \max(1+ \beta_2 -\alpha, \alpha) \}.
\end{eqnarray}
Comparing (\ref{d_sum}), (\ref{d_SR}), and (\ref{d_IF}), it is easy to
verify that
\begin{eqnarray}
\label{d_compare}
d_{sum} = \left\{
  \begin{array}{ll}
    d_{IF} & {\rm \ when\ \ } \beta_1 \le 1 + \beta_2 - \alpha \\
    d_{SR} & {\rm \ when\ \ } \beta_1 \ge 1 + \beta_2 - \alpha
  \end{array}
\right..
\end{eqnarray}
Therefore, we observe the following threshold effects. When the relay link from user $1$ is weak, the sum-capacity gain is
equivalent to a channel with a single source-relay link from $X_2$. As the source-relay link from
$X_1$ grows stronger and crosses a threshold $\beta_1 \ge 1 + \beta_2 -
\alpha \triangleq \beta_1^*$,  the
sum-capacity gain becomes equivalent to that of a single source-relay link from
$X_1$. Note that this is a GDoF phenomenon in the high SNR regime. In the
general SNR regime, the sum-capacity gain contains contributions from both
signal relaying and interference forwarding.

\begin{figure} [t]
\centering \psfrag{dsum}{$d_{\mathrm{sum}}$} \psfrag{1.5}{$1.5$} \psfrag{one}{$1$} \psfrag{zero}{$0$}
\psfrag{Interference}{$\mathrm{interference}$} \psfrag{Forwarding}{$\mathrm{forwarding}$} \psfrag{Signal}{$\mathrm{signal}$}  \psfrag{relaying}{$\mathrm{relaying}$} \psfrag{0.7}{$0.7$}\psfrag{beta1}{$\beta_1$}
\psfrag{capacity}{$\mathrm{capacity}$} \psfrag{baseline}{$\mathrm{baseline}$} \psfrag{R1}{$\mathcal{R}_2$}\psfrag{R2}{$\mathcal{R}_2$}
\includegraphics[width=3.4in]{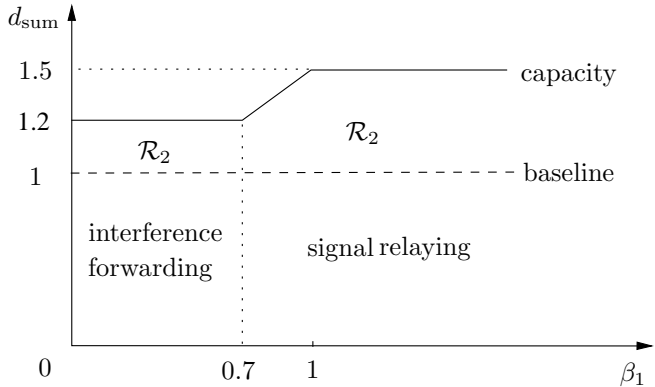}
\caption{Signal relaying vs. interference forwarding} \label{gdof}
\end{figure}

To visualize the interaction of signal relaying and interference
forwarding, a numerical example is provided in Fig.~\ref{gdof}. The
channel parameters are set to $\alpha=0.5$, $\beta_2 =0.2$, and
$\kappa_1=0.5$. The GDoF of the sum capacity
is plotted as a function of $\beta_1$. The sum
capacity of the interference channel without the relay serves as the
baseline: \begin{equation}
d_{BL} = \min\left \{2-\alpha, 2\max (\alpha, 1-\alpha) \right \}.
\label{d_BL}
\end{equation}
Fig.~\ref{gdof} shows the sum-capacity gain due to the relay.  When
$\beta_1 \le \beta_1^*=0.7$, the gain (labeled as $\mathcal{R}_1$) is
equivalent to that by forwarding interference signal only.  When $\beta \ge 0.7$,
the gain (labeled as $\mathcal{R}_2$) is equivalent to that by forwarding intended signal
only.

\section{Conclusion}
This paper investigates GHF as an incremental relay strategy for a
Gaussian interference channel augmented with an out-of-band
broadcasting relay, in which the relay message to one receiver is a
degraded version of the message to the other receiver.
We focus on a weak-relay regime, where the
transmitter-to-relay links are not unboundedly stronger than the
interfering links of the interference channel, and show that GHF
achieves to within a constant gap to the capacity region in the weak-relay regime.
Further, in a symmetric setting, each common relay bit can be worth
either one or two bits in the sum capacity gain, illustrating the
potential for a cell-edge relay in improving the system throughput of
a wireless cellular network.

Furthermore, the Gaussian relay-interference channels with a single
relay link is also studied. The capacity region is characterized to
within a constant gap for a larger range of channel parameters. It is
shown that in the high SNR regime, the sum-capacity improvement is equivalent either to that of a single source-relay link from user $1$ or  that of a single source-relay link from user $2$.

\appendix

\subsection{Proof of Theorem~\ref{outerbound_theorem_twolinks}} \label{proof_outer_bound_twolinks}
Define $V_1^n$ as the output of the digital link $\C_1$, and
$V_2^n$ as the output of the digital link $\C_2$. The outer bounds are
proved as follows:

(i) Individual-rate bounds: First, the first term of (\ref{R1_bound_twolinks}) is the simple cut-set upper bound for $R_1$. For the second term, starting from Fano's inequality, we have
    \begin{eqnarray}
    n(R_1 - \epsilon_n) &\le& I(X_1^n; Y_1^n, V_1^n)  \\
    &\le& I(X_1^n; Y_1^n, Y_R^n, X_2^n) \nonumber \\
    &\le& \frac{n}{2} \log(1 + \SNR_1 + \SNR_{r1}).\nonumber
    \end{eqnarray}
    The outer bound of $R_2$ in (\ref{R2_bound_twolinks}) can be proved in the same way.

(ii) Sum-rate bounds:
\begin{itemize}
\item First, (\ref{sumrate_bound_1_twolinks})-(\ref{sumrate_bound_3_twolinks}) are obtained from Fano's inequalities, i.e.,
\begin{eqnarray}
\lefteqn{n(R_1+R_2 -\epsilon_n)}  \\
 &\le& I(X_1^n; Y_1^n, V_1^n) + I(X_2^n; Y_2^n, V_2^n) \nonumber \\
&=& I(X_1^n; Y_1^n) + I(X_2^n; Y_2^n)  + I(X_1^n; V_1^n|Y_1^n) \nonumber \\
&& +  I(X_2^n; V_2^n|Y_2^n) \nonumber \\
&\le&  I(X_1^n; Y_1^n) + I(X_2^n; Y_2^n)  + h(V_1^n) + h(V_2^n) \nonumber \\
&\le& nC_{sum}(0) + n\C_1 + n\C_2,\nonumber
\end{eqnarray}
where $C_{sum}(0)$ is the sum capacity of the interference channel without relay. Clearly, the sum-rate gain due to the digital relay is upper bounded by the rates of digital links. Although the sum-rate capacity $C_{sum}(0)$ is not known in general, its upper bound has been studied in literature \cite{Tse2007,Wang_Allerton,Biao,Khandani08,VVV, Telatar_bound}. Applying the sum-rate outer bounds in \cite{Wang_Allerton}, we obtain (\ref{sumrate_bound_1_twolinks})-(\ref{sumrate_bound_3_twolinks}).

\item Second,  (\ref{sumrate_bound_4_twolinks})-(\ref{sumrate_bound_6_twolinks}) can be obtained by the following steps:
\begin{eqnarray}
\lefteqn{n(R_1+R_2 -\epsilon_n)}  \\
 &\le& I(X_1^n; Y_1^n, V_1^n) + I(X_2^n; Y_2^n, V_2^n) \nonumber \\
&\stackrel{(a)}{\le}& I(X_1^n; Y_1^n) + h(V_1^n) + I(X_2^n; Y_2^n, Y_R^n),\nonumber
\end{eqnarray}
where in (a) we give genie $Y_R^n$ to receiver $2$ and apply the fact that $\hat{Y}_R$ is a function of $Y_R$. Note that $I(X_1^n; Y_1^n) + I(X_2^n; Y_2^n, Y_R^n)$ is upper bounded by the sum capacity of the SIMO interference channel with $X_1^n$ and $X_2^n$ as the input, and $Y_1^n$ and $(Y_2^n, Y_R^n)$ as the output. The sum-rate outer bound of such a SIMO interference channel has been studied in \cite{Wang_Allerton},
which along with $h(V_1^n) \le n\C_1$ gives the outer bounds of
(\ref{sumrate_bound_4_twolinks})-(\ref{sumrate_bound_6_twolinks}).

\item Third,
(\ref{sumrate_bound_7_twolinks})-(\ref{sumrate_bound_9_twolinks}) can be
similarly derived following the same steps of
(\ref{sumrate_bound_4_twolinks})-(\ref{sumrate_bound_6_twolinks}) with indices
$1$ and $2$ switched.

\item Fourth, (\ref{sumrate_bound_10_twolinks})-(\ref{sumrate_bound_12_twolinks}) can
be obtained by giving $Y_R^n$ as a genie to both receivers, i.e.,
\begin{eqnarray}
\lefteqn{n(R_1+R_2 -\epsilon_n) }  \\
&\le& I(X_1^n; Y_1^n, V_1^n) + I(X_2^n; Y_2^n, V_2^n) \nonumber \\
&\le& I(X_1^n; Y_1^n, Y_R^n) +I(X_2^n; Y_2^n, Y_R^n),  \nonumber
\end{eqnarray}
which is upper bounded by the sum capacity of the SIMO interference
channel with $X_1^n$ and $X_2^n$ as input, and $(Y_1^n, Y_R^n)$ and $(Y_2^n,
Y_R^n)$ as output. Applying the result in \cite{Wang_Allerton}, we have
(\ref{sumrate_bound_10_twolinks})-(\ref{sumrate_bound_12_twolinks}).
\end{itemize}

(iii) $2R_1 + R_2$ bounds: Six upper bounds on $2R_1 + R_2$.
\begin{itemize}
\item First,
(\ref{2R1R2_bound_1_twolinks}) is simply the cut-set bound, i.e.,
\begin{eqnarray}
\lefteqn{n(2R_1 + R_2 - \epsilon_n)}  \\
&\le&  2I(X_1^n; Y_1^n, V_1^n) + I(X_2^n; Y_2^n, V_2^n) \nonumber \\
&\le& 2I(X_1^n; Y_1^n) + I(X_2^n; Y_2^n) + 2h(V_1^n) + h(V_2^n), \nonumber
\end{eqnarray}
where $2I(X_1^n; Y_1^n) + I(X_2^n; Y_2^n)$ is upper bounded by the $2R_1 + R_2$
bound of the interference channel with $X_1^n$ and $X_2^n$ as the input, and $Y_1^n$ and $Y_2^n$ as the output, which together with $h(V_1^n) \le n\C_1$ and $h(V_2^n) \le n\C_2$ gives the upper bound in
(\ref{2R1R2_bound_1_twolinks}).

\item Second, (\ref{2R1R2_bound_2_twolinks}) can be derived by giving genie $Y_R^n$
to both receivers:
\begin{eqnarray}
\lefteqn{n(2R_1 + R_2 - \epsilon_n)} \\
&\le&  2I(X_1^n; Y_1^n, V_1^n) + I(X_2^n; Y_2^n, V_2^n) \nonumber \\
&\le& 2I(X_2^n; Y_1^n, Y_R^n) + I(X_2^n; Y_2^n, Y_R^n),\nonumber
\end{eqnarray}
which is upper bounded by the $2R_1+R_2$ bound of the SIMO interference channel
with $X_1^n$ and $X_2^n$ as the input, and $(Y_1^n, Y_R^n)$ and $(Y_2^n,
Y_R^n)$ as the output. Applying the result of \cite{Wang_Allerton}, we obtain
(\ref{2R1R2_bound_2_twolinks}).

\item Third, (\ref{2R1R2_bound_3_twolinks}) can be obtained by giving genies $(X_2^n, Y_R^n, S_1^n)$ to $Y_1^n$ in one of the two $R_1$ expressions and $(S_2^n, Y_R^n)$ to $Y_2^n$, where genies $S_1^n$ and $S_2^n$ are defined as
    \begin{equation}
    S_1^n = h_{12}X_1^n + Z_2,  \qquad S_2^n = h_{21}X_2^n + Z_1.
    \end{equation}
    According to Fano's inequality, we have
\begin{eqnarray}
\lefteqn{n(2R_1 + R_2 - \epsilon_n)} \label{2R1R2_1_010}  \\
&\le&  2I(X_1^n; Y_1^n, V_1^n) + I(X_2^n; Y_2^n, V_2^n) \nonumber \\
&\le& I(X_1^n; Y_1^n, Y_R^n, S_1^n, X_2^n) + I(X_1^n; Y_1^n) + h(V_1^n) \nonumber \\
&&+ I(X_2^n; Y_2, Y_R^n, S_2^n) \nonumber \\
&\stackrel{(a)}{\le}& I(X_1^n;Y_1^n, Y_R^n, S_1^n|X_2^n) + h(Y_1^n) - h(S_2^n) + n\C_1 \nonumber \\
&&+ I(X_2^n; S_2^n) + I(X_2^n;Y_2^n, Y_R^n|S_2^n) \nonumber \\
&=& I(X_1^n; S_1^n) + I(X_1^n; Y_1^n, Y_R^n|S_1^n, X_2^n) + h(Y_1^n)  \nonumber \\
&& - h(S_2^n) + n\C_1 + h(S_2^n) - h(Z_1^n) \nonumber \\
&& + h(Y_2^n, Y_R^n|S_2^n) - h(S_1^n) -  h(Y_R^n|Y_2^n, X_2^n)\nonumber \\
&=& h(Y_1^n) - h(Z_1^n) + h(Y_1^n, Y_R^n|S_1^n, X_2^n) + n\C_1 \nonumber \\
&& - h(Z_1^n, Z_R^n) + h(Y_2^n, Y_R^n|S_2^n) - h(Z_2^n, Z_R^n) \nonumber \\
&& - I(Y_R^n; X_1^n|X_2^n, Y_2^n) \nonumber \\
&\le&h(Y_1^n) - h(Z_1^n) + h(Y_1^n, Y_R^n|S_1^n, X_2^n) + n\C_1 \nonumber \\
&& - h(Z_1^n, Z_R^n) + h(Y_2^n, Y_R^n|S_2^n) - h(Z_2^n, Z_R^n) , \nonumber
\end{eqnarray}
where in (a) we use the fact that $X_1^n$ is independent  with $X_2^n$. Note
that, the last inequality of (\ref{2R1R2_1_010}) is maximized by Gaussian inputs $X_1^n$ and $X_2^n$
with i.i.d $\mathcal{N} (0, 1)$ entries, because
\begin{itemize}
\item $h(Y_1^n)$ is maximized by Gaussian distributions, and
\item $h(Y_1^n, Y_R^n|S_1^n, X_2^n)$ and $h(Y_2^n, Y_R^n|S_2^n)$ are both maximized by Gaussian
inputs since the conditional entropy under a power constraint is maximized by
Gaussian distributions.
\end{itemize}
Applying Gaussian distributions to the last inequality of (\ref{2R1R2_1_010}),
we have (\ref{2R1R2_bound_3_twolinks}).

\item Fourth, (\ref{2R1R2_bound_4_twolinks}) can be obtained by giving genie $Y_R^n$ to $Y_1^n$, i.e.,
\begin{eqnarray}
\lefteqn{n(2R_1 + R_2 - \epsilon_n)} \\
&\le&  2I(X_1^n; Y_1^n, V_1^n) + I(X_2^n; Y_2^n, V_2^n) \nonumber \\
&\le& 2I(X_1^n; Y_1^n, Y_R^n) + I(X_2^n; Y_2^n) + h(V_2^n),\nonumber
\end{eqnarray}
where $2I(X_1^n; Y_1^n, Y_R^n) + I(X_2^n; Y_2^n)$ is upper bounded by the $2R_1 + R_2$ bound of the SIMO interference channel with $X_1^n$ and $X_2^n$ as the input, and $(Y_1^n, Y_R^n)$ and $Y_2^n$ as the output. Applying the result of \cite{Wang_Allerton} and the fact that $h(V_2^n) \le n\C_2$, we obtain (\ref{2R1R2_bound_4_twolinks}).

\item Fifth, (\ref{2R1R2_bound_5_twolinks}) can be obtained by giving genie $Y_R^n$ to $Y_2^n$, i.e.,
\begin{eqnarray}
\lefteqn{n(2R_1 + R_2 - \epsilon_n)} \\
&\le&  2I(X_1^n; Y_1^n, V_1^n) + I(X_2^n; Y_2^n, V_2^n) \nonumber \\
&\le& 2I(X_1^n; Y_1^n) + 2 h(V_1^n) + I(X_2^n; Y_2^n, Y_R^n),\nonumber
\end{eqnarray}
where $2I(X_1^n; Y_1^n) + I(X_2^n; Y_2^n, Y_R^n)$ is upper bounded by the $2R_1 + R_2$ bound of the SIMO interference channel with $X_1^n$ and $X_2^n$ as the input, and $Y_1^n$ and $(Y_2^n, Y_R^n)$ as the output. Applying the result of \cite{Wang_Allerton} and the fact that $h(V_1^n) \le n\C_1$, we obtain (\ref{2R1R2_bound_5_twolinks}).

\item Sixth, (\ref{2R1R2_bound_6_twolinks}) can be obtained by giving genies $(X_2^n, Y_R^n, S_1^n)$ to $Y_1^n$ in one of the two $R_1$ expressions, and $S_2^n$ to $Y_2^n$, i.e.,
\begin{eqnarray}
\lefteqn{n(2R_1 + R_2 - \epsilon_n)} \label{2R1R2_1_proof_6} \\
&\le&  2I(X_1^n; Y_1^n, V_1^n) + I(X_2^n; Y_2^n, V_2^n) \nonumber \\
&\le& I(X_1^n; Y_1^n, Y_R^n, S_1^n, X_2^n) + I(X_1^n; Y_1^n) + h(V_1^n) \nonumber \\
&&+ I(X_2^n; Y_2^n, S_2^n) + h(V_2^n) \nonumber \\
&\le& I(X_1^n; S_1^n) + I(X_1^n; Y_1^n Y_R^n| S_1^n,X_2^n) + h(Y_1^n) \nonumber \\
&&- h(S_2^n) + I(X_2^n; S_2^n) + I(X_2^n; Y_2^n|S_2^n) \nonumber \\
&& + n\C_1 + n \C_2 \nonumber \\
&\le& h(S_1^n) - h(Z_2^n) + h(Y_1^n, Y_R^n|S_1^n, X_2^n) \nonumber \\
&& - h(Z_1^n, Z_R^n) + h(Y_1^n) - h(S_2^n) + h(S_2^n) - h(Z_1^n) \nonumber \\
&&+ h(Y_2^n|S_2^n)  - h(S_1^n) + n\C_1 + n\C_2 \nonumber \\
&=& h(Y_1^n) - h(Z_1^n) + h(Y_1^n, Y_R^n|S_1^n, X_2^n) \nonumber \\
&&- h(Z_1^n, Z_R^n) +  h(Y_2^n|S_2^n) - h(Z_2^n) + n\C_1 + n\C_2, \nonumber
\end{eqnarray}
which is maximized by Gaussian distributions of $X_1^n$ and $X_2^n$ with i.i.d entries following $\mathcal{N}(0,1)$. Applying Gaussian distributions to (\ref{2R1R2_1_proof_6}), we obtain  (\ref{2R1R2_bound_6_twolinks}).
\end{itemize}

\subsection{Useful Inequalities} \label{useful_ineq}
This appendix provides several inequalities that are useful to prove the  constant-gap theorems.

\begin{lemma} \label{lemma_useful_ineq}
For $ \Delta{a}_i, a_i, \Delta{d}_i, d_i, \Delta{e}_i, e_i, \Delta{g}_i, g_i$
and $\xi_i, i=1, 2$ as defined in (\ref{a2g_begin})-(\ref{a2g_end}), with $Q$ set
as a constant, when $W_i, X_i$ are generated from a superposition coding of
$X_i=U_i + W_i$ with $U_i \sim ~\mathcal{N}(0, P_{ip})$ and $W_i \sim
\mathcal{N}(0, P_{ic})$, where $P_{ip} + P_{ic}=1$ and $P_{1p}=\min\{1,
h_{12}^{-2}\}, P_{2p} = \min\{1, h_{21}^{-2}\}$, and when the GHF quantization
variables are set to $\hat{Y}_{R1} = Y_R + e_1, \hat{Y}_{R2} = Y_R + e_2$,
where $e_1 \sim \mathcal{N}(0, \q_1)$ and $e_2 \sim \mathcal{N}(0, \q_2)$, in
the weak-relay regime of $|g_1| \le \sqrt{\rho}|h_{12}|, |g_2| \le
\sqrt{\rho}|h_{21}|$, the mutual information terms in
(\ref{a2g_begin})-(\ref{a2g_end}) can be bounded as follows:
\begin{eqnarray}
a_1 &\ge& \hf \log \left( 1 + \frac{\SNR_1}{1 + \INR_1}\right) - \hf,  \label{lower_bound_a_1} \\
a_1 + \Delta {a}_1 &\ge& \hf \log \left(1 + \frac{\SNR_1 + \SNR_{r1}}{1 + \INR_1} \right) - \alpha(\q_1), \label{lower_bound_a_1_da1}\\
d_1 &\ge& \hf \log(1 + \SNR_1) -\hf, \label{lower_bound_d_1}\\
d_1 + \Delta {d}_1 &\ge& \hf \log (1 + \SNR_1 + \SNR_{r1}) - \alpha(\q_1), \label{lower_bound_d_1_dd1}\\
e_1 &\ge& \hf \log \left(1 + \frac{\SNR_1}{1 + \INR_1} + \INR_2\right) -\hf, \label{lower_bound_e_1}\\
e_1 + \Delta {e}_1 &\ge& \hf \log \left(1 + \frac{\SNR_1(1 + \phi_1^2 \SNR_{r2}) + \SNR_{r1}}{1 + \INR_1} \right. \nonumber \\
&& \left.  \qquad \;\;+ \INR_2 + \SNR_{r2} \right) - \alpha(\q_1), \label{lower_bound_e_1_de1}\\
g_1 &\ge& \hf \log(1 + \SNR_1 + \INR_2) - \hf \label{lower_bound_g_1} \\
g_1 + \Delta {g}_1 &\ge& \hf \log \left( 1+ \SNR_1(1 + \phi_1^2 \SNR_{r2}) + \SNR_{r1}  \right. \nonumber \\
&& \left.  \qquad \;\;+ \INR_2 + \SNR_{r2}\right) - \alpha(\q_1) \label{lower_bound_g_1_dg1}, \\
\xi_1 &\le& \hf \log\left(1 + \frac{1+\rho}{\q_1} \right) = \beta(\q_1) - \hf,
\label{lower_bound_xi_1}
\end{eqnarray}
and the lower bounds of $a_2, a_2 + \Delta {a}_2, d_2, d_2 + \Delta
{d}_2, e_2, e_2 + \Delta {e}_2, g_2, g_2 + \Delta {g}_2$ and
the upper bound of $\xi_2$ can be obtained by switching the indices of $1$ and
$2$ in (\ref{lower_bound_a_1})-(\ref{lower_bound_xi_1}).
\end{lemma}

\begin{IEEEproof}
First, define the signal-to-noise and interference-to-noise ratios of the
private messages as
\begin{align}
\SNR_{1p} &= |h_{11}|^2 P_{1p}, \qquad  \SNR_{2p} &= |h_{22}|^2 P_{2p},  \\
\INR_{1p} &= |h_{12}|^2 P_{1p},\qquad   \INR_{2p} &= |h_{21}|^2 P_{2p},  \\
 \SNR_{r1p} &= |g_1|^2 P_{1p}, \qquad  \SNR_{r2p} &= |g_2|^2 P_{2p},
\end{align}
which can be lower bounded or upper bounded as follows:
\begin{eqnarray}
\SNR_{1p} &=& |h_{11}|^2 P_{1p} \nonumber \\
&=& \min \left\{|h_{11}|^2, \frac{|h_{11}|^2}{|h_{12}|^2}  \right\} \nonumber \\
&=& \min\left\{\SNR_1, \frac{\SNR_1}{\INR_1}  \right\} \nonumber \\
&\ge&  \frac{\SNR_1}{1 + \INR_1},
\end{eqnarray}
and
\begin{eqnarray}
0 \le \INR_{1p} = \min\{1, \INR_1 \} \le 1,
\end{eqnarray}
and
\begin{eqnarray}
\SNR_{r1p} &=& |g_{1}|^2 P_{1p} \nonumber \\
&=& \min \left\{|g_1|^2, \frac{|g_1|^2}{|h_{12}|^2} \right\} \nonumber \\
&=& \min\left\{\SNR_{r1}, \frac{\SNR_{r1}}{\INR_1}  \right\} \nonumber \\
&\ge&  \frac{\SNR_{r1}}{1 + \INR_1}.
\end{eqnarray}
Since $|g_1| \le \sqrt{\rho}|h_{12}|$, $\SNR_{r1p}$ is upper bounded by
$\rho$. Therefore
\begin{equation}
\rho \ge \SNR_{r1p} \ge \frac{\SNR_{r1}}{1 + \INR_1}.
\end{equation}

Switching the indices of $1$ and $2$, we have
\begin{align}
\SNR_{2p} &\ge  \frac{\SNR_2}{1 + \INR_2}, \\
1 \ge \INR_{2p} &\ge 0, \\
\rho \ge \SNR_{r2p} &\ge \frac{\SNR_{r2}}{1 + \INR_2}.
\end{align}

Now, starting from (\ref{lower_bound_a_1}), we prove the inequalities one by
one.
\begin{itemize}
\item First, (\ref{lower_bound_a_1}) is lower bounded by
\begin{eqnarray}
a_1 &=& I(X_1; Y_1|W_1, W_2) \nonumber \\
&=& \hf \log \left(\frac{1 + \SNR_{1p} + \INR_{2p}}{1 + \INR_{2p}} \right)
\nonumber \\
&\stackrel{(a)}{\ge}&  \hf \log(1 + \SNR_{1p}) -\hf \nonumber \\
&\stackrel{(b)}{\ge}&  \hf \log \left( 1 + \frac{\SNR_1}{1 + \INR_1}\right) -
\hf,
\end{eqnarray}
where (a) holds because $0 \le \INR_{2p} \le 1$ and (b) is due to the fact that
$\SNR_{1p} \ge  \frac{\SNR_1}{1 + \INR_1}$.

\item Second, (\ref{lower_bound_a_1_da1}) is lower bounded by (\ref{eq:a1_delta_a1}),
\begin{figure*}
\begin{eqnarray}
a_1 + \Delta {a}_1
 &=& I(X_1; Y_1|W_1, W_2) + I(X_1; \hat{Y}_{R1}|Y_1,
W_1, W_2) \nonumber \\
&=& \hf \log \left( \frac{(\q_1 + 1)(1 + \SNR_{1p} + \INR_{2p}) + \SNR_{r1p} +
\SNR_{r2p}(1 + \phi_1^2\SNR_{1p})}{(\q_1 + 1)(1 + \INR_{2p}) + \SNR_{r2p}}
\right) \nonumber \\
&\ge& \hf \log (1 + \SNR_{1p} + \SNR_{r1p}) - \hf \log ((\q_1 + 1)(1 +
\INR_{2p}) + \SNR_{r2p}) \nonumber \\
&\stackrel{(a)}{\ge}& \hf \log \left(1 + \frac{\SNR_1 + \SNR_{r1}}{1 + \INR_1}
\right) - \alpha(\q_1), \label{eq:a1_delta_a1}
\end{eqnarray}
\hrule
\end{figure*}
where (a) holds because $\SNR_{1p} \ge  \frac{\SNR_1}{1 + \INR_1}$, $\SNR_{r1p}
\ge \frac{\SNR_{r1}}{1 + \INR_1}$, and
\begin{eqnarray}
\lefteqn{\hf \log((\q_1 + 1)(1 + \INR_{2p}) + \SNR_{r2p})} \nonumber \\
&\le& \hf \log((\q_1+1)(1+1) + \rho) \nonumber \\
&=& \alpha(\q_1).
\end{eqnarray}

\item Third, (\ref{lower_bound_d_1}) is lower bounded by
\begin{eqnarray}
d_1 &=& I(X_1; Y_1|W_{2}) \nonumber \\
&=& \hf \log \left(\frac{1 + \SNR_1 + \INR_{2p}}{1 + \INR_{2p}} \right)
\nonumber \\
&\ge& \hf \log(1 + \SNR_1) -\hf.
\end{eqnarray}

\item Fourth, (\ref{lower_bound_d_1_dd1}) is lower bounded by (\ref{eq:d1_delta_d1}).
\begin{figure*}
\begin{eqnarray}
d_1 + \Delta {d}_1
 &=& I(X_1;Y_1|W_{2}) + I(X_1; \hat{Y}_{R1}|Y_1, W_2)
\nonumber \\
&=& \hf \log\left(\frac{(\q_1 +1)(1 + \SNR_1 + \INR_{2p}) + \SNR_{r1} +
\SNR_{r2p}(1 + \phi_1^2 \SNR_1)} {(\q_1 + 1)(1 + \INR_{2p}) +
\SNR_{r2p}} \right) \nonumber \\
&\ge& \hf \log (1 + \SNR_1 + \SNR_{r1}) -\alpha(\q_1). \label{eq:d1_delta_d1}
\end{eqnarray}
\end{figure*}

\item Fifth, (\ref{lower_bound_e_1}) is lower bounded by
\begin{eqnarray}
e_1 &=& I(X_1, W_2; Y_1|W_1) \nonumber \\
&=& \hf \log \left( \frac{1 + \SNR_{1p} + \INR_2}{1 + \INR_{2p}}\right)
\nonumber \\
&\ge& \hf \log \left(1 + \frac{\SNR_1}{1 + \INR_1} + \INR_2\right) -\hf.
\end{eqnarray}

\item Sixth, (\ref{lower_bound_e_1_de1}) is lower bounded by (\ref{eq:e1_delta_e1}).
\begin{figure*}[t]
\begin{eqnarray}
e_1 + \Delta {e}_1&=& I(X_1, W_2; Y_1|W_1) + I(X_1,W_2; \hat{Y}_{R1}
|Y_1, W_1) \nonumber \\
&=& \hf \log \left( \frac{(\q_1 + 1)(1 + \SNR_{1p} + \INR_{2}) + \SNR_{r1p} +
\SNR_{r2}(1 + \phi_1^2\SNR_{1p})}{(\q_1 + 1)(1 + \INR_{2p}) + \SNR_{r2p}}
\right) \nonumber \\
&\ge& \hf \log \left(1 + \frac{\SNR_1(1 + \phi_1^2 \SNR_{r2}) + \SNR_{r1}}{1 +
\INR_1} + \INR_2 + \SNR_{r2} \right) - \alpha(\q_1). \label{eq:e1_delta_e1}
\end{eqnarray}
\end{figure*}

\item Seventh, (\ref{lower_bound_g_1}) is lower bounded by
\begin{eqnarray}
g_1 &=& I(X_1, W_2; Y_1) \nonumber \\
&=& \hf \log \left(\frac{1 + \SNR_1 + \INR_2}{1 + \INR_{2p}} \right) \nonumber
\\
&\ge& \hf \log(1 + \SNR_1 + \INR_2) - \hf.
\end{eqnarray}

\item Eighth, (\ref{lower_bound_g_1_dg1}) is lower bounded by (\ref{eq:g1_delta_g1}).
\begin{figure*}
\begin{eqnarray}
g_1 + \Delta {g}_1
 &=& I(X_1, W_2; Y_1) + I(X_1, W_2; \hat{Y}_{R1}|Y_1)
\nonumber \\
&=& \hf \log \left( \frac{(\q_1 + 1)(1 + \SNR_{1} + \INR_{2}) + \SNR_{r1} +
\SNR_{r2}(1 + \phi_1^2\SNR_{1})}{(\q_1 + 1)(1 + \INR_{2p}) + \SNR_{r2p}}
\right) \nonumber \\
&\ge& \hf \log \left( 1+ \SNR_1(1 + \phi_1^2 \SNR_{r2}) + \SNR_{r1}  + \INR_2 +
\SNR_{r2}\right) -\alpha(\q_1). \label{eq:g1_delta_g1}
\end{eqnarray}
\hrule
\end{figure*}

\item Ninth, (\ref{lower_bound_xi_1}) is upper bounded by
\begin{eqnarray}
\xi_1 &=& I(Y_{R}: \hat{Y}_{R1}|Y_1, X_1, W_2) \nonumber \\
&=& \hf \log \left(1 + \frac{1}{\q_1}\left(1 + \frac{\SNR_{r2p}}{1 + \INR_{2p}}
\right) \right) \nonumber \\
&\le& \hf \log\left(1 + \frac{1 +\rho}{\q_1} \right)
\end{eqnarray}
\end{itemize}
\end{IEEEproof}

\subsection{Proof of Theorem~\ref{constantgap_theorem_twolinks}}
\label{proof_constant_theorem_twolinks}

In this appendix, we show that using the Han-Kobayashi power splitting
strategy with the private message power set to $P_{1p}=\min\{1, h_{12}^{-2}\}$
and $P_{2p} = \min \{1, h_{21}^{-2} \}$, all the achievable rates in (\ref{achievable_R1})-(\ref{achievable_R12R2})
are within constant bits of their corresponding outer bounds in
Theorem~\ref{outerbound_theorem_twolinks}. Note that, in the following proof,
inequalities in Appendix~\ref{useful_ineq} are implicitly used without being
mentioned.

(i) First, (\ref{achievable_R1}) is within constant bits of (\ref{R1_bound_twolinks}), and (\ref{achievable_R2}) is within constant bits of (\ref{R2_bound_twolinks}). To see this, the first term of (\ref{achievable_R1}) is lower bounded by
\begin{eqnarray}
\lefteqn{d_1 + (\C_1 -\xi_1)^+}  \nonumber \\
&\ge& \hf \log(1 + \SNR_1) - \hf + \C_1 - \xi_1 \nonumber \\
&\ge& \hf \log(1 + \SNR_1) + \C_1 - \left(\hf + \hf \log\left(1 + \frac{1
+\rho}{\q_1}\right) \right) \nonumber \\
&&
\end{eqnarray}
which is within $\beta(\q_1)$ bits of the
first term of (\ref{R1_bound_twolinks}).

According to Lemma~\ref{lemma_useful_ineq}, the second term of
(\ref{achievable_R1}) is lower bounded by
\begin{eqnarray}
d_1 + \Delta {d}_1 &\ge& \hf \log (1 + \SNR_1 + \SNR_{r1}) - \alpha(\q_1), \nonumber \\
&&
\end{eqnarray}
which is within $\alpha(\q_1)$ bits of the second term of
(\ref{R1_bound_twolinks}). As a result, the gap between (\ref{achievable_R1})
and (\ref{R1_bound_twolinks}) is bounded by
\begin{equation}
\delta_{R_1} = \max \left\{\alpha(\q_1), \beta(\q_1) \right\}. \label{delta_R1}
\end{equation}

Due to symmetry, (\ref{achievable_R2}) is within
\begin{equation}
\delta_{R_2} =\max \left\{\alpha(\q_2), \beta(\q_2) \right\}  \label{delta_R2}
\end{equation}
bits of the upper bound (\ref{R2_bound_twolinks}).

(ii)  Second, (\ref{achievable_R1R2_1})-(\ref{achievable_R1R2_3}) are within constant
    bits of their upper bounds
    (\ref{sumrate_bound_1_twolinks})-(\ref{sumrate_bound_12_twolinks}). To see
    this, inspecting the expressions of the achievable sum rates, it is easy to see that
    each of (\ref{achievable_R1R2_1})-(\ref{achievable_R1R2_3}) has four
    possible combinations: having both $\C_1$ and $\C_2$, having $\C_1$ only, having
    $\C_2$ only, and having none of $\C_1$ and $\C_2$. In the following, we
    show that, when specialized into the above four combinations, (\ref{achievable_R1R2_1})-(\ref{achievable_R1R2_3}) are
    within constant gap to the upper bounds
    (\ref{sumrate_bound_1_twolinks})-(\ref{sumrate_bound_12_twolinks}). The
    constant gaps are given by $\delta _{R_1+R_2}^{(\C_1, \C_2)}$, $\delta_{R_1+R_2}^{(\C_1,
    \mathsf{0})}$, $\delta_{R_1+R_2}^{(\mathsf{0}, \C_2)}$, and $\delta_{R_1+R_2}^{(\mathsf{0},
    \mathsf{0})}$ (to be defined later) respectively, each corresponding to a specific combination.

    \begin{itemize}
        \item First, when having  both $\C_1$  and $\C_2$, (\ref{achievable_R1R2_1})-(\ref{achievable_R1R2_3}) become
        \begin{eqnarray}
        R_1 + R_2 &\le& a_1 + g_2 + (\C_1-\xi_1)^+ + (\C_2-\xi_2)^+, \label{achievable_proof_R1R2_1} \nonumber \\
        && \\
        R_1 + R_2 &\le& a_2 + g_1 +  (\C_1-\xi_1)^+ + (\C_2-\xi_2)^+, \label{achievable_proof_R1R2_2} \nonumber \\
        && \\
        R_1 + R_2 &\le& e_1 + e_2 + (\C_1-\xi_1)^+ + (\C_2-\xi_2)^+, \label{achievable_proof_R1R2_3} \nonumber \\
        &&
        \end{eqnarray}
        which are within constant bits of (\ref{sumrate_bound_1_twolinks})-(\ref{sumrate_bound_3_twolinks}) respectively. To show this, first, according to Lemma~\ref{lemma_useful_ineq}, (\ref{achievable_proof_R1R2_1}) is lower bounded by
        \begin{eqnarray}
        \lefteqn{a_1 + g_2 + (\C_1-\xi_1)^+ + (\C_2-\xi_2)^+} \nonumber  \\
        && \ge \hf \log \left(1 + \frac{\SNR_1}{1 + \INR_1} \right) - \hf  \nonumber \\
        && + \hf \log(1 + \SNR_2 + \INR_1) -\hf \nonumber \\
        && +  \C_1 - \xi_1 + \C_2 - \xi_2,
        \end{eqnarray}
        which is within
        \begin{equation}
        \delta _{R_1+R_2}^{(\C_1, \C_2)} =\beta(\q_1) + \beta(\q_2)
        \end{equation}
        bits of the upper bound (\ref{sumrate_bound_1_twolinks}). Due to symmetry,
        (\ref{achievable_proof_R1R2_2}) is within $\delta _{R_1+R_2}^{(\C_1, \C_2)}$
        bits of the upper bound (\ref{sumrate_bound_2_twolinks}) as well. Now
        applying Lemma~\ref{lemma_useful_ineq}, (\ref{achievable_proof_R1R2_3}) is lower bounded by
        \begin{eqnarray}
        \lefteqn{e_1 + e_2 + (\C_1-\xi_1)^+ + (\C_2-\xi_2)^+} \nonumber \\
        &\ge& \hf \log \left(1 + \frac{\SNR_1}{1 + \INR_1} + \INR_2\right) -\hf \nonumber \\
        &&+ \hf \log \left(1 + \frac{\SNR_2}{1 + \INR_2} + \INR_1\right) -\hf \nonumber  \\
        && + \C_1 - \xi_1 + \C_2 - \xi_2,
        \end{eqnarray}
        which is within $\delta _{R_1+R_2}^{(\C_1, \C_2)}$ bits of the upper bound
        (\ref{sumrate_bound_3_twolinks}). Therefore, when specialized to the form
        with both $\C_1$ and $\C_2$ as shown in
        (\ref{achievable_proof_R1R2_1})-(\ref{achievable_proof_R1R2_3}),
        (\ref{achievable_R1R2_1})-(\ref{achievable_R1R2_3}) have a gap of $\delta _{R_1+R_2}^{(\C_1, \C_2)}$ bits
        to their upper bounds
        (\ref{sumrate_bound_1_twolinks})-(\ref{sumrate_bound_3_twolinks}).

        \item Second, when having $\C_1$ only, (\ref{achievable_R1R2_1})-(\ref{achievable_R1R2_3}) become
            \begin{eqnarray}
            R_1 + R_2 &\le& a_1 + g_2 + \Delta {g}_2 +(\C_1-\xi_1)^+, \label{achievable_proof_R1R2_4}\\
            R_1 + R_2 &\le& a_2 + \Delta {a}_2 + g_1 +  (\C_1-\xi_1)^+, \label{achievable_proof_R1R2_5}\\
            R_1 + R_2 &\le& e_1 + e_2 + \Delta {e}_2 + (\C_1-\xi_1)^+, \label{achievable_proof_R1R2_6}
            \end{eqnarray}
            where (\ref{achievable_proof_R1R2_4}) is lower bounded by
            \begin{eqnarray}
            \lefteqn{a_1 + g_2 + \Delta {g}_2 +(\C_1-\xi_1)^+}  \nonumber \\
            &\ge& \hf \log \left(1 + \frac{\SNR_1}{1 + \INR_1} \right) - \hf + \C_1 - \xi_1 \nonumber \\
            &&+ \hf \log \left( 1+ \SNR_2(1 + \phi_2^2 \SNR_{r1}) + \SNR_{r2} \right. \nonumber \\
            &&\qquad  \quad \left. + \INR_1 + \SNR_{r1}\right) - \alpha(\q_2),
            \end{eqnarray}
            which is within
            \begin{equation}
            \delta_{R_1+R_2}^{(\C_1, \mathsf{0})} = \alpha(\q_2) + \beta(\q_1)
            \end{equation}
            bits of the upper bound (\ref{sumrate_bound_4_twolinks}), and (\ref{achievable_proof_R1R2_5}) is lower bounded by
            \begin{eqnarray}
            \lefteqn{ a_2 + \Delta {a}_2 + g_1 +  (\C_1-\xi_1)^+} \nonumber \\
            &\ge& \hf \log \left(1 + \frac{\SNR_2 + \SNR_{r2}}{1 + \INR_2} \right) -\alpha(\q_2) \nonumber \\
            &&+ \hf \log(1 + \SNR_1 + \INR_2) - \hf + \C_1 + \xi_1, \nonumber \\
            &&
            \end{eqnarray}
            which is within $\delta_{R_1+R_2}^{(\C_1, \mathsf{0})}$ bits of the upper bound (\ref{sumrate_bound_5_twolinks}), and (\ref{achievable_proof_R1R2_6}) can be lower bounded by
            \begin{eqnarray}
            \lefteqn{e_1 + e_2 + \Delta {e}_2 + (\C_1-\xi_1)^+} \nonumber \\
            &\ge& \hf \log \left(1 + \frac{\SNR_1}{1 + \INR_1} + \INR_2\right) -\hf +  \C_1 -\xi_1 \nonumber \\
            &&+ \hf \log \left(1 + \frac{\SNR_2(1 + \phi_2^2 \SNR_{r1}) + \SNR_{r2}}{1 + \INR_2} \right. \nonumber \\
            &&  \qquad \qquad \left. + \INR_1 + \SNR_{r1} \right) - \alpha(\q_2),
            \end{eqnarray}
            which is within $\delta_{R_1+R_2}^{(\C_1, \mathsf{0})}$ bits of the upper bound (\ref{sumrate_bound_6_twolinks}).

        \item Third, when having $\C_2$ only, (\ref{achievable_R1R2_1})-(\ref{achievable_R1R2_3}) become
            \begin{eqnarray}
            R_1 + R_2 &\le& a_1 + \Delta {a}_1 + g_2  +(\C_2-\xi_2)^+, \label{achievable_proof_R1R2_7}\\
            R_1 + R_2 &\le& a_2 + g_1 +  \Delta {g}_1  + (\C_2-\xi_2)^+, \label{achievable_proof_R1R2_8}\\
            R_1 + R_2 &\le& e_1 + \Delta {e}_1 +e_2 + (\C_2-\xi_2)^+. \label{achievable_proof_R1R2_9}
            \end{eqnarray}
            Due to the symmetry between
            (\ref{achievable_proof_R1R2_7})-(\ref{achievable_proof_R1R2_9}) and
            (\ref{achievable_proof_R1R2_4})-(\ref{achievable_proof_R1R2_6}), and the symmetry between their upper bounds, we can see that (\ref{achievable_proof_R1R2_7}), (\ref{achievable_proof_R1R2_8}) and (\ref{achievable_proof_R1R2_9}) are
            within
            \begin{equation}
            \delta_{R_1+R_2}^{(\mathsf{0}, \C_2)} = \alpha(\q_1) + \beta(\q_2)
            \end{equation}
            bits of the upper bounds (\ref{sumrate_bound_7_twolinks}),
            (\ref{sumrate_bound_8_twolinks}), and
            (\ref{sumrate_bound_9_twolinks}) respectively.

        \item Fourth, when having none of $\C_1$ and $\C_2$, (\ref{achievable_R1R2_1})-(\ref{achievable_R1R2_3}) become
            \begin{eqnarray}
            R_1 + R_2 &\le& a_1 + \Delta {a}_1 + g_2 + \Delta {g}_2, \label{achievable_proof_R1R2_10}\\
            R_1 + R_2 &\le& a_2 + \Delta {a}_2 + g_1 +  \Delta {g}_1, \label{achievable_proof_R1R2_11}\\
            R_1 + R_2 &\le& e_1 + \Delta {e}_1 +e_2 + \Delta {e}_2, \label{achievable_proof_R1R2_12}
            \end{eqnarray}
            where (\ref{achievable_proof_R1R2_10}) is lower bounded by
            \begin{eqnarray}
            \lefteqn{a_1 + \Delta {a}_1 + g_2 + \Delta {g}_2} \nonumber \\
            &\ge& \hf \log \left(1 + \frac{\SNR_1 + \SNR_{r1}}{1 + \INR_1} \right) - \alpha(\q_1) \nonumber \\
            &&+ \hf \log \left( 1+ \SNR_2(1 + \phi_2^2 \SNR_{r1}) + \SNR_{r2}  \right. \nonumber \\
            && \qquad \quad \;\; \left. + \INR_1 + \SNR_{r1}\right) - \alpha(\q_2),
            \end{eqnarray}
            which is within
            \begin{equation}
            \delta_{R_1+R_2}^{(\mathsf{0},\mathsf{0})} = \alpha(\q_1) + \alpha(\q_2)
            \end{equation}
            bits of the upper bound (\ref{sumrate_bound_10_twolinks}). Due to symmetry, (\ref{achievable_proof_R1R2_11}) is within $\delta_{R_1+R_2}^{(\mathsf{0},\mathsf{0})}$ bits of the upper bound (\ref{sumrate_bound_11_twolinks}) as well. Further, (\ref{achievable_proof_R1R2_12}) can be lower bounded by
            \begin{eqnarray}
            \lefteqn{e_1 + \Delta {e}_1 +e_2 + \Delta {e}_2} \nonumber \\
            &\ge&  \hf \log \left(1 + \frac{\SNR_1(1 + \phi_1^2 \SNR_{r2}) + \SNR_{r1}}{1 + \INR_1} \right. \nonumber \\
            && \qquad \quad  \left. + \INR_2 + \SNR_{r2} \right) -\alpha(\q_1) \nonumber \\
            &&+  \hf \log \left(1 + \frac{\SNR_2(1 + \phi_2^2 \SNR_{r1}) + \SNR_{r2}}{1 + \INR_2}  \right. \nonumber \\
            && \qquad  \qquad \left. +\INR_1 + \SNR_{r1} \right) -\alpha(\q_2),
            \end{eqnarray}
            which is within
$\delta_{R_1+R_2}^{(\mathsf{0},\mathsf{0})}$ bits of the upper bound
(\ref{sumrate_bound_12_twolinks}). Therefore, when specialized into
the form with none of $\C_1$ and $\C_2$, (\ref{achievable_R1R2_1})-(\ref{achievable_R1R2_3}) is within $\delta_{R_1+R_2}^{(\mathsf{0},\mathsf{0})}$ bits of their upper bounds (\ref{sumrate_bound_10_twolinks})-(\ref{sumrate_bound_12_twolinks}).
        \end{itemize}

        Overall, the gap between the achievable sum-rates (\ref{achievable_R1R2_1})-(\ref{achievable_R1R2_3}) and the upper bounds in (\ref{sumrate_bound_1_twolinks})-(\ref{sumrate_bound_12_twolinks}) is upper bounded as follows:
        \begin{equation}
        \delta_{R_1+R_2} = \max \left\{\delta_{R_1+R_2}^{(\C_1, \C_2)}, \delta_{R_1+R_2}^{(\C_1, \mathsf{0})}, \delta_{R_1+R_2}^{(\mathsf{0}, \C_2)}, \delta_{R_1+R_2}^{(\mathsf{0}, \mathsf{0})}
        \right\}. \label{delta_Rsum}
        \end{equation}

    (iii) Third, the achievable rate (\ref{achievable_2R1R2}) is within constant bits of upper bounds (\ref{2R1R2_bound_1_twolinks})-(\ref{2R1R2_bound_6_twolinks}). To see this, note that (\ref{achievable_2R1R2}) has $8$ different forms as follows:
        \begin{eqnarray}
        a_1 + (\C_1 - \xi_1)^+ + g_1 + (\C_1 - \xi_1)^+ + e_2 + (\C_2 -\xi_2)^+, \label{achievable_2R1R2_1_proof} \\
        a_1 + \Delta {a}_1 + g_1 + \Delta {g}_1 + e_2 + \Delta {e}_2, \label{achievable_2R1R2_2_proof} \\
        a_1 + \Delta {a}_1 + g_1 + (\C_1 - \xi_1)^+ + e_2 + \Delta {e}_2, \label{achievable_2R1R2_3_proof} \\
        a_1 + \Delta {a}_1 +g_1 + \Delta {g}_1 + e_2 + (\C_2 -\xi_2)^+, \label{achievable_2R1R2_4_proof} \\
        a_1 + (\C_1 - \xi_1)^+ + g_1 + (\C_1 - \xi_1)^+ + e_2 + \Delta {e}_2, \label{achievable_2R1R2_5_proof} \\
        a_1 + \Delta {a}_1 + g_1 + (\C_1 - \xi_1)^+ + e_2 + (\C_2 -\xi_2)^+, \label{achievable_2R1R2_6_proof} \\
        a_1 + (\C_1 - \xi_1)^+ +g_1 + \Delta {g}_1 + e_2 + \Delta {e}_2, \label{achievable_2R1R2_7_proof} \\
        a_1 + (\C_1 - \xi_1)^+ + g_1 + \Delta {g}_1  + e_2 + (\C_2 -\xi_2)^+, \label{achievable_2R1R2_8_proof}
        \end{eqnarray}
        where (\ref{achievable_2R1R2_7_proof}) is redundant compared with
        (\ref{achievable_2R1R2_3_proof}) and (\ref{achievable_2R1R2_8_proof}) is
        redundant compared with (\ref{achievable_2R1R2_6_proof}) due to the fact that
        $\Delta {g}_1 \ge \Delta {a}_1$. Therefore, there are six
        active rate constraints in total. In the following, we prove that all active achievable rates
        of $2R_1 + R_2$ in
        (\ref{achievable_2R1R2_1_proof})-(\ref{achievable_2R1R2_6_proof}) are within
        constant bits of their corresponding upper bounds in
        (\ref{2R1R2_bound_1_twolinks})-(\ref{2R1R2_bound_6_twolinks}).
        \begin{itemize}
        \item First, (\ref{achievable_2R1R2_1_proof}) is lower bounded by
        \begin{eqnarray}
        \lefteqn{a_1 + (\C_1 - \xi_1)^+ + g_1 + (\C_1 - \xi_1)^+ + e_2 + (\C_2 -\xi_2)^+} \nonumber \\
        &\ge& \hf \log \left( 1 + \frac{\SNR_1}{1 + \INR_1}\right) - \hf + \C_1 - \xi_1 \nonumber \\
        &&+ \hf \log(1 + \SNR_1 + \INR_2) - \hf + \C_1 - \xi_1 \nonumber \\
        &&+ \hf \log \left(1 + \frac{\SNR_2}{1 + \INR_2} + \INR_1\right) -\hf + \C_2 - \xi_2, \nonumber \\
        &&
        \end{eqnarray}
        which is within
        \begin{equation}
        \delta_{2R_1+R_2}^{( 2\C_1, \C_2)} = 2\beta(\q_1) + \beta(\q_2)
        \end{equation}
        bits of the upper bound (\ref{2R1R2_bound_1_twolinks}).

        \item Second, (\ref{achievable_2R1R2_2_proof}) is lower bounded by
        \begin{eqnarray}
        \lefteqn{a_1 + \Delta {a}_1 + g_1 + g_1 + \Delta {g}_1 + e_2 + \Delta {e}_2} \nonumber \\
        &\ge& \hf \log \left(1 + \frac{\SNR_1 + \SNR_{r1}}{1 + \INR_1} \right) - \alpha(\q_1) \nonumber \\
        &&+ \hf \log \left( 1+ \SNR_1(1 + \phi_1^2 \SNR_{r2}) + \SNR_{r1} \right. \nonumber \\
            && \qquad \quad \;\; \left. + \INR_2 + \SNR_{r2}\right) - \alpha(\q_1) \nonumber \\
        &&+ \hf \log \left(1 + \frac{\SNR_2(1 + \phi_2^2 \SNR_{r1}) + \SNR_{r2}}{1 + \INR_2} \right. \nonumber \\
            && \qquad \qquad  \left. +\INR_1 + \SNR_{r1} \right) - \alpha(\q_2),
        \end{eqnarray}
        which is within
        \begin{equation}
        \delta_{2R_1+R_2}^{(\mathsf{0}, \mathsf{0})} = 2\alpha(\q_1) + \alpha(\q_2)
        \end{equation}
        bits of the upper bound (\ref{2R1R2_bound_2_twolinks}).

        \item Third, (\ref{achievable_2R1R2_3_proof}) is lower bounded by
        \begin{eqnarray}
        \lefteqn{a_1 + \Delta {a}_1 + g_1 + (\C_1 - \xi_1)^+ + e_2 + \Delta {e}_2} \nonumber \\
        &\ge& \hf \log \left(1 + \frac{\SNR_1 + \SNR_{r1}}{1 + \INR_1} \right) - \alpha(\q_1) \nonumber \\
        &&+ \hf \log(1 + \SNR_1 + \INR_2) - \hf + \C_1 - \xi_1 \nonumber \\
        &&+ \hf \log \left(1 + \frac{\SNR_2(1 + \phi_2^2 \SNR_{r1}) + \SNR_{r2}}{1 + \INR_2} \right. \nonumber \\
            && \left.  \qquad \qquad + \INR_1 + \SNR_{r1} \right) - \alpha(\q_2),
        \end{eqnarray}
        which is within
        \begin{equation}
        \delta_{2R_1+R_2}^{(\C_1, \mathsf{0})} = \alpha(\q_1) + \alpha(\q_2) + \beta(\q_1)
        \end{equation}
        bits of the upper bound (\ref{2R1R2_bound_3_twolinks}).

        \item Fourth, (\ref{achievable_2R1R2_4_proof}) is lower bounded by
        \begin{eqnarray}
        \lefteqn{a_1 + \Delta {a}_1 +g_1 + \Delta {g}_1 + e_2 + (\C_2 -\xi_2)^+} \nonumber \\
        &\ge& \hf \log \left(1 + \frac{\SNR_1 + \SNR_{r1}}{1 + \INR_1} \right) - \alpha(\q_1) \nonumber \\
        &&+ \hf \log \left( 1+ \SNR_1(1 + \phi_1^2 \SNR_{r2}) + \SNR_{r1} \right. \nonumber \\
            && \left.\qquad \quad \;\;+ \INR_2 + \SNR_{r2}\right) - \alpha(\q_1) \nonumber \\
        &&+ \hf \log \left(1 + \frac{\SNR_2}{1 + \INR_2} + \INR_1\right) -\hf + \C_2 - \xi_2 , \nonumber \\
        &&
        \end{eqnarray}
        which is within
        \begin{equation}
         \delta_{2R_1+R_2}^{(\mathsf{0}, \C_2)} = 2\alpha(\q_1) + \beta(\q_2)
        \end{equation}
        bits of the upper bound (\ref{2R1R2_bound_4_twolinks}).

        \item Fifth, (\ref{achievable_2R1R2_5_proof}) is lower bounded by
        \begin{eqnarray}
        \lefteqn{a_1 + (\C_1 - \xi_1)^+ + g_1 + (\C_1 - \xi_1)^+ + e_2 + \Delta \tilde{e}_2} \nonumber \\
        &\ge& \hf \log \left( 1 + \frac{\SNR_1}{1 + \INR_1}\right) - \hf + \C_1 - \xi_1 \nonumber \\
        &&+ \hf \log(1 + \SNR_1 + \INR_2) - \hf + \C_1 - \xi_1 \nonumber \\
        &&+ \hf \log \left(1 + \frac{\SNR_2(1 + \phi_2^2 \SNR_{r1}) + \SNR_{r2}}{1 + \INR_2} \right. \nonumber \\
        && \qquad \qquad \left. + \INR_1 + \SNR_{r1} \right) - \alpha(\q_2),
        \end{eqnarray}
        which is within
        \begin{equation}
        \delta_{2R_1+R_2}^{(2\C_1, \mathsf{0})} = \alpha(\q_2) + 2\beta(\q_1)
        \end{equation}
        bits of the upper bound (\ref{2R1R2_bound_5_twolinks}).

        \item Sixth, (\ref{achievable_2R1R2_6_proof}) is lower bounded by
        \begin{eqnarray}
        \lefteqn{a_1 + \Delta {a}_1 + g_1 + (\C_1 - \xi_1)^+ + e_2 + (\C_2 -\xi_2)^+} \nonumber \\
        &\ge& \hf \log \left(1 + \frac{\SNR_1 + \SNR_{r1}}{1 + \INR_1} \right) - \alpha(\q_1) \nonumber \\
        &&+ \hf \log(1 + \SNR_1 + \INR_2) - \hf + \C_1 - \xi_1 \nonumber \\
        &&+ \hf \log \left(1 + \frac{\SNR_2}{1 + \INR_2} + \INR_1\right)  \nonumber \\
        && -\hf + \C_2 - \xi_2 ,
        \end{eqnarray}
        which is within
        \begin{equation}
        \delta_{2R_1+R_2}^{(\C_1, \C_2)} = \alpha(\q_1) + \beta(\q_1) + \beta(\q_2)
        \end{equation}
        bits of the upper bound (\ref{2R1R2_bound_6_twolinks}).
        \end{itemize}

        Therefore, the gap between the achievable rate (\ref{achievable_2R1R2}) and the corresponding upper bounds (\ref{2R1R2_bound_1_twolinks})-(\ref{2R1R2_bound_6_twolinks}) is bounded by the following constant
        \begin{eqnarray}
        \delta_{2R_1+R_2}&=& \max\left\{\delta_{2R_1+R_2}^{( 2\C_1, \C_2)}, \delta_{2R_1+R_2}^{(\mathsf{0}, \mathsf{0})}, \delta_{2R_1+R_2}^{(\C_1, \mathsf{0})}, \delta_{2R_1+R_2}^{(\mathsf{0}, \C_2)}, \right. \nonumber \\
            && \left. \qquad \;\; \delta_{2R_1+R_2}^{(2\C_1, \mathsf{0})}, \delta_{2R_1+R_2}^{(\C_1,
        \C_2)}\right\}. \label{delta_2R1R2}
        \end{eqnarray}

        Due the the symmetry between (\ref{achievable_R12R2}) and (\ref{achievable_2R1R2}), and the symmetry between their corresponding upper bounds, it is easy to see that (\ref{achievable_R12R2}) is also within constant gap to the upper
        bounds. The constant gap $\delta_{R_1+2R_2}$ can be obtained by simply
        switching indices of $1$ and $2$ in $\delta_{2R_1 + R_2}$.

\bibliographystyle{IEEEtran}


\begin{IEEEbiographynophoto}{Lei Zhou}
(S'05) received the B.E. degree in Electronics Engineering
from Tsinghua University, Beijing, China,
in 2003 and M.A.Sc. degree in Electrical and Computer Engineering from the University of Toronto, ON, Canada, in 2008. During 2008-2009, he was with Nortel Networks, Ottawa, ON, Canada. He is currently pursuing the Ph.D. degree with the Department of Electrical and Computer Engineering, University of Toronto, Canada. His research interests include multiterminal information theory, wireless communications, and signal processing.

He is a recipient of the Shahid U.H. Qureshi Memorial Scholarship in 2011, the Alexander Graham Bell Canada Graduate Scholarship in 2011, and the Chinese government award for outstanding self-financed students abroad in 2012.

\end{IEEEbiographynophoto}

\begin{IEEEbiographynophoto}{Wei Yu}
(S'97-M'02-SM'08) received the B.A.Sc. degree in Computer Engineering and
Mathematics from the University of Waterloo, Waterloo, Ontario, Canada in 1997
and M.S. and Ph.D. degrees in Electrical Engineering from Stanford University,
Stanford, CA, in 1998 and 2002, respectively. Since 2002, he has been with the
Electrical and Computer Engineering Department at the University of Toronto,
Toronto, Ontario, Canada, where he is now Professor and holds a
Canada Research Chair in Information Theory and Wireless Communications. His
main research interests include multiuser information theory, optimization,
wireless communications and broadband access networks.

Prof. Wei Yu currently serves as an Associate Editor for
{\sc IEEE Transactions on
Information Theory}. He was an Editor for {\sc IEEE Transactions on Communications} (2009-2011), an Editor for {\sc IEEE Transactions on Wireless
Communications} (2004-2007), and a Guest Editor for a number of
special issues for the {\sc IEEE Journal on Selected Areas in
Communications} and the {\sc EURASIP Journal on Applied Signal Processing}.
He is member of the Signal Processing for Communications and Networking
Technical Committee of the IEEE Signal Processing Society.
He received the IEEE Signal Processing Society Best Paper Award in 2008,
the McCharles Prize for Early Career Research Distinction in 2008,
the Early Career Teaching Award from the Faculty of Applied Science
and Engineering, University of Toronto in 2007, and the Early Researcher
Award from Ontario in 2006.
\end{IEEEbiographynophoto}

\end{document}